\documentclass[]{siamltex}
\usepackage{algorithmic}
\usepackage{enumerate}
\usepackage{multirow}
\usepackage{psfrag}
\usepackage[crop=pdfcrop]{pstool}
\usepackage{tikz}
\usepackage{varwidth}
\usetikzlibrary{calc,trees,positioning,arrows,chains,shapes.geometric,%
    decorations.pathreplacing,decorations.pathmorphing,shapes,%
    matrix,shapes.symbols}

\usepackage{fullpage}
\usepackage[algo2e,linesnumbered,ruled]{algorithm2e}
\DontPrintSemicolon
\SetKwComment{tcm}{ }{}
\newcommand{\snapsNo}{\ensuremath{\mathcal X}}	
\newcommand{\Qfull}{\ensuremath{\Qmat_{\mathrm{full}}}}	
\newcommand{\leftSing}{\ensuremath{\mathbf{u}}}	
\newcommand{\snaps}[1]{\ensuremath{\snapsNo_{#1}}}	
\newcommand{\range}[1]{\ensuremath{\mathrm{range}\left(#1\right)}}	
\newcommand{\leftGappy}[1]{\ensuremath{\mathbf Y}_{#1}}	
\newcommand{\matApproxANo}{\ensuremath{\tilde \A}}	
\newcommand{\matApproxA}{\ensuremath{\matApproxANo\left(\param\right)}}	
\newcommand{\matApproxAOn}{\ensuremath{\matApproxANo\left(\paramOnline\right)}}	
\newcommand{\matApproxMNo}{\ensuremath{\tilde \M}}	
\newcommand{\matApproxM}{\ensuremath{\matApproxMNo\left(\param\right)}}	
\newcommand{\podArgsNo}{\ensuremath{{\mathbf W}}}	
\newcommand{\podArgs}[2]{\ensuremath{\podArgsNo\left(#1,#2\right)}} 
\newcommand{\podArgsSub}[1]{\ensuremath{\podArgsNo_{#1}}}
\newcommand{\Mbases}{\ensuremath{\underline \M}}	
\newcommand{\Mbasis}[1]{\ensuremath{\Mbases^{#1}}}	
\newcommand{\paramTraini}{\param^i}	
\newcommand{\paramTrain}{\{\paramTraini\}}	
\newcommand{\paramNom}{\ensuremath{\bar\param}}
\newcommand{\tTrain}{\ensuremath{\mathsf T_\mathrm{sample}}}	
\newcommand{\lastT}{\ensuremath{\mathsf T}}	
\newcommand{\timeDom}{\ensuremath{\left[0,\lastT\right]}}	
\newcommand{\paramOnline}{\ensuremath{\param^\star}}	
\newcommand{\x}{\ensuremath{\mathbf x}}	

\newcommand{\ext}{\ensuremath{}}
\newcommand{\fext} {\ensuremath{\mathbf f\ext}}
\newcommand{\fextApprox} {\ensuremath{\tilde \fext}}
\newcommand{\fextApproxparam} {\ensuremath{\fextApprox\left(\q,\dot \q, t;\param\right)}}
\newcommand{\tildefext} {\ensuremath{\fextApprox}}
\newcommand{\fextred} {\ensuremath{\mathbf f_r\ext}}
\newcommand{\fextredparam} {\ensuremath{\mathbf f_r\ext\left(\qred,\dot \qred, t;\param\right)}}
\newcommand{\fextredApprox} {\ensuremath{\tilde \fextred}}
\newcommand{\fextredApproxparam} {\ensuremath{\fextredApprox\ext\left(\qred,\dot \qred, t;\param\right)}}

\newcommand{\X} {\ensuremath{{\sparsePotEnNo}}}
\newcommand{\PsiMat} {\ensuremath{{\mathbf \Psi}}}
\newcommand{\q} {\ensuremath{{\mathbf q}}}
\newcommand{\p} {\ensuremath{{\mathbf p}}}
\newcommand{\pRed} {\ensuremath{{\mathbf p_r}}}
\newcommand{\vVec} {\ensuremath{{\mathbf v}}}
\newcommand{\w} {\ensuremath{{\mathbf w}}}
\newcommand{\Mi} {\ensuremath{{\M_i}}}
\newcommand{\qr} {\ensuremath{{\q_r}}}

\newcommand{\podstate} {\ensuremath{{\mathbf V}}}
\newcommand{\podVectorizedA} {\ensuremath{{\podArgsNo_{\vectorizeAsample}}}}
\newcommand{\sparse} {\ensuremath{{\mathbf U}}}
\newcommand{\sparsePotEn} {\ensuremath{{\sparse_V\left(\param\right)}}}
\newcommand{\sparsePotEnOn} {\ensuremath{{\sparse_V\left(\paramOnline\right)}}}
\newcommand{\sparsePotEnNo} {\ensuremath{{\sparse_V}}}
\newcommand{\sparseA} {\ensuremath{{\sparse_\A}}}
\newcommand{\Aone} {\ensuremath{{\A_1}}}
\newcommand{\Atwo} {\ensuremath{{\A_2}}}
\newcommand{\h} {\ensuremath{{h}}}
\newcommand{\hone} {\ensuremath{{\h_1}}}
\newcommand{\htwo} {\ensuremath{{\h_2}}}
\newcommand{\honeparam} {\ensuremath{{\hone\left(\param\right)}}}
\newcommand{\htwoparam} {\ensuremath{{\htwo\left(\param\right)}}}
\newcommand{\sparseMass} {\ensuremath{{\sparse_M}}}

\newcommand{\sampleMat} {\ensuremath{{\mathbf P}}}
\newcommand{\sampleMatT} {\ensuremath{{\sampleMat^T}}}
\newcommand{\sampleMatMin} {\ensuremath{{\mathbf P_1}}}
\newcommand{\sampleMatMinT} {\ensuremath{{\sampleMatMin^T}}}
\newcommand{\sampleMatVec} {\ensuremath{{\mathbf p}}}
\newcommand{\sampleMatVectorized} {\ensuremath{\bar\sampleMat}}
\newcommand{\sampleMatVectorizedVec} {\ensuremath{\bar{\mathbf p}}}
\newcommand{\sampleMatVectorizedT} {\ensuremath{\bar\sampleMat^T}}

\newcommand{\podf} {\ensuremath{\podArgsNo_{\f}}}

\newcommand{\podfn} {\ensuremath{\podArgsNo_{\f}}}

\newcommand{\metricRed} {\ensuremath{g_r}}
\newcommand{\metricRedApprox} {\ensuremath{\tilde g_r}}
\newcommand{\nf} {\ensuremath{n_{\f}}}

\newcommand{\nsample} {\ensuremath{m}}
\newcommand{\nTrain} {\ensuremath{p}}

\newcommand{\A}{\mathbf A}
\newcommand{\B}{\mathbf B}
\newcommand{\D}{\mathbf D}
\newcommand{\I}{\mathbf I}
\newcommand{\Lmat}{\mathbf L}

\newcommand{\Ared} {\ensuremath{{\podstate^T\A(\param)\podstate}}}

\newcommand{\AredApproxOn} {\ensuremath{{\tilde \A\left(\paramOnline\right)}}}
\newcommand{\Aparam} {\ensuremath{{\A\left(\param\right)}}}
\newcommand{\AparamArg}[1]{\ensuremath{{\A\left(#1\right)}}}
\newcommand{\Abases} {\ensuremath{{\underline \A}}}
\newcommand{\Abasis}[1] {\ensuremath{\Abases_#1}}
\newcommand{\Abasisi} {\ensuremath{{\Abases_i}}}
\newcommand{\vectorizeAbasis} {\ensuremath{{\underline {\mathbf a}}}}
\newcommand{\vectorizeAbasisi} {\ensuremath{{\underline {\mathbf a}^i}}}
\newcommand{\Coeff} {\ensuremath{\xi}}
\newcommand{\Acoeff} {\ensuremath{\xi_\A}}

\newcommand{\AcoeffParamOn} {\ensuremath{\Acoeff\left(\paramOnline\right)}}
\newcommand{\Acoeffi} {\ensuremath{{\Acoeff^i}}}
\newcommand{\AcoeffiParam} {\ensuremath{{\Acoeffi\left(\param\right)}}}

\newcommand{\Mcoeffi} {\ensuremath{\Coeff^i_{\M}}}

\newcommand{\nA} {\ensuremath{{n_{\A}}}}
\newcommand{\nM} {\ensuremath{{n_M}}}

\newcommand{\param} {\ensuremath{{\mu}}}

\newcommand{\M} {\ensuremath{{\mathbf M}}}
\newcommand{\Mparam} {\ensuremath{{\M\left(\param\right)}}}
\newcommand{\MEmptyparam} {\ensuremath{{\M}}}
\newcommand{\MparamOn} {\ensuremath{{\M\left(\paramOnline\right)}}}
\newcommand{\AparamOn} {\ensuremath{{\A\left(\paramOnline\right)}}}
\newcommand{\C} {\ensuremath{{\mathbf C}}}
\newcommand{\Cparam} {\ensuremath{{\C\left(\param\right)}}}
\newcommand{\rayleighM} {\ensuremath{\alpha}}
\newcommand{\rayleighK} {\ensuremath{\beta}}
\newcommand{\CEmptyparam} {\ensuremath{{\C}}}
\newcommand{\CparamOn} {\ensuremath{{\C\left(\paramOnline\right)}}}

\newcommand{\fextparam} {\ensuremath{\fext\left(\q,\dot \q, t;\param\right)}}
\newcommand{\fextparamRom} {\ensuremath{\fext\left(\initialQ + \podstate\qred,\podstate\dot \qred, t;\param\right)}}

\newcommand{\fextparamRomRef} {\ensuremath{\fext\left(\qRef + \podstate\qred,\podstate\dot \qred, t;\param\right)}}
\newcommand{\fextEmptyparam} {\ensuremath{\fext}}
\newcommand{\tildefextparam} {\ensuremath{\tildefext\left(\q,\dot \q, t;\param\right)}}
\newcommand{\fextparamOn} {\ensuremath{{\fext\left(\q,\dot \q, t;\paramOnline\right)}}}
\newcommand{\dissFunNo} {\ensuremath{\mathcal F}}
\newcommand{\dissFun}[1] {\ensuremath{{\dissFunNo\left(#1;\param\right)}}}

\newcommand{\dissFunRedNo} {\ensuremath{{\dissFunNo_r}}}
\newcommand{\dissFunRed}[1] {\ensuremath{{\dissFunNo_r\left(#1;\param\right)}}}
\newcommand{\dissFunRedApprox}[1] {\ensuremath{{\tilde \dissFunNo_r\left(#1;\param\right)}}}
\newcommand{\dissFunRedApproxNo} {\ensuremath{{\tilde \dissFunNo_r}}}

\newcommand{\potEn} {\ensuremath{V}}
\newcommand{\potEnParam} {\ensuremath{{\potEn\left(\q;\param\right)}}}
\newcommand{\potEnEmptyParam} {\ensuremath{{\potEn}}}
\newcommand{\potEnSpace} {\ensuremath{\mathsf \potEn}}
\newcommand{\potEnRed} {\ensuremath{\potEn_r}}
\newcommand{\potEnRedApprox} {\ensuremath{\tilde\potEn_r}}

\newcommand{\qred} {\ensuremath{{ \q_r}}}
\newcommand{\qredOnline} {\ensuremath{{ \q_r^\star}}}

\newcommand{\nstate} {\ensuremath{n}}
\newcommand{\Lred} {\ensuremath{L_r}}
\newcommand{\tildeLred} {\ensuremath{\tilde L_r}}
\newcommand{\Lredapprox} {\ensuremath{\tilde L_r}}

\newcommand{\vectorizeNo} {\ensuremath{v}}
\newcommand{\vectorize}[1] {\ensuremath{\vectorizeNo\left(#1\right)}}
\newcommand{\vectorizeInvNo} {\ensuremath{\vectorizeNo^{-1}}}
\newcommand{\vectorizeInv}[1] {\ensuremath{\vectorizeInvNo\left(#1\right)}}
\newcommand{\vectorizeAsample} {\ensuremath{\mathbf{a}}}
\newcommand{\vectorizeAsamplei} {\ensuremath{\vectorizeAsample^i}}
\newcommand{\vectorAparamCoeffNo} {\ensuremath{\mathbf z}}
\newcommand{\vectorAparamCoeff} {\ensuremath{\vectorAparamCoeffNo\left(\paramOnline\right)}}
\newcommand{\vectorAparamCoeffi} {\ensuremath{z^i\left(\paramOnline\right)}}
\newcommand{\vectorAparamCoeffiNo} {\ensuremath{z}}
\newcommand{\energyCrit} {\ensuremath{\eta}}
\newcommand{\snapmat} {\ensuremath{\mathbf X}}
\newcommand{\snapvec} {\ensuremath{\mathbf x}}
\newcommand{\nsnap} {\ensuremath{{n_\snapvec}}}
\newcommand{\approxlambda} {\ensuremath{{\tilde \lambda}}}
\newcommand{\approxeig} {\ensuremath{{\tilde {\mathbf y}}}}

\newcommand{\TransformedsparsePotEn} {\ensuremath{\widetilde \sparsePotEnNo}}
\newcommand{\sparsePotEnSmall} {\ensuremath{\underline \sparsePotEnNo\left(\param\right)}}
\newcommand{\sparsePotEnSmallOn} {\ensuremath{\underline \sparsePotEnNo\left(\paramOnline\right)}}
\newcommand{\sparsePotEnSmallNo} {\ensuremath{\underline \sparsePotEnNo}}
\newcommand{\cholSample } {\ensuremath{\mathbf L_1}}
\newcommand{\cholHess } {\ensuremath{\mathbf L_2}}
\newcommand{\cholSol } {\ensuremath{\mathbf X}}
\newcommand{\sparseASmall} {\ensuremath{\underline \sparseA}}
\newcommand{\sparseASmallT} {\ensuremath{\underline \sparseA^T}}

\newcommand{\sparsePotEnSmallVar} {\ensuremath{\mathbf X}}

\newcommand{\LambdaMat} {\ensuremath{\boldsymbol{\Lambda}}}

\newcommand{\entrytuple}[2]{\left(#1^{1}, \ldots, #1^{#2}\right)}
\newcommand{\spd}[1]{\ensuremath{\mathrm{SPD}\left(#1\right)}}

\newcommand{\cvector}{\ensuremath{\mathbf{c}}}
\newcommand{\evector}{\ensuremath{\mathbf{e}}}

\newcommand{\nenergy}{\ensuremath{n_e}}
\newcommand{\paramDomain}{\ensuremath{\mathcal D}}
\newcommand{\half}{\ensuremath{\frac{1}{2}}}

\newcommand{\Q}{\ensuremath{Q}}
\newcommand{\U}{\ensuremath{\mathbf{U}}}
\newcommand{\V}{\ensuremath{\mathbf{V}}}
\newcommand{\Qmat}{\ensuremath{\mathbf Q}}
\newcommand{\Rmat}{\ensuremath{\mathbf R}}
\newcommand{\sparsity}{\ensuremath{\omega}}
\newcommand{\Qred}{\ensuremath{Q_r}}

\newcommand{\Qredfull}{\ensuremath{\mathsf Q_r}}
\newcommand{\qRef}{\ensuremath{\bar\q(\param)}}	
\newcommand{\qRefOn}{\ensuremath{\bar\q\left(\paramOnline\right)}}	
\newcommand{\qRefNo}{\ensuremath{\bar\q}}	
\newcommand{\initialQ}{\ensuremath{\q_0\left(\param\right)}}	
\newcommand{\initialQOn}{\ensuremath{\q_0\left(\paramOnline\right)}}	
\newcommand{\initialQNo}{\ensuremath{\q_0}}	

\newcommand{\f}{{\boldsymbol \theta}}
\newcommand{\fEmptyparam}{\f}

\newtheorem{lem}{Lemma}
\tikzset{
>=stealth',
  punktchain/.style={
    rectangle, 
    rounded corners, 
    draw=black, very thick,
    text width=10em, 
    minimum height=3em, 
    text centered, 
    on chain},
  graychain/.style={
    rectangle, 
    rounded corners, 
    draw=gray, very thick,
    text width=10em, 
    minimum height=3em, 
    text centered, 
    on chain},
  approxMeth/.style={
    rectangle, 
    rounded corners, 
    draw=none, thin,
    text width=10em, 
    minimum height=3em, 
    text centered, 
    on chain},
  line/.style={draw, thick, <-},
  element/.style={
    tape,
    top color=white,
    bottom color=blue!50!black!60!,
    minimum width=8em,
    draw=blue!40!black!90, very thick,
    text width=10em, 
    minimum height=3.5em, 
    text centered, 
    on chain},
  every join/.style={->, thick,shorten >=1pt},
  decoration={brace},
  tuborg/.style={decorate},
  tubnode/.style={midway, right=2pt},
}
\newcommand{\myspace}{\hskip .025in}

  \newcommand{\IO}{[\I~~~\mathbf{0}]}
  \newcommand{\IOT}{\IO^T}  
\newcommand{\REMOVE}[1]{}
\newcommand{\bordered}{\ensuremath{\mathrm{bordered}}}
\newcommand{\forceMag}{\ensuremath{\gamma}}
\newcommand{\forcing}{\ensuremath{\fext}}	
\newcommand{\forceDist}{\ensuremath{{\mathbf r}}}	
\newcommand{\forceTime}{\ensuremath{{r}}}	

\newcommand{\initialCondition}{\ensuremath{{\q(0;\param)}}}	
\newcommand{\initialCond}{\ensuremath{{\mathbf s}}}	
\newcommand{\initialCondMag}{\ensuremath{{ s}}}	
\newcommand{\fNomi}{\ensuremath{\underline f_i}}	
\newcommand{\fNomnum}[1]{\ensuremath{\underline f_{#1}}}	
\newcommand{\y}{\ensuremath{\mathbf{y}}}	

\newcounter{remctr}
\newenvironment{remark}{\refstepcounter{remctr}\begin{trivlist}
\item {\emph {Remark}.\:}}
{\end{trivlist}}
\newcounter{problemctr}
\newcommand{\RR}[1]{\ensuremath{\mathbb{R}^{ #1 }}}

\newcommand{\vecmat}[2]{\left[#1^1 \ \cdots\ #1^{#2}\right]}

\usepackage{subfigure, algorithmic, algorithm}
\usepackage{amsmath,amssymb}
\usepackage{graphics,epsfig}
\usepackage{wrapfig}
\usepackage{subfigmat}
\usepackage{amsopn}

\definecolor{Red}{rgb}{1,0,0}

\definecolor{Green}{rgb}{.2,.5,.3}

\definecolor{Blue}{rgb}{0,0,1}

\title{Preserving Lagrangian structure in\\ nonlinear model reduction with \\
application to structural dynamics}
\author{Kevin Carlberg\thanks{Harry S.\ Truman Fellow, Quantitative Modeling
\& Analysis Department
\texttt{ktcarlb@sandia.gov}}
\and Ray Tuminaro\thanks{Numerical
Analysis and Applications Department, \texttt{rstumin@sandia.gov}}
\and
Paul
Boggs\thanks{Quantitative Modeling \&
Analysis Department (retired), \texttt{ptboggs@sandia.gov}}%
}

\begin{document}
\maketitle

\SetAlgorithmName{Procedure}{}

\begin{abstract}
This work proposes a model-reduction methodology that preserves Lagrangian
structure (equivalently Hamiltonian structure) and achieves computational
efficiency in the presence of high-order nonlinearities and arbitrary
parameter dependence. As such, the resulting reduced-order model retains key
properties such as energy conservation and symplectic time-evolution maps.  We
focus on parameterized simple mechanical systems subjected to Rayleigh damping
and external forces, and consider an application to nonlinear structural
dynamics. To preserve structure, the method first approximates the system's
`Lagrangian ingredients'---the Riemannian metric, the potential-energy
function, the dissipation function, and the external force---and subsequently
derives reduced-order equations of motion by applying the (forced)
Euler--Lagrange equation with these quantities.  From the algebraic
perspective, key contributions include two efficient techniques for
approximating parameterized reduced matrices while preserving symmetry and
positive definiteness: matrix gappy POD and reduced-basis sparsification
(RBS).  Results for a parameterized truss-structure problem demonstrate the
importance of preserving Lagrangian structure and illustrate the proposed
method's merits: it reduces computation time while maintaining high accuracy
and stability, in contrast to existing nonlinear model-reduction techniques
that do not preserve structure.  
\end{abstract}

\begin{keywords}
nonlinear model reduction,  structure preservation, Lagrangian dynamics,
Hamiltonian dynamics, structural dynamics, positive definiteness, matrix
symmetry
\end{keywords}

\section{Introduction}\label{sec:intro}

Computational modeling and simulation for simple mechanical systems characterized by a Lagrangian
formalism has become indispensable across a variety of industries. Such
simulations enable the understanding of complex systems, reduced design costs,
and improved reliability for a wide range of applications. For example,
computational structural dynamics tools have become widely used in
applications ranging from aerospace to biomedical-device design;
molecular-dynamics simulations have gained popularity in materials science and
biology, in particular.  However, the high computational cost incurred by
simulating large-scale simple mechanical systems can result in simulation
times on the order of weeks, even when using high-performance computers.  As a
result, these simulation tools are impractical for time-critical applications
that demand the accuracy provided by large-scale, high-fidelity models. In
particular, applications such as nondestructive evaluation for structural
health monitoring, multiscale modeling, embedded control, design optimization,
and uncertainty quantification require highly accurate results to be obtained
quickly.

In this work, we consider models that depend on a set of parameters, e.g.,
design variables, operating conditions.  In this context, model-reduction
methods present a promising approach for addressing time-critical problems.
During the \emph{offline stage}, these methods perform computationally
expensive `training' tasks, which may include evaluating the high-fidelity model for
several instances of the system parameters  and computing a low-dimensional
subspace for the solution. Then, during the inexpensive \emph{online stage},
these methods quickly compute approximate solutions for arbitrary values of the system
parameters. To accomplish this, they reduce the dimension of the high-fidelity
model by restricting solutions to lie in the low-dimensional subspace that was
computed offline; they also introduce other approximations when nonlinearities are
present.  Thus, the  reduced-order model used online is characterized
by a low-dimensional dynamical system that arises from a projection process on
the high-fidelity-model equations.  This offline/online strategy is effective
primarily in two scenarios: `many query' problems (e.g., Bayesian inference),
where the high offline cost is amortized over many online evaluations, and
real-time problems (e.g., control) characterized by stringent constraints on
online evaluation time.

Generating a reduced-order model that preserves the Lagrangian
structure intrinsic to mechanical systems is not a trivial task. Such
structure is critical to preserve, as it leads to fundamental properties such
as energy conservation (in the absence of non-conservative forces),
conservation of quantities associated with symmetries in the system, and
symplectic time-evolution maps. In fact, the class of structure-preserving
time integrators (e.g., geometric integrators \cite{hairer2006geometric},
variational integrators \cite{marsden2001discrete}) has been developed to
ensure that the discrete solution to the high-fidelity computational model
associates with the time-evolution map of a (modified) Lagrangian system.

Lall et al.~\cite{lall2003structure} show that performing a Galerkin
projection on the Euler--Lagrange equation---as opposed to the first-order
state-space form---leads to a reduced-order model that preserves Lagrangian
structure.  However, the computational cost of assembling the associated
low-dimensional equations of motion scales with the dimension of the
high-fidelity model. For this reason, this approach is efficient only when
the low-dimensional operators can be assembled \textit{a
priori}; this occurs only in very limited cases 
 e.g., when operators have a low-order polynomial dependence on the state and
 are affine in functions of the parameters
 \cite{maday2002reliable}.

Several methods have been developed in the context of nonlinear-ODE model
reduction that can reduce the computational cost of assembling the
low-dimensional equations of motion. However, these methods destroy Lagrangian
structure when applied to simple mechanical systems. For example, collocation
approaches \cite{astrid2007mpe,ryckelynck2005phm} perform a Galerkin
projection on only a small subset of the full-order equations characterizing
the high-fidelity model. Although this method works well for some nonlinear
ODEs, it destroys Lagrangian structure. The discrete empirical interpolation
method
(DEIM) \cite{chaturantabut2010journal,galbally2009non,drohmannEOI} and gappy
proper orthogonal decomposition
(POD) reconstruction methods \cite{sirovichOrigGappy,CarlbergGappy,carlbergJCP}
compute a few entries of the vector-valued nonlinear functions, and then
approximate the uncomputed entries by interpolation or least-squares
regression using an empirically derived basis. Galerkin projection can then be
performed with the approximated nonlinear function. Again, this technique
destroys Lagrangian structure. 

The goal of this work is to devise a reduced-order model for nonlinear simple
mechanical systems with general parameter dependence that leads to
computationally inexpensive online solutions and preserves Lagrangian
structure. We focus particularly on parameterized structural-dynamics models
under Rayleigh damping and external forces. The methodology we propose
constructs a reduced-order model by first approximating the `Lagrangian
ingredients' (i.e., quantities defining the problem's Lagrangian structure)
and subsequently deriving the equations of motion by applying the
Euler--Lagrange equation to these ingredients. The method approximates the
Lagrangian ingredients as follows:
 \begin{enumerate}[I.]
 \item \label{prop:config}\emph{Configuration space}. The
 low-dimensional configuration space is derived using standard
 dimension-reduction techniques, e.g., proper orthogonal decomposition, modal
 decomposition.
  \item \label{prop:metric}\emph{Riemannian metric}. The Riemannian metric is defined by a
	low-dimensional symmetric positive-definite matrix. We propose two
	efficient methods for approximating this
	low-dimensional matrix that preserve symmetry and positive definiteness.
	\item \label{prop:potential}\emph{Potential-energy function}. The potential energy function is
	approximated by employing the original potential-energy function, but with
	the low-dimensional reduced-basis matrix replaced by a low-dimensional
	\emph{sparse} matrix with only a few nonzero
	rows. This sparse matrix is computed online by matching the gradient of the
	reduced potential to first order about the equilibrium configuration.
	\item \label{prop:diss}\emph{Dissipation function}. The
	damping matrix associated with the Rayleigh dissipation function is a linear
	combination of the mass matrix (which defines the Riemannian metric) and the Hessian
	of the potential. Thus, we form the approximated Rayleigh dissipation
	function in the same fashion, but employ the approximated mass matrix
	from ingredient \ref{prop:metric} and
	approximated potential from ingredient \ref{prop:potential}.
	\item \label{prop:ext}\emph{External force}. The external force is derived by applying the
	Lagrange--D'Alembert principle with variations in the configuration space.
	We approximate this by applying gappy POD
	reconstruction to the external force as expressed in the original
	coordinates.  As a result, the external force appearing in the reduced-order
	equations of motion can be derived by applying the Lagrange--D'Alembert principle
	to this modified external force with variations restricted to the
	low-order configuration space.
	 \end{enumerate}
We note that a structure-preserving method \cite{structurePreserveBeattie} has been recently proposed for
nonlinear port-Hamiltonian systems, which are generalizations of Hamiltonian
systems. While this technique guarantees that properties such as stability and
passivity are preserved, it does not in fact preserve Lagrangian or classical
Hamiltonian structure. In particular, the resulting reduced-order equations of motion
cannot be derived from approximated ingredients such as those enumerated above; as
a consequence, the approach does not ensure symplecticity or energy conservation
for conservative systems, for example.

As hinted above, preserving structure for Lagrangian ingredient
\ref{prop:metric} is equivalent to
efficiently approximating a low-dimensional reduced matrix while preserving
symmetry and positive definiteness. This algebraic task is relevant to a
broad 
scope of applications, 
e.g.,
approximating extreme eigenvalues/eigenvectors of a parameterized
matrix, preserving Hessian positive definiteness in optimization algorithms. For this reason, Section \ref{sec:matrixApprox} presents 
approximation techniques for Lagrangian ingredient
\ref{prop:metric}
in a stand-alone algebraic setting that does
not rely on the Lagrangian formalism.
Similarly, Section \ref{sec:potEnGen} considers 
Lagrangian ingredient \ref{prop:potential} 
in a purely alegraic
context 
that does
not depend on Lagrangian dynamics.

The remainder of the paper is organized as follows. Section \ref{sec:prob}
introduces the Lagrangian-mechanics formulation. Section \ref{sec:existMR}
outlines existing model-reduction techniques and highlights the need for an
efficient, structure-preserving method. Section \ref{sec:newMethod} presents
the proposed method.  Section \ref{sec:experiments} presents numerical
experiments applied to a simple mechanical system from structural dynamics.
Finally, Section \ref{sec:conclusions} 
summarizes the contributions and suggests further research.

The structure-preserving model-reduction methods proposed in this work also
preserve Hamiltonian structure when the Hamiltonian formulation of classical
mechanics is used. Appendix \ref{app:Hamiltonian} provides this connection.

\section{Preserving matrix symmetry and positive definiteness}\label{sec:matrixApprox}

This section presents approximation techniques for Lagrangian ingredient
\ref{prop:metric} 
in an algebraic setting.
First, to establish notation, denote the system parameters by
$\param\in\paramDomain$, where $\paramDomain$ represents the parameter domain.
Let $\Aparam$ denote an $N\times N$  parameterized symmetric
positive-definite (possibly dense) matrix. Finally, let $\podstate$ denote a
dense, parameter-independent, full-column-rank $N\times\nstate$ matrix with
$\nstate\ll N$ whose columns can be interpreted as a \emph{reduced basis}
spanning an $\nstate$-dimensional subspace of $\RR{N}$.  We consider the
following \emph{online problem}:
\vskip .1in
\begin{list}
{(P\arabic{problemctr})}{\usecounter{problemctr}\setlength{\rightmargin}{\leftmargin}}
\item \label{probMat}
 At a cost independent of $N$, compute a symmetric positive-definite matrix
 $\AredApproxOn$ that is appropriately close to the matrix
 $\podstate^T\AparamOn\podstate$ for any specified online point
 $\paramOnline\in\paramDomain$.
\refstepcounter{problemctr}
\end{list}
\vskip .1in
Note that due to the density of $\podstate$, directly computing
$\podstate^T\AparamOn\podstate$ requires computing all $\mathcal O(N)$ entries of the matrix
$\AparamOn$. 
Recall that the offline/online strategy we adopt permits expensive offline
operations that facilitate the solution to online problem (P\ref{probMat}).
These operations may include collecting $\nTrain$ `snapshots'
of the matrix $ \AparamArg{\paramTraini}$, $i=1,\ldots,\nTrain$, where
$\paramTraini\in\paramDomain$ denotes the $i$th instance of the training set.

We assume that computing a single entry of $\AparamOn$ for any specified
online point $\paramOnline\in\paramDomain$ is inexpensive, i.e., the number of
floating-point operations (flops) is independent of $N$. However, we make no
other assumptions regarding the parameters or the matrix. In particular, we do
not assume affine parametric dependence of the matrix, and we view
$\param\mapsto \Aparam$ simply as a mechanism for generating instances of the matrix
$\A$. 

We now present two methods for solving online problem (P\ref{probMat}).
Method 1 approximates the reduced matrix by projecting the full matrix
onto a sparse basis, while Method 2 approximates the reduced matrix as a
linear combination of pre-computed reduced matrices.
Later, Section \ref{sec:existingMethods} constrasts the proposed methods with
existing model reduction approaches such as DEIM, gappy POD, and collocation.
These existing methods apply a one-sided sampling operator to the
matrix, i.e., they replace $\podstate^T$ with some type of sparse matrix.
While this leads to an inexpensive approximation, it
gives rise to a non-symmetric reduced matrix approximation. This destroys the
underlying problem structure and therefore fails to meet the requirements of
online problem (P\ref{probMat}).

\subsection{Reduced-basis sparsification (RBS)}\label{sec:sparseMat}
We first consider a strategy that
`injects sparseness' into the matrix $\podstate$.
That is, we replace $\podstate$ 
by $\sparseA\in\RR{N\times\nstate}$, which has full column rank and only
$\nsample$ rows  (with $\nstate\leq \nsample\ll N$)
containing
nonzero entries: 
\begin{align} 
\label{eq:spApproxSparse}\matApproxA&=
\sparseA^T\Aparam\sparseA.
  \end{align} 
This sparse matrix may be expressed as $\sparseA\equiv\sampleMat
\sparseASmall$, where $\sampleMat\in\{0,1\}^{N\times \nsample}$ is a `sampling
matrix' consisting of $\nsample$ selected columns of the $N\times N$ identity
matrix, $\sparseASmall\in\RR{\nsample\times \nstate}_*$ is a dense matrix
with full column rank $\nstate$, and
$\RR{\nsample\times \nstate}_*$ denotes the
noncompact Stiefel manifold: the set of full-rank $\nsample\times \nstate$
matrices.  Clearly, $\matApproxA$ is symmetric positive definite if $\Aparam$
is symmetric positive definite; thus, the approximation defined by
\eqref{eq:spApproxSparse} preserves the requisite structure. Note that this approximation will also
preserve structure in cases where $\Aparam$ is symmetric positive
semidefinite or simply symmetric. 
Further,
the (online) operation count for computing $\matApproxAOn$ for online point
$\paramOnline\in\paramDomain$ is independent of $N$.
Computing $\sampleMatT\AparamOn\sampleMat $ is equivalent to computing only
$\nsample^2$ (symmetric) entries of $\AparamOn$ and entails $\mathcal
O(\nsample^2)$ flops; subsequently computing $\AredApproxOn =
 \sparseASmall^T\left[\sampleMatT\AparamOn\sampleMat \right]\sparseASmall$ requires $\mathcal
O(\nsample^2 \nstate + \nsample
 \nstate^2)$ flops.

Given a sampling matrix $\sampleMat$, the matrix $\sparseASmall$ can be
computed offline to minimize the average approximation error over the
snapshots, i.e., according to the following optimization problem:
 \begin{equation} \label{eq:optMatchMat}
 \sparseASmall =\arg\min_{\sparsePotEnSmallVar\in\RR{\nsample\times \nstate}_*}\sum_{i=1}^{\nTrain}\bigl\|
\sparsePotEnSmallVar^T\sampleMatT\AparamArg{\paramTraini}\sampleMat \sparsePotEnSmallVar - 
\podstate^T\AparamArg{\paramTraini}\podstate\bigr\|_F^2,
  \end{equation} 
where
the subscript $F$ denotes the Frobenius
norm.  
To handle the
fact that $\RR{\nsample\times \nstate}_*$ is an open set, optimization problem \eqref{eq:optMatchMat} can
first be solved over $\RR{\nsample\times \nstate}$ and the solution can be
subsequently projected onto $\RR{\nsample\times \nstate}_*$, which is
analogous to the approach taken by Vandereycken \cite[Algorithm 6]{vandereyckenMatCom}.
We note that other objective functions may be considered for specialized
online analyses, e.g., the approximation error of
extreme eigenvalues or eigenvectors over the matrix snapshots.
Note that problem \eqref{eq:optMatchMat} is a small-scale optimization
problem, as $\nsample,\nstate\ll N$.  It can be solved at a cost independent
of $N$ (for each optimization iteration) during the offline stage after the matrix snapshots
$\AparamArg{\paramTraini}$, $i=1,\ldots,\nTrain$ and their
reduced counterparts $\podstate^T\AparamArg{\paramTraini}\podstate$,
$i=1,\ldots,\nTrain$ have been computed.

Procedure \ref{procedure:rbs} provides the offline and online steps required
to implement the RBS approximation.
\begin{algorithm2e}\label{procedure:rbs}
\caption{\textsf{Reduced-basis sparsification for symmetric matrices}}
\BlankLine
\tcm{\textbf{\textsf{Offline stage}}}
\BlankLine
Collect matrix snapshots $\AparamArg{\paramTraini}$, $i=1,\ldots,\nTrain$. \;
Form reduced the matrices
	$\podstate^T\AparamArg{\paramTraini}\podstate$,
	$i=1,\ldots,\nTrain$.\;
Choose the sample matrix $\sampleMat$.\;
Determine $\sparseASmall$ as the solution to problem
	\eqref{eq:optMatchMat}.\;
\setcounter{AlgoLine}{0}
\BlankLine
\tcm{\textbf{\textsf{Online stage}} \textsf{(given $\paramOnline$)}}
\BlankLine
Compute $\sampleMatT\AparamOn\sampleMat $.\;\label{step:sampleA} 
Form $\AredApproxOn =
 \sparseASmall^T\left[\sampleMatT\AparamOn\sampleMat \right]\sparseASmall$.\;
\label{step:projectSampledA} 
\end{algorithm2e}


\begin{remark}
This paper does not focus on methods for selecting the sampling matrix
$\sampleMat$, a task that is typically carried out during the offline stage
and uses
the collected snapshots.  All
numerical experiments presented in Section \ref{sec:experiments} use the GNAT
greedy approach~\cite{CarlbergGappy} for this purpose.  This method has
proven to be quite robust, even when applied to the approximation techniques
proposed in this paper. A more careful study of sampling algorithms will be
addressed in future work, where ideas of tailoring the sampling approach to
the specific reduced-order modeling approximation will be explored.
\end{remark}
\subsubsection{Exactness conditions}\label{sec:exactSparseQuad}
In the full-sampling case where $\nsample = N$, the approximation is exact if
problem \eqref{eq:optMatchMat} is solved via a gradient-based method with an
initial guess of $\sparsePotEnSmallVar^{(0)} = \sampleMatT \podstate$; we do
this in practice. 
Under these conditions, $\sparseA = \podstate$ and so $\matApproxA =
\podstate^T\Aparam\podstate$.

In the general case where $\nsample < N$, it is possible to show that 
the approximation is exact if the matrix is parameter-independent (i.e.,
$\Aparam = \A$) and $\nsample \ge \nstate$.  This situation is considered
in the discussion that follows Theorem~\ref{thm:quadCase} below.
It is also possible to prove a more general exactness result in cases where
the parametric dependence of $\Aparam$ is relatively simple. In particular,
consider the parametric form
\begin{equation}\label{eq:aAffine}
\Aparam = \honeparam~\Aone + \htwoparam~\Atwo,
\end{equation}
where $\Aone$ and $\Atwo$ are $N\times N$ symmetric
positive-definite matrices and $\hone,\htwo: \paramDomain \rightarrow \RR{}$.
It can then be shown that a sparse reduced basis exists that exactly 
captures $\podstate^T\Aparam\podstate$ under conditions related to how 
well the eigenvalues of the sampled matrix
$ \sampleMatT \A (\param) \sampleMat $
{\it encompass} (or surround) those of the reduced matrix $ \podstate^T \A (\param) \podstate $.
Loosely stated, the encompassing conditions amount to how well
the sampled matrix captures the behavior of the reduced matrix.
Formally, the following theorem  makes precise the notion of encompassing
using eigenvalue interlacing ideas from classical linear algebra.
\begin{theorem}\label{thm:quadCase}
Let $\Aparam $ have the form given by Eq.~\eqref{eq:aAffine}.
\vskip .1in
\noindent
Then, 
\begin{equation} \label{quadratic exact}
\exists~\sparseASmall\in\RR{\nsample\times\nstate} ~~\mbox{such that}~~~
\sparseASmallT\myspace\sampleMatT\myspace\Aparam\myspace\sampleMat\myspace\sparseASmall
= \podstate^T \myspace\Aparam\myspace\podstate,\quad \forall \param \in \paramDomain
\end{equation}
if and only if the generalized eigenvalues of $(\podstate^T \Atwo
	\podstate,\podstate^T \Aone \podstate)$ interlace the generalized
	eigenvalues of $(\sampleMatT \Atwo \sampleMat, \sampleMatT \Aone
	\sampleMat)$, i.e., 
\begin{equation} \label{new interlace}
\lambda^{(s)}_i \le \lambda^{(r)}_i \le \lambda^{(s)}_{i+\nsample -\nstate } 
~~~for~i=1,...,\nstate,
\end{equation} 
with
$$
\left [ \podstate^T \Atwo \podstate \right ] \x_i^{(r)} = 
\lambda_i^{(r)} \left [ \podstate^T \Aone \podstate \right ] \x_i^{(r)}, \quad
i=1,\ldots \nstate 
$$
and
$$
\left [ \sampleMatT \Atwo \sampleMat \right ] \x_i^{(s)} = 
\lambda_i^{(s)} \left [ \sampleMatT \Aone \sampleMat \right ] \x_i^{(s)},\quad
i=1,\ldots,\nsample.
$$
Note that the eigenvalues are sorted in order of increasing magnitude.
\end{theorem}
Appendix \ref{app:quadCase} contains the proof. Here, we
discuss the theorem's implications.

When $\A$ is independent of $\param$, we can choose $\hone = \htwo = 1$
and
$\Aone = \Atwo$. The interlacing property is
then trivially satisfied for $\nsample  = \nstate $ with 
$\lambda^{(s)}_i = \lambda^{(r)}_i = 1$, and so the equality in \eqref{quadratic exact}
always holds.  When instead
$\Aone \ne \Atwo$ and 
$\nsample  \hskip -.04in = \hskip -.04in \nstate  \hskip -.04in + \hskip
-.04in 1$, the interlacing definition 
is quite restrictive, as it implies that
$\lambda^{(s)}_k \le \lambda^{(r)}_k \le \lambda^{(s)}_{k+1}$.
We would not generally expect the eigenvalues of the sampled and reduced
matrices to have this property. However, when $\nsample  \gg \nstate +1$, each
interval width is (much) larger and so the condition is not nearly as
restrictive. For example, if $\nstate  = 100$ and $\nsample  = 300$, then interlacing
implies that
$\lambda^{(s)}_i \le \lambda^{(r)}_i \le \lambda^{(s)}_{i+200}$.
Thus, we generally expect the conditions of the theorem to be 
satisfied for sufficiently large
$\nsample $, though this is not
guaranteed and depends on matrix spectrum.
As a final note, although the theorem assumes a
specific form of $\Aparam$, it should characterize the rough
behavior of a more general $\Aparam$ that does not vary `too much' from the
affine functional form \eqref{eq:aAffine}.

\subsection{Matrix gappy POD}\label{sec:lsrecon}
An alternative structure-preserving approximation applicable to problem
(P\ref{probMat}) assumes the following form:
 \begin{align} 
\label{eq:spApprox}\matApproxA&=\sum\limits_{i=1}^\nA
\AcoeffiParam\podstate^T\Abasis{i}\podstate .
  \end{align} 
Here, the matrices $\Abasis{i}$, $i=1,\ldots,\nA$ are $N\times N$ symmetric
matrices that are computed offline and define a basis for the
matrix $\Aparam$. Due to the symmetry of $\Abasis{i}$, $i=1,\ldots,\nA$, the
approximation $\matApproxA$ will always be symmetric. The parameter-dependent coefficients 
 $\Acoeff\equiv\entrytuple{\Acoeff}{\nA}$ are computed online in
 an efficient manner that ensures $\matApproxA$ is positive definite and
 thereby preserves requisite structure.
 
The next sections describe procedures for computing the matrix basis and
coefficients. We refer to this method as `matrix gappy POD', as it amounts
to the gappy POD procedure \cite{sirovichOrigGappy} applied to matrix data
with modifications to preserve positive definiteness. The approach, which we
originally proposed \cite{carlbergStructureAiaa}, is
a more general formulation of the `matrix DEIM' approach \cite{wirtzDeim} (or
`multi-component EIM' \cite{tonnThesis} in the context of the reduced-basis method applied to
parametrized non-affine elliptic PDEs), as 
it permits least-squares reconstruction (not simply interpolation). Further,
it is equipped with a mechanism to maintain positive definiteness. 

\subsubsection{Offline computation: matrix basis}\label{sec:matrixBasis}

To obtain the matrix basis, we propose applying a vectorized POD method,
wherein
the basis can be considered a set of `principal matrices' that
optimally represent\footnote{These matrices are optimal in the sense that they
minimize the average projection error (as measured in the Frobenius norm) of
the matrix snapshots.} the matrix $\A$ over the training set
$\paramTrain$. The (offline) steps for this method are as follows:
 \begin{enumerate} 
  \item Collect matrix snapshots $\AparamArg{\paramTraini}$,
	$i=1,\ldots,\nTrain$.
  \item Vectorize the snapshots 
  $\vectorizeAsamplei \equiv \vectorize{\AparamArg{\paramTraini}}\in\RR{N^2}$,
	$i=1,\ldots,\nTrain$,
	where the function $\vectorizeNo:\RR{N\times N}\rightarrow
	\RR{N^2}$ vectorizes a matrix.
	\item Compute an $\nA$-dimensional (with $\nA\leq\nTrain$) POD basis of the vectorized snapshots
 \begin{equation} \label{eq:vectorizedSnaps}
\podVectorizedA\equiv\vecmat{\vectorizeAbasis}{\nA}\in\RR{N^2\times
	\nA}
 \end{equation} 
	 using vectorized snapshots $\{\vectorizeAsamplei\}_{i=1}^\nTrain$ and
	an `energy criterion' $\energyCrit_{\A}\in\left[0,1\right]$ as inputs to
	Algorithm \ref{PODSVD} of Appendix \ref{app:POD}.  
	\item Transform these POD basis vectors into their matrix counterparts:
 \begin{equation} 
 \Abasis{i} = \vectorizeInv{\vectorizeAbasisi},\quad
 i=1,\ldots,\nA.
  \end{equation} 
	 \end{enumerate}
	 Each matrix $\Abasis{i}$, $i=1,\ldots,\nA$ is guaranteed to be symmetric,
as Algorithm \ref{PODSVD} forms this basis by taking 
	 linear combinations of symmetric matrices.

\subsubsection{Online computation: coefficients}
The approximation error  can be bounded as
follows:
\begin{align}
\|\Ared - \matApproxA \|_F &=
\|\Ared - \sum_{i=1}^\nA\AcoeffiParam \podstate^T\Abasisi\podstate
\|_F\\
\label{eq:upperboundLS}& \leq \|\podstate\|_F^2\bigl\|\Aparam - \sum_{i=1}^\nA\AcoeffiParam \Abasisi
\bigr\|_F
\end{align}
where
$\|\podstate\|_F^2 = \nstate$ if $\podstate$ is orthogonal.  
This leads to a natural choice for the scalar coefficients based on 
minimizing the upper bound \eqref{eq:upperboundLS}.
In particular, we compute coefficients $\AcoeffParamOn$ online for a specific 
$\paramOnline\in\paramDomain$ 
as the solution to 
\begin{align}\label{eq:optProblem}
\begin{split}
\underset{\left(x_1,\ldots,x_{\nA}\right)}{\mathrm{minimize}}\quad&
\|
\sampleMatT\AparamOn\sampleMat  - 
\sum\limits_{i=1}^\nA
x_i\sampleMatT\Abasis{i} \sampleMat 
\|_F^2\\
\mathrm{subject\ to}\quad &\sum_{i=1}^\nA x_i\podstate^T\Abasis{i}\podstate
> 0.
\end{split}
\end{align}
Note that the coefficients are computed to match (as closely as possible) the
full matrix and the linear combination of pre-computed full matrices at a few
entries. The constraints amount to a strict linear-matrix-inequality, where
$\A>
0$ denotes a generalized inequality that indicates $\A$ is a
positive-definite matrix.  This  constraint ensures that structure is
preserved. Note that the constraint can be modified (resp.\
dropped) in cases where positive semidefiniteness (resp.\ simply symmetry)
aims to be preserved.

Problem \eqref{eq:optProblem} is equivalent to a linear least-squares
problem with nonlinear constraints; this can be seen from its vectorized form:
\begin{align}\label{eq:optProblem2}
\begin{split}
\underset{\x=\left[x_1\ \cdots\ x_\nA\right]^T}{\mathrm{minimize}}\quad&
\|
\sampleMatVectorizedT\vectorize{\AparamOn}
- 
\sampleMatVectorizedT\podVectorizedA \x 
\|_2^2\\
\mathrm{subject\ to}\quad &\sum_{i=1}^\nA x_i\podstate^T\Abasis{i}\podstate
> 0.
\end{split}
\end{align}
Here, 
$\sampleMatVectorized$ is an alternate form of the sampling matrix 
that can be applied to vectorized matrices, i.e. 
$\sampleMatVectorized^T \vectorize{\AparamOn} =
\vectorize{\sampleMatT\AparamOn\sampleMat}.$\footnote{Exploiting symmetry,
this sampling matrix can be expressed as $\sampleMatVectorized\equiv
\vecmat{\sampleMatVectorizedVec}{(\nsample^2+\nsample)/2}\in\{0,1\}^{N^2\times
(\nsample^2+\nsample)/2}$, where
$\sampleMatVectorizedVec^{i+(j^2 - j)/2} =
\vectorize{\sampleMatVec^i\left[\sampleMatVec^j\right]^T}$ for
$i=1,\ldots,j$ and $j=1,\ldots,\nsample$
and $\sampleMat\equiv \vecmat{\sampleMatVec}{\nsample}$.
From the the definition of $\sampleMatVectorizedVec^{i+(j^2 - j)/2}$,
it follows that $\sampleMatVectorizedVec^{i+(j^2 - j)/2}$
extracts the  $(i,j)$th entry from the vectorized
form of a matrix.}

The objective function is equivalent to that of the gappy POD method \cite{sirovichOrigGappy}---which will be
further discussed in Section \ref{sec:eim}---applied to matrix
data.
Note that this optimization problem is solved \emph{online} using the
online-sampled data $\sampleMatT\AparamOn\sampleMat$;
Appendix~\ref{app:optimization} describes a method for solving this optimization
problem.
In practice, we
usually observe the constraints to be inactive at the unconstrained solution.
Therefore, typically the constraints need not be handled directly, and solving
problem \eqref{eq:optProblem} amounts to solving a small-scale linear
least-squares problem characterized by an $(\nsample^2 + \nsample)/2 \times
\nA$ matrix. To ensure a unique solution to
problem~\eqref{eq:optProblem}, the matrix
$\sampleMatVectorizedT\podVectorizedA$ must have full column rank. This can be
achieved by enforcing $(\nsample^2 +
\nsample)/2\geq\nA$ as well as mild conditions on the sampling matrix
$\sampleMat$.

Procedure \ref{procedure:matrixGappy} describes the offline and online stages
for implementing the matrix gappy POD approximation.
\begin{algorithm2e}\label{procedure:matrixGappy}
\caption{\textsf{Matrix gappy POD}}
\BlankLine
\tcm{\textbf{\textsf{Offline stage}}}
\BlankLine
\setcounter{AlgoLine}{0}
Compute the basis matrices $\Abasis{i}$, $i=1,\ldots,\nA$ using the
vectorized POD approach described in Section \ref{sec:matrixBasis}.\;
	Determine the sampling matrix $\sampleMat$ 
which gives rise to 
a full column rank matrix $\sampleMatVectorizedT\podVectorizedA$
and with 
$m$ chosen so that $(\nsample^2 + \nsample)/2\geq\nA$.\;
Compute low-dimensional matrices $\podstate^T\Abasisi\podstate$,
	$i=1,\ldots,\nA$.\;
Retain the sampled entries of the matrix basis
	$\sampleMatT\Abasisi\sampleMat$, $i=1,\ldots,\nA$; discard other entries.\;
\setcounter{AlgoLine}{0}
\BlankLine
\tcm{\textbf{\textsf{Online stage}} \textsf{(given $\paramOnline$)}}
\BlankLine
Compute $\sampleMatT\A(\paramOnline)\sampleMat $.\;
Solve the small-scale optimization problem \eqref{eq:optProblem} for coefficients
$\AcoeffParamOn$.\;
Assemble the low-dimensional matrix
$\matApproxAOn$  by Eq.\ \eqref{eq:spApprox}.\;
\end{algorithm2e}

\subsubsection{Exactness conditions}\label{sec:matGappyExact}

\begin{theorem} \label{matrixGapExact}
The matrix gappy POD approximation is exact for any specified online parameters
$\paramOnline\in\paramDomain$ if 
 \begin{enumerate} 
	\item $\vectorize{\AparamOn}\in\range{\podVectorizedA}$
	 and
	\item $\sampleMatVectorizedT\podVectorizedA$ has full column rank.
	 \end{enumerate}

\end{theorem}
See Appendix \ref{sec:matrixGapExactProof} for the proof.

Condition 1 holds, e.g.,  when $\paramOnline \in\paramTrain$ and $\nA =
\nTrain$. Condition 2 can
be straightforwardly enforced by the choice of $\sampleMat$
and automatically holds in the case of full sampling, i.e., $\nsample = N$.

%

\section{Preserving potential-energy structure}\label{sec:potEnGen}

This section presents a technique for approximating Lagrangian ingredient
\ref{prop:potential} within an algebraic setting.  To begin, define
a parameterized scalar-valued function
$\potEn:\RR{N}\times\paramDomain\rightarrow \RR{}$ with $(\q;\param)\mapsto \potEn$  that is nonlinear
in both arguments and can
be interpreted as a Lagrangian dynamical system's (parameterized) potential
energy.
Here, 
$\q\in\RR{N}$ denotes the system's configuration variables and
$\param\in\paramDomain$ denotes the
system parameters that belong to parameter domain $\paramDomain$. Unlike the
matrix approximations of the previous section, the nonlinear dependence on the
configuration variables $\q$ introduces additional challenges that must be considered
carefully. 

We aim to devise an \emph{offline} method---which may entail expensive
operations---for constructing a scalar-valued
function $\potEnRedApprox:\RR{\nstate}\times\paramDomain\rightarrow{\RR{}}$.
This function will be used \emph{online} and should satisfy the demands of
online problem (P\ref{probPotEn}):
\vskip .1in
\begin{list}
{(P\arabic{problemctr})}{\usecounter{problemctr}\setlength{\rightmargin}{\leftmargin}}
\refstepcounter{problemctr}
\item \label{probPotEn}
 Compute the gradient vector
 $\nabla_\qred\potEnRedApprox(\qredOnline;\paramOnline)$ at a cost independent
 of $N$. Given any online parameters $\paramOnline\in\paramDomain$, this
 vector should be appropriately close to
 $\podstate^T\nabla_\q\potEn(\qRefOn + \podstate\qredOnline;\paramOnline)$
 for all coordinates $\qredOnline\in\RR{\nstate}$.
\end{list}
As before, $\podstate$ represents a dense, parameter-independent,
full-column-rank $N\times \nstate$ matrix. 
We denote by $\qRefNo:\paramDomain\rightarrow\RR{N}$ a parameterized reference configuration about which
the low-dimensional reduced configuration space is centered.
Notice that this problem is concerned with approximating the gradient of the
scalar-valued function as opposed to the function itself. As will be discussed in Section \ref{sec:existMR}, this problem arises in
model reduction of parameterized Lagrangian-dynamics systems, where the
gradient of the potential appears in the equations of motion.

In certain specialized cases, the above
approximation can be simplified considerably. For example, 
when $\qRef = 0$, $\forall \param\in\paramDomain$ and the function $\potEn( \q;\param)$ is
purely quadratic in its first argument, then $\podstate^T\nabla_\q\potEn(\qRef +
\podstate\qred;\param) =
\podstate^T\Aparam\podstate\qred$, where
$\Aparam$ is a symmetric Hessian matrix; in this case, one of the
approximation techniques described in Section \ref{sec:matrixApprox} can be
straightforwardly applied. Alternatively, if the potential energy is defined by
the integral over a domain (i.e., $\potEn(\q ;\param) = \int_\Omega
\potEnSpace(X,\q ;\param) d_{\Omega_X}$), a sparse cubature method
\cite{an2008optimizing} can be used to achieve computational efficiency and
structure preservation. In more general cases, however, another approach is 
needed. In the following, we develop a method that makes no simplifying assumptions about
the dependence of the potential $\potEn$ on the configuration variables or
parameters. 

Due to the density of the matrix $\podstate$, the most straightforward approach of setting
$\potEnRedApprox(\qred;\param) = \potEn(\qRef + \podstate\qred;\param)$
leads to expensive online operations: computing the gradient 
$\nabla_\qred\potEnRedApprox(\qredOnline;\paramOnline) =
\podstate^T\nabla_\q\potEn(\qRefOn +
\podstate\qredOnline;\paramOnline)$ requires first computing all $N$
entries of the gradient vector $\nabla_\q\potEn(\qRefOn + \podstate\qredOnline;\paramOnline)$.
To rectify this, we revisit the RBS technique proposed in Section
\ref{sec:sparseMat} and introduce some minor 
modifications.  In particular, we replace $\podstate$ by a 
sparse \emph{parameter-dependent} matrix $\sparsePotEn\equiv\sampleMat
\sparsePotEnSmall\in\RR{N\times\nstate}_*$ with only $\nsample\ll N$ nonzero
rows, where 
$\sparsePotEnSmall\in\RR{\nsample\times \nstate}_*$ is a dense matrix. That is, we approximate the
potential energy as
\begin{align}\label{eq:potEnSparse}
\potEnRedApprox(\qred;\param)
\equiv&\potEn(\qRef+\sparsePotEn\qred;\param).
\end{align}
This approximation preserves structure, as $\potEnRedApprox$ remains a
parameterized scalar-valued function.
Now, we wish to compute
$\sparsePotEnNo$ such that 
$\nabla_\qred\potEnRedApprox(\qredOnline;\paramOnline) =
\sparsePotEnOn^T\nabla_\q\potEn\left(\qRefOn+\sparsePotEnOn\qredOnline;\paramOnline\right)$ is
as close as possible to $\podstate^T\nabla_\q\potEn(\qRefOn +
\podstate\qredOnline;\paramOnline)$ 
for any online point $\paramOnline\in\paramDomain$ and
any $\qredOnline\in\RR{\nstate}$. 
One
can imagine a variety of methods for computing $\sparsePotEnNo$ toward this
stated goal. For example, one can formulate an optimization problem to match
the potential gradient at training points \cite{carlbergStructureAiaa};
this effectively leads to a parameter-independent matrix $\sparsePotEnNo$.
However, we found this approach to lead to significant errors for many
problems. Instead, we pursue an idea motivated by the analysis in Section
\ref{sec:exactPotEn}, which centers on the first two terms in a 
Taylor expansion of 
$\podstate^T\nabla_\q\potEn(\qRefOn + \podstate\qredOnline;\paramOnline)$ 
about the reference configuration.

In practice, we often find that the trajectories of dynamical systems are
localized in the configuration space. This is particularly true for mechanical
oscillators often encountered in structural dynamics, where the trajectory
does not deviate drastically from the
equilibrium configuration. Using this observation, we focus our approximation
efforts on accurately capturing the behavior of the potential in a
neighborhood of the
online reference configuration $\qRefOn$. Implicitly, this assumes that the online
configurations do not greatly diverge from this point.
To this end, 
consider computing
$\sparsePotEnOn$ online such that the approximation
$\sparsePotEnOn^T\nabla_\q\potEn(\qRefOn+\sparsePotEnOn\qredOnline;\paramOnline)$
matches 
$\podstate^T\nabla_\q\potEn(\qRefOn + \podstate\qredOnline;\paramOnline)$ 
to first order about the reference configuration:
 \begin{align} \label{eq:matchFirstOrder}
 \begin{split}
&\sparsePotEnOn^T\nabla_\q\potEn(\qRefOn;\paramOnline) +
\sparsePotEnOn^T\nabla_{\q\q}\potEn(\qRefOn;\paramOnline)\sparsePotEnOn\qredOnline\\
&=
\podstate^T\nabla_\q\potEn(\qRefOn ;\paramOnline) + 
\podstate^T\nabla_{\q\q}\potEn(\qRefOn;\paramOnline)\podstate\qredOnline,\quad
\forall \qredOnline\in\RR{\nstate}.
 \end{split}
  \end{align} 
Notice that the high-order terms amount to
approximating a reduced Hessian (defined via the dense matrix $\podstate$)
by a second reduced Hessian (defined via the sparse matrix $\sparsePotEnOn$).
This is equivalent to online problem (P\ref{probMat}) presented 
in Section~\ref{sec:matrixApprox} that was addressed by the
RBS algorithm (as well as a matrix gappy POD approach). This RBS algorithm
is supported by Theorem~\ref{thm:quadCase}, which
shows that an exact approximation of the reduced Hessian is possible under certain 
assumptions. While these assumptions do not always hold, the theorem  gives an
expectation that a good approximation can be found under more general circumstances.
Unfortunately, the presence of the low-order terms in Eq.\
\eqref{eq:matchFirstOrder}
alters
the character of the reduced approximation and so Theorem~\ref{thm:quadCase} no 
longer applies. In this case, the matrix $\sparsePotEnOn$ must 
serve to capture both gradient and Hessian information simultaneously, which
introduces
restrictive assumptions
in order to
obtain an equivalent result to Theorem~\ref{thm:quadCase}; this will be shown
in Lemma \ref{quadratic solvability}.
%

To avoid the limitations associated with these restrictions,
we choose the reference configuration to be equilibrium, i.e.,
$\qRef=\initialQ$ with equilibrium defined as $\nabla_\q V(\initialQ;\param) = 0$. 
This forces the low-order Taylor terms to zero and
simplifies Eq.~\eqref{eq:matchFirstOrder} to 
 \begin{equation} \label{eq:matchFirstOrderSimple}
 \sparsePotEnOn^T\nabla_{\q\q}\potEn(\initialQOn;\paramOnline)\sparsePotEnOn\\ =
 \podstate^T\nabla_{\q\q}\potEn(\initialQOn;\paramOnline)\podstate.
  \end{equation} 
Now, Theorem~\ref{thm:quadCase} holds, implying that 
Equation \eqref{eq:matchFirstOrderSimple} can be exactly solved 
when $\nsample = \nstate$.  For this reason, we compute $\sparsePotEnOn$ online
for each $\paramOnline$ to satisfy \eqref{eq:matchFirstOrderSimple} using
	$\nstate$ sample indices.
Specifically, we define it according to 
 \begin{align} \label{eq:optMatch}
 \begin{split}
\sparsePotEnSmallOn &= \left[\begin{array}{c}
\cholSol\\
\mathbf 0_{(\nsample-\nstate)\times \nstate}
\end{array}\right], 
 \end{split}
  \end{align} 
where $\cholSol$ is given by solving
$$
 \cholSample^T \cholSol = \cholHess^T,
$$
$\cholHess\in\RR{\nstate \times \nstate}$ denotes the lower-triangular Cholesky factor of
	$\podstate^T\nabla_{\q\q}\potEn(\initialQOn;\paramOnline)\podstate$, 
	$\cholSample\in\RR{\nstate \times \nstate}$ denotes the lower-triangular Cholesky factor of
	$\sampleMatMinT\nabla_{\q\q}\potEn(\initialQOn;\paramOnline)\sampleMatMin$,
	and
	$\sampleMatMin$ represents the first $\nstate$ columns of $\sampleMat$.
We defer discussing the computational cost for this approach to Section~\ref{sec:procedure:rbsPotEn},
and now return to the previously alluded difficulties associated with 
solving \eqref{eq:matchFirstOrder} when the reference configuration does not
correspond to equilibrium.



\subsection{Solvability of the two-term Taylor equation}\label{sec:exactPotEn}
The method presented in the previous section was motivated by difficulties
in inexpensively approximating the reduced gradient of a nonlinear function.
In this section, we give
some insight into these difficulties by investigating a much easier situation:
the solvability of the two-term Taylor equation~\eqref{eq:matchFirstOrder}, which we write 
in matrix/vector form as
\begin{equation} \label{FirstOrderSimpleForm}
 \sparsePotEnSmallNo^T \sampleMat^T 
\cvector 
+
\sparsePotEnSmallNo^T \sampleMat^T 
\A
 \sampleMat \sparsePotEnSmallNo\qredOnline
=
\podstate^T
\cvector
+
\podstate^T
\A
\podstate\qredOnline,\quad
\forall \qredOnline\in\RR{\nstate}.
\end{equation}
Here, we have set 
$\cvector = \nabla_\q\potEn(\qRefOn;\paramOnline)$
and 
$ \A = \nabla_{\q\q}\potEn(\qRefOn;\paramOnline)$.
We have also dropped dependence on $\paramOnline$ such that $\cvector$ and
$\A$ are parameter independent in the following analysis;
this is equivalent to restricting equation~\eqref{eq:matchFirstOrder} 
to a single instance of $\paramOnline$. This is somewhat less than
ideal in that we would normally wish to minimize online costs
by computing a single $\sparsePotEnSmallNo$ during the offline phase 
that is then valid for all subsequent online calculations.
However, what we now show is that it is not always possible to satisfy 
equation~\eqref{eq:matchFirstOrder} even when one is restricted to finding 
a $\sparsePotEnSmallNo$ for a single instance of $\paramOnline$. 

As \eqref{FirstOrderSimpleForm} must hold for all $\qredOnline$, we have
the following two necessary and sufficient conditions 
\begin{equation} \label{original conditions}
\sparsePotEnSmallNo ^T \sampleMatT  \A  \sampleMat  \sparsePotEnSmallNo
= 
\podstate^T \A  \podstate 
~~~\mbox{and}~~~
\sparsePotEnSmallNo^T \sampleMat^T \cvector 
 = \podstate^T  \cvector .
\end{equation}
It is possible to show that satisfying these conditions is equivalent to
finding a $\TransformedsparsePotEn \in \RR{\nsample \times \nstate }$
such that 
\begin{equation} \label{quadratic conditions}
\TransformedsparsePotEn^T \TransformedsparsePotEn = I  
~~~\mbox{and}~~~
\TransformedsparsePotEn^T \sampleMat^T \tilde{\cvector} = 
\widetilde{\podstate}^T \tilde{\cvector} 
\end{equation}
where
$$
\widetilde{\podstate}^T \widetilde{\podstate} = I.
$$
The definitions of $\TransformedsparsePotEn ,\widetilde{\podstate},$
and $\tilde{\cvector}$ are given below. The key point is that
the necessary and sufficient conditions for equation 
\eqref{quadratic conditions} amount to finding
an orthogonal matrix, $\TransformedsparsePotEn $, such that
$ \TransformedsparsePotEn^T \sampleMat^T \tilde{\cvector} = 
\widetilde{\podstate}^T \tilde{\cvector} $ 
for a given orthogonal matrix $\widetilde{\podstate}$, and a 
given vector, $\tilde{\cvector}$.
In the simple case
when $\A$ is the identity and $\podstate$ is orthogonal,
we have 
$\TransformedsparsePotEn = \sparsePotEnSmallNo $
and $\widetilde{\podstate} = \podstate$.
More generally, we have the following definitions:
$$
\TransformedsparsePotEn = \sampleMat^T \Lmat ^T  \sampleMat \sparsePotEnSmallNo  \Lmat_\phi^{-T}\mbox{,\ \ \ }~
\widetilde{\podstate} = \Lmat ^T \podstate \Lmat_\phi^{-T}
\mbox{,\ \ \ }~
\tilde{\cvector} = \Lmat ^{-1} \cvector ,
$$
where $\Lmat $ is the lower-triangular Cholesky factor of $\A $,
and 
$\Lmat_\phi$ is the lower-triangular Cholesky factor of $\podstate^T \A  \podstate$.
The above equivalence hinges on the identities
$\sampleMat^T \Lmat \sampleMat \sampleMat^T \Lmat^T \sampleMat = 
\sampleMat^T \A \sampleMat $ and
$\sampleMat^T \Lmat \sampleMat \sampleMat^T \Lmat^{-1} = 
 \sampleMat^T $.
These hold due to the lower-triangular form of the matrix $\Lmat$.
The following lemma addresses the conditions under which
Eq.~\eqref{quadratic conditions}
or equivalently
Eq.~\eqref{FirstOrderSimpleForm}
 hold.
\begin{lem} \label{quadratic solvability}
Consider the equations
\begin{equation} \label{lemma version}
\TransformedsparsePotEn^T \TransformedsparsePotEn = I  
~~~\mbox{and}~~~
\TransformedsparsePotEn^T \sampleMatT  \widetilde{\cvector}
= 
\widetilde{\podstate}^T \tilde{\cvector} 
\end{equation}
with the following matrices given:
$\sampleMat  \in \{0,1\}^{N\times \nsample}$ consists of selected columns of the identity
matrix (see prior definition),
$\widetilde{\podstate} \in \RR{N\times \nstate }$ with
$\widetilde{\podstate}^T \widetilde{\podstate} = I$,
and
$\tilde{\cvector}    \in \RR{N\times 1}$.
Then, assuming that $\nsample  \ge \nstate $, 
some $\TransformedsparsePotEn  \in \RR{\nsample \times \nstate }$ exists such that
Eq.~\eqref{lemma version} is satisfied if and only if
\begin{equation} \label{new conditions}
|| \widetilde{\podstate}^T \tilde{\cvector}  ||_2 = || \sampleMatT  \tilde{\cvector}  ||_2 
~~\mbox{and}~~ \nsample  = \nstate 
~~~\mbox{or}~~~
|| \widetilde{\podstate}^T \tilde{\cvector}  ||_2 \le || \sampleMatT  \tilde{\cvector}  ||_2 
~~\mbox{and}~~ \nsample  > \nstate  .
\end{equation}
\end{lem}
See Appendix~\ref{app:ray} for the proof.

Obviously, equation \eqref{new conditions} is satisfied if either $\tilde{\cvector} = 0$ (i.e., the equilibrium
configuration is taken as the reference configuration) or if $\nsample = N$.
Unfortunately, however, equation \eqref{new conditions} is not guaranteed to be satisfiable
in more general situations. Specifically, when  $\nsample  > \nstate $, the condition
$|| \widetilde{\podstate}^T \tilde{\cvector}  ||_2 \le || \sampleMatT  \tilde{\cvector}  ||_2 $
corresponds to comparing the magnitude of a vector of length $\nstate $ 
obtained by rotating and dropping components with a second vector
of length $\nsample $ obtained by simply dropping components. If 
$\widetilde{\podstate}$ and $\tilde{\cvector}$ are not correlated, then one could perhaps
hope that on average the vector with more components would generally 
have a larger magnitude.  However, when 
$\tilde{\cvector}$ lies completely within the subspace spanned by the columns of
$\widetilde{\podstate}$ and all components of $\tilde{\cvector}$ are non-zero,
then
$|| \widetilde{\podstate}^T \tilde{\cvector}  ||_2 = || \tilde{\cvector}  ||_2 $ and so
satisfying the necessary and sufficient conditions requires `full sampling' $\nsample  = N$. 
While this scenario may be considered pessimistic, one can expect that a very large
value of $\nsample$ will be required when $\tilde{\cvector}$ lies
primarily in the range space of $\widetilde{\podstate}$.
In general, there is no guarantee that even the 
simplified (i.e., parameter-independent)
form of the two-term Taylor equation is solvable. When one also considers
that Eq.~\eqref{FirstOrderSimpleForm} corresponds to the
restriction of Eq.~\eqref{eq:matchFirstOrder} to a single instance
of $\paramOnline$, the above result should be seen as quite discouraging.

For this reason, we abandon any attempt at computing a parameter-independent sparse matrix
$\sparsePotEnNo$ during the offline phase that can serve to approximate
the reduced gradient for all online points $\paramOnline\in\paramDomain$. Instead, we
limit ourselves to the online computation of a parameter-dependent matrix
$\sparsePotEnOn$ 
that is only valid for a single point $\paramOnline$ but can be used for all
reduced configuration variables $\qredOnline\in\RR{n}$ that may arise during the
online evaluation, e.g., at each nonlinear iteration and time instance
considered while numerically solving the equations of motion.
Additionally, we set the reference configuration to equilibrium, which
results in  $\tilde{\cvector}=0$ and guarantees solvability of the
the two-term Taylor expression
with $\nsample = \nstate$.

\subsection{Implementation and cost}\label{sec:procedure:rbsPotEn}

Procedure \ref{procedure:rbsPotEn} summarizes the offline/online strategy for
implementing the RBS strategy for approximating the potential energy.
\begin{algorithm2e}\label{procedure:rbsPotEn}
\caption{\textsf{Reduced-basis sparsification for potential energy}}
\BlankLine
\tcm{\textbf{\textsf{Offline stage}}}
\BlankLine
\setcounter{AlgoLine}{0}
  \label{step:potEnSnapshot} 
	Determine the sampling matrix $\sampleMat$.\;
\setcounter{AlgoLine}{0}
\BlankLine
\tcm{\textbf{\textsf{Online stage}} \textsf{(given $\paramOnline$)
}}
\BlankLine
Compute $\podstate^T
\nabla_{\q\q}V\left(\initialQOn;\paramOnline\right)\podstate$.
\;\label{step:redPotEnGrad}
Compute $\sampleMatMinT \nabla_{\q\q}  V\left(\initialQOn
;\paramOnline\right)\sampleMatMin$.
\;\label{step:samplePotEnGrad}
Solve Equation \eqref{eq:optMatch} for $\sparsePotEnSmallOn$.\;\label{step:finishApproxPotEn}
For any $\qredOnline\in\RR{\nstate}$, set
$\potEnRedApprox(\qredOnline;\paramOnline)
=\potEn(\initialQOn+\sampleMat\sparsePotEnSmallOn\qredOnline;\paramOnline)$,
and compute the gradient as
$\nabla_\qred\potEnRedApprox\left(\qredOnline;\paramOnline\right) =
\sparsePotEnSmallOn^T\sampleMat^T\nabla_q\potEn(\initialQOn+\sampleMat\sparsePotEnSmallOn\qredOnline;\paramOnline)$
\;\label{step:finalApproxPotEn}

\end{algorithm2e}
This method satisfies the online computational cost requirements of
problem (P\ref{probPotEn}) with one exception: 
online step \ref{step:redPotEnGrad} incurs an $N$-dependent operation
count. However, online steps 1--3 depend only on the
online point $\paramOnline$ and not on the reduced configuration variables
$\qredOnline$. Thus, these steps are performed only once per
parameter instance, and their cost can be amortized over all online-queried
values of $\qredOnline$.
As a result, this does not preclude significant computational savings, as will
be shown in the numerical results reported in Section \ref{sec:experiments}.
Note that online
step~\ref{step:samplePotEnGrad} is equivalent to computing just $\mathcal
O(\nstate^2)$
entries of $\nabla_{\q\q} V$, which can be completed at a cost independent of
$N$.
 Step~\ref{step:finishApproxPotEn} requires $\mathcal
O(\nstate^3)$ operations.

\begin{remark}
Most nonlinear reduced-order modeling methods
\cite{astrid2007mpe,ryckelynck2005phm,chaturantabut2010journal,galbally2009non,drohmannEOI,CarlbergGappy,carlbergJCP}
assume `$H$-independence' \cite{drohmannEOI}, which states that the Jacobian
of the vector-valued nonlinear function  is sparse; in the present context,
this corresponds to sparsity of the matrix $\nabla_{\q\q}\potEn$. When this assumption
holds, the proposed methodology incurs low online computational cost. This efficiency results from the fact
that 
computing $\sampleMatT\nabla_\q
\potEn(\initialQOn+\sampleMat \sparsePotEnSmallOn\qredOnline;\paramOnline)$  in Step \ref{step:finalApproxPotEn} of
Procedure \ref{procedure:rbsPotEn}
requires that only $\nsample$ components of the gradient $\nabla_\q\potEn$ be
evaluated; if $H$-independence holds, then these $\nsample$ components depend
on only $\mathcal
O(\nsample)$ components of the argument $\initialQOn+\sampleMat
\sparsePotEnSmallVar\qred$, leading to an $N$-independent operation count. 

Unfortunately, $H$-independence does not hold for some problems in Lagrangian
dynamics. For example molecular-dynamics models can be characterized by
a potential that includes interaction terms between all particles, resulting in
a dense matrix $\nabla_{\q\q}\potEn$. Here, the proposed method can
still achieve efficiency by `centering' the configuration space at equilibrium
such that $\initialQ = 0$, $\forall \param\in\paramDomain$.  
In this case, the method requires computing only $\nsample$
components of the argument $\initialQOn+\sampleMat
\sparsePotEnSmallVar\qred$ in Step \ref{step:finalApproxPotEn} of
Procedure \ref{procedure:rbsPotEn} regardless of the sparsity of the matrix
$\nabla_{\q\q}\potEn$. This efficiency is achievable due to the fact that the
method injects `sparsification'
in the argument of the nonlinear function.  This ability to achieve an
$N$-independent operation count when $H$-independence is violated
distinguishes this method from others in the literature.

\end{remark}

\section{Lagrangian dynamics formulation}\label{sec:prob}

We have now developed techniques to approximate parameterized reduced
symmetric-positive-definite matrices and potential functions. In this section, we show how these
methods enable us to achieve the objective of this work: preserving Lagrangian structure
in model reduction for nonlinear mechanical systems. We begin by presenting
the Lagrangian-dynamics formulation for such systems and highlighting critical
problem structure.  Later, Section \ref{sec:existMR} describes existing nonlinear
model-reduction techniques and explains how they destroy structure in this
context. Section \ref{sec:newMethod} presents the proposed
structure-preserving methodology, which employs the approximation techniques
proposed in Sections \ref{sec:matrixApprox} and \ref{sec:potEnGen}.

We consider parameterized, nonlinear \emph{simple mechanical systems}, with a
particular focus on structural-dynamics models constructed by a finite-element
formulation.  Such models are defined by a triple
$(\Q,g,\potEn)$ parameterized by system parameters $\param\in\paramDomain$. The parameters may describe
variations in shape and material
properties, for example. The triple is composed of:
 \begin{itemize} 
  \item 
A differentiable configuration manifold $\Q$.  We take $\Q =
\RR{N}$ where $N$ denotes the number of degrees of freedom in the
model, considered to be `large' in this work.
\item A parameterized Riemannian metric $g(\vVec , \w ;\param)$,
where $\vVec $ and $\w $ belong to the tangent bundle of $\Q$.
We take $g(\vVec , \w ; \param) = \vVec ^T\Mparam \w $, 
where $\Mparam$ denotes the $N\times N$ parameterized symmetric positive-definite mass matrix.
 \item A parameterized potential-energy function 
 $\potEn:\Q\times \paramDomain\rightarrow \RR{}$.
	 \end{itemize}


The kinetic energy of a simple mechanical system can be expressed as
$T(\dot\q;\param)=\half g(\dot
		\q ,\dot \q ;\param) = \half\dot
		\q ^T\Mparam\dot \q $, where
$\q:\left[0,\lastT\right]\rightarrow\Q$ denotes the time-dependent
configuration variables and 
$\lastT$ denotes the final time.
This leads to the following expression for the
Lagrangian, which represents the difference between the kinetic and potential
energies:
	\begin{align}
	L(\q ,\dot \q ;\param) &=\frac{1}{2}g\left(\dot
	\q,\dot\q;\param\right)-\potEn\left(\q;\param\right)\\
	\label{eq:lagrangian}&= \half\dot \q ^T\Mparam\dot \q  - \potEn(\q ;\param).
	\end{align}
In many cases, the non-conservative forces\footnote{Conservative forces can be
handled by directly including them in the Lagrangian.} consist of
an applied external
force 
and a dissipative force arising from Rayleigh viscous damping. This
dissipative force derives from a positive-semidefinite dissipation
function\footnote{Non-viscously damped systems can also often be derived by a
positive-semidefinite dissipation function \cite{adhikari2000damping}.}
 \begin{equation}\label{eq:dissFun}
 \dissFun{\dot \q }\equiv \frac{1}{2}\dot \q ^T\Cparam \dot \q ,
  \end{equation} 
 where
 $\Cparam = \rayleighM \Mparam + \rayleighK
 \nabla_{\q\q}\potEn(\initialQ;\param)$ denotes a parameterized $N\times N$ symmetric
 positive-semidefinite matrix with $\rayleighM\in\RR{}$ and
 $\rayleighK\in\RR{}$. 
Here, $\initialQNo:\paramDomain\rightarrow\RR{N}$ denotes the (parameterized) equilibrium
configuration such that $\nabla_\q\potEn(\initialQ;\param) = 0$. So, we consider
 non-conservative forces of the form $\fextparam-\nabla_{\dot \q }
 \dissFun{\dot \q }$, where $\fext$ denotes the external force that is derived from
 the Lagrange--D'Alembert variational principle.

Given the Lagrangian \eqref{eq:lagrangian}, one can derive the equations of
motion for a simple mechanical system subject to an external force and Rayleigh
viscous damping from the forced Euler--Lagrange equation
 \begin{equation} \label{eq:fomLag1}
 \frac{d}{dt}\nabla_{\dot \q }L (\q ,\dot \q ;\param)- \nabla_\q L(\q ,\dot \q ;\param)=
 \fextparam-\nabla_{\dot \q } \dissFun{\dot \q }.
  \end{equation} 
Substituting Eqs.~\eqref{eq:lagrangian} and 
 \eqref{eq:dissFun} into Eq.~\eqref{eq:fomLag1} leads
 to the familiar equations of motion 
 \begin{equation} \label{eq:sdmEom}
 \Mparam\ddot \q  + \Cparam\dot \q  + \nabla_\q  \potEnParam = \fextparam.
  \end{equation} 

Conservative mechanical systems, where $\fextparam=0$ and $\Cparam = 0$,
exhibit important
properties and can be characterized using the Hamiltonian formulation of
classical mechanics discussed in Appendix \ref{app:Hamiltonian}. For example, these systems conserve energy and quantities
associated with symmetry, and their time-evolution maps are symplectic.
Because these properties are intrinsic characteristics of the mechanical
systems, it is desirable for numerical methods to preserve
these properties. As mentioned in the
introduction, the
class of structure-preserving time integrators has been developed for this purpose.
This class of integrators ensures that the numerical solution preserves essential properties such as
energy conservation, momentum conservation, and symplecticity 
\cite{hairer2006geometric,marsden2001discrete}.

For this reason, we aim to develop a reduced-order model that preserves the
structure of the mechanical system, yet is computationally inexpensive to
simulate. This will ensure that the reduced-order model preserves these
characteristic properties. Further, the reduced-order equations of motion for these
can be solved with a structure-preserving time integrator; this will ensure that the
numerical solution computed using the reduced-order model will also preserve
these properties. The properties of the system we seek to preserve
are those enumerated in Section \ref{sec:intro}: a configuration space, a parameterized Riemannian metric, a parameterized potential-energy function,
	a parameterized positive-semidefinite dissipation function, and
	an external force derived from the Lagrange--D'Alembert 
	principle.  The first three properties constitute the
	 parameterized triple that ensures the model describes a simple mechanical
	 system; the last two characterize the 
	 non-conservative forces.

\section{Existing model-reduction techniques}\label{sec:existMR}
Model-reduction techniques aim to generate a low-dimensional model that is
inexpensive to evaluate, yet captures the essential features of the
high-fidelity model. These methods first conduct a computationally
expensive offline stage during which they perform analyses (e.g., solving the
equations of motion, modal analyses) for a training set
$\paramTrain_{i=1}^\nTrain\subset\paramDomain$. Then, these methods employ the
data generated during these analyses to define a  configuration manifold of
reduced dimension, as well as other approximations to achieve efficiency in
the presence of nonlinearities or arbitrary parameter dependence.  This low-dimensional configuration manifold is
subsequently employed to generate a low-dimensional model that can be used to
perform inexpensive analyses for any specified point
$\paramOnline\in\paramDomain$ during the online stage.

When the configuration space is Euclidean (as is the case for
the models considered herein), the configuration space of reduced dimension
$\nstate\ll N$ can be
expressed as
\begin{equation} \label{eq:relate}
\Qredfull\equiv \{\qRef +
	\podstate \qred\ |\ \qred \in \Qred\},
\end{equation} 
where $\qRef:\paramDomain\rightarrow\RR{N}$ denotes the (parameterized)
reference
configuration about which the affine reduced subspace is centered,
$\Qred= \RR{\nstate}$, and $\podstate \in\RR{N\times\nstate}_*$ defines
the reduced basis represented as a (typically dense) matrix. This leads to the
following expression for the generalized coordinates and their derivatives:
 \begin{gather} 
\label{eq:qSpace1}\q  = \qRef + \podstate \qred\\
\dot \q  = \podstate \dot \qred\\
\label{eq:qSpaceLast}\ddot \q  = \podstate \ddot \qred.
  \end{gather} 
Thus, the low-dimensional configuration space can be described in terms of
low-dimensional generalized coordinates $\qred\in\Qred$ or
in terms of original coordinates by Eq.\ \eqref{eq:qSpace1}.
The basis $\podstate$ can be determined by a variety of
techniques, including proper orthogonal decomposition and modal decomposition.

\subsection{Galerkin projection}\label{sec:galerkin}
Model reduction based on Galerkin projection preserves Lagrangian
structure. As pointed out by Lall et al.\
 \cite{lall2003structure}, the Galerkin projection must  be carried out
on the Euler--Lagrange equation \eqref{eq:fomLag1}---not the first-order
state-space form---in order to preserve this structure.

Following their approach,
Galerkin-projection-based methods replace the original configuration space
$\Q$ by the reduced-order configuration space $\Qredfull$ and subsequently
derive the equations of motion in the usual way using a set of
lower-dimensional generalized coordinates.
In this way, the resulting model has an identical structure to the original
problem.

For simple mechanical systems subject to non-conservative forces, this amounts to defining the Lagrangian as
 \begin{align} \label{eq:redLagrangian1}
 \Lred(\qred,\dot\qred;\param)&\equiv L(\qRef + \podstate
 \qred,\podstate\dot\qred;\param)\\
 \label{eq:redLagrangian}&=\frac{1}{2}\dot\qred^T\podstate^T\Mparam\podstate\dot\qred -
 \potEn(\qRef+\podstate\qred;\param) 
 \end{align} 
 and the dissipation function as
 \begin{align} 
 \dissFunRed{\dot \qred} &\equiv\dissFun{\podstate\dot\qred}\\
 &=\frac{1}{2}\dot\qred^T\podstate^T\Cparam\podstate\dot\qred.
 \end{align} 
 The external force, which is derived based on the
 Lagrange--D'Alembert variational principle, is transformed by relation
 \eqref{eq:qSpace1} into
 \begin{equation} 
 \fextredparam \equiv \podstate^T\fextparamRomRef.
  \end{equation} 
 Following Section \ref{sec:prob}, the forced Euler--Lagrange equation applied to
the Lagrangian $\Lred$, the dissipation function $\dissFunRedNo$, and the
external force $\fextred$ leads to the reduced-order equations of motion
 \begin{equation} \label{eq:reducedOrderEOM}
 \frac{d}{dt}\nabla_{\dot \qred}\Lred (\qred,\dot \qred;\param)-
 \nabla_{\qred}\Lred(\qred,\dot \qred;\param) +
 \nabla_{\dot \qred}\dissFunRed{\dot \qred}=\fextredparam.
  \end{equation} 
	This can be rewritten as
 \begin{equation} \label{eq:lagrangeGal}
 \podstate^T\Mparam\podstate\ddot\qred +\podstate^T\Cparam\podstate\dot\qred+
 \podstate^T\nabla_\q  \potEn(\qRef+\podstate\qred;\param)
  = \podstate^T\fextparamRomRef.
  \end{equation} 
	Note that Eq.\ \eqref{eq:lagrangeGal} could have also been derived by
	applying Galerkin projection to the original Euler--Lagrange equation
	\eqref{eq:sdmEom}, i.e., making substitutions
	\eqref{eq:qSpace1}--\eqref{eq:qSpaceLast} and left multiplying the system of
	equations by $\podstate^T$.

Thus, the Galerkin reduced-order model preserves the problem structure because it
preserves all five Lagrangian properties:
\begin{enumerate} [I.]
  \item a configuration space $\Qred = \RR{\nstate}$, which
	relates to the original configuration space by  Eq.\ \eqref{eq:relate},
	\item a parameterized Riemannian metric
	$\metricRed\left(\vVec _r,\w _r;\param\right) = \vVec _r^T\podstate^T\Mparam\podstate
	\w _r$,
	\item a parameterized potential-energy function
	$\potEnRed(\q _r;\param) =   \potEn(\qRef + \podstate \q _r;\param)$,
	\item a parameterized positive-semidefinite dissipation function
$\dissFunRedNo(\dot
\qred;\param) = \half\dot\qred^T\podstate^T\Cparam\podstate\dot\qred$, and
	\item an external force $\fextred$ derived from applying the
	Lagrange--D'Alembert principle to the original external force $\fext$, but
	restricted to variations in the configuration space $\Qredfull$.
	 \end{enumerate}

\subsubsection{Computational bottleneck}\label{sec:bottleneck}
Although the equations of motion \eqref{eq:lagrangeGal} are low
dimensional,
they remain  computationally expensive to solve when the operators exhibit arbitrary parameter
dependence and the potential is nonlinear. The reason
is simple: computing the low-dimensional components of
\eqref{eq:lagrangeGal} incurs large-scale operations due to the density of
$\podstate$.  
For example, the following steps are required to compute
$\podstate^T\MparamOn\podstate$ for a specific $\paramOnline\in\paramDomain$
during the online stage:
 \begin{enumerate} [i.]
  \item Compute
$\MparamOn$, which incurs $\mathcal O(N\sparsity)$ flops, 
where $\sparsity$ denotes the average number of nonzeros per row of the matrix
$\MparamOn$.
\item Compute the product $\MparamOn\podstate$, which incurs $\mathcal
O(N\sparsity\nstate)$
flops.
\item Compute the product $\podstate^T(\MparamOn\podstate)$, which incurs
$\mathcal O(N\nstate^2)$ flops. 
	 \end{enumerate}
Thus, the cost scales with the large dimension $N$ of the original
configuration manifold. The same analysis holds for the product
$\podstate^T\CparamOn\podstate$. 

If the potential energy $\potEn$ exhibits a (general) nonlinear dependence on
coordinates $\q$, the situation worsens. In this case, the vector
$\nabla_\q\potEn(\qRef+\podstate\qred;\paramOnline)$ and 
product $\podstate^T\nabla_\q\potEn$ must be computed
for every instance of 
$\qred$.
Similarly, $\podstate^T\fextparamOn$ must be
computed for every time instance.
Thus, a dimension reduction is generally insufficient to generate models
with computational complexity independent of $N$.

\begin{remark}
If the mass matrix is affine in functions of
the parameters $\Mparam =
\sum_i\alpha_i(\param)\Mi$ with $\alpha_i:\paramDomain\rightarrow \RR{}$ and
$\Mi\in\RR{N\times N}$, then products $\podstate^T\Mi\podstate$ can be
assembled offline, and $\podstate^T\MparamOn\podstate =
\sum_i\alpha_i(\paramOnline)\left[\podstate^T\Mi\podstate\right]$ can be
computed in $\mathcal O(\nstate^2)$ floating-point operations (flops) during
the online stage \cite{ito1998reduced,prud2002reliable}. Similar
low-complexity results can be obtained for the other terms if they can be
similarly expressed in separable form. However, affine
parameter dependence is a quite limiting scenario and
does not generally hold.
\end{remark}

\subsection{Complexity reduction}\label{sec:existingMethods}
Several techniques have been developed to mitigate the computational
bottleneck described in Section \ref{sec:bottleneck}.  Before applying
projection, these 
methods compute (or sample) only a few entries of the vector-valued functions;
other entries are not computed. In effect,
this complexity-reduction strategy is equivalent to employing a sparse
left-projection test basis. Such methods have been successfully applied to ODEs
that do not exhibit particular structure.  However, when applied to mechanical
systems described by Lagrangian mechanics, these techniques destroy Lagrangian
structure.

\subsubsection{Collocation}\label{sec:coll}
Collocation approaches \cite{astrid2007mpe,ryckelynck2005phm} compute only a
subset of the full-order equations of motion
\eqref{eq:sdmEom} before applying Galerkin projection. That is, the
reduced-order equations of motion \eqref{eq:lagrangeGal} are approximated by
 \begin{align} \label{eq:lagrangeGalColl}
 \begin{split}
 \podstate^T\sampleMat \sampleMatT\Mparam\podstate\ddot\qred
 +\podstate^T\sampleMat \sampleMatT\Cparam\podstate\dot\qred+
 \podstate^T\sampleMat \sampleMatT\nabla_\q
 \potEn(\qRef+\podstate\qred;\param)&\\
  = \podstate^T\sampleMat \sampleMatT \fextparamRomRef&.
  \end{split} 
  \end{align} 
Recall that the sampling matrix $\sampleMat$ consists of $\nsample$ selected columns of
the identity matrix. If one considers the matrix $\sampleMat
\sampleMatT\podstate$ as defining a basis for a test space,
Eq.~\eqref{eq:lagrangeGalColl} can be viewed as a Petrov--Galerkin projection.

Computing the components of Eq.~\eqref{eq:lagrangeGalColl} is inexpensive in
the case of $H$-independence, i.e., when
the matrices $\M$,  $\C$, $\nabla_{\q\q}\potEn$, $\nabla_\q \fext$, and
$\nabla_{\dot \q }\fext$ are sparse.  To see this, consider the first term in
Eq.~\eqref{eq:lagrangeGalColl}:  computing
$\left(\podstate^T\sampleMat\right)
\left(\sampleMatT\MparamOn\right)\podstate$ for specific online point
$\paramOnline\in\paramDomain$ incurs $\mathcal O(\nsample\sparsity\nstate +
\nsample\nstate^2)$ flops when operations are carried out in the order implied
by the parentheses.  This cost is small if the sparsity measure of $\M$ is
small, i.e., $\sparsity\ll N$.
 
However, this cost-reduction approach destroys the problem's structure, as it
does not preserve the following Lagrangian properties described in Section
\ref{sec:prob}:
 \begin{enumerate} 
	\item[\ref{prop:metric}.] The approximated reduced mass matrix $\podstate^T\sampleMat \sampleMatT\Mparam\podstate$ is not
	symmetric, so it does not define a metric.
	\item[\ref{prop:potential}.]
	The term 
 $\podstate^T\sampleMat \sampleMatT\nabla_{\q\q}\potEn(\qRef+\podstate\qred;\param)\podstate$
 is not symmetric, so it cannot be the Hessian of a potential-energy function.
 \item[\ref{prop:diss}.]The
approximated reduced damping matrix  $\podstate^T\sampleMat \sampleMatT\Cparam\podstate$ is
not symmetric, so it does not derive from a dissipation function.
	 \end{enumerate}
Note that Property \ref{prop:config} is trivially satisfied, as the
configuration space can be described as $\Qred = \RR{\nstate}$ and relates to
the original configuration space by Eq.\ \eqref{eq:relate}. Further, Property
\ref{prop:ext} is satisfied, because the non-conservative forces can be
derived by applying the  Lagrange--D'Alembert variational principle to a
modified external force $\sampleMat \sampleMatT\fextparam$, but restricted
to variations in the (true) configuration space $\Qredfull$.
\subsubsection{DEIM/gappy POD}\label{sec:eim}
Methods based on the discrete empirical interpolation method
 \cite{chaturantabut2010journal,galbally2009non,drohmannEOI} or
gappy POD \cite{sirovichOrigGappy,CarlbergGappy,carlbergJCP} approximate via
least-squares regression or interpolation
the nonlinear vector-valued functions appearing in
Eq.~\eqref{eq:sdmEom}; these include $ \Mparam\ddot \q $, $\Cparam\dot \q $,
$\nabla_\q  \potEnParam$, and $\fextparam$. Because these approaches
construct a separate approximation for each term in the governing equations,
they often achieve higher accuracy than collocation.

During the offline stage, these methods construct an orthogonal basis $\podf
\in \RR{N \times \nf}$ with $\nf \le \nsample$ for each nonlinear
function
$\f(t;\param)$ appearing in the equations of motion.
The basis $\podf$ can be computed empirically via proper
orthogonal decomposition (POD), in which case the approximation technique is
referred to as `gappy POD' \cite{sirovichOrigGappy}. This consists of two
steps: 1) collect snapshots $\snaps{\f} =\{\f(t;\param)\ | \
t\in\tTrain(\param),\ \param\in\paramTrain\}$, where
$\tTrain(\param)\subset\timeDom$ designates the time instances taken by the
time-integration method for the training simulation; and 2) compute
$\podf$ by Algorithm \ref{PODSVD} of Appendix~\ref{app:POD} using $\snaps{\f}$
and an energy criterion $\energyCrit_\f\in\left[0,1\right]$ as inputs. 

During the online stage, these methods approximate the nonlinear function as
        \begin{equation}\label{eq:gappyApprox}
\f(t;\param)\approx \podf  [ \sampleMatT \podf ]^+ \sampleMatT \f(t;\param)
        \end{equation}
where a superscript $+$ denotes the Moore--Penrose pseudoinverse and $ [
\sampleMatT \podf ]^+ \sampleMatT \f $ is simply the solution to the linear least-squares
 problem 
\begin{equation}\underset{\f_r\in\RR{\nf}}{\mathrm{minimize\ }} \|\sampleMatT \f -
  \sampleMatT \podf \f_r \|_2^2 . 
\end{equation} 
Notice that when $\nf = \nsample$, the least-squares residual is zero
(assuming the $\sampleMatT \podf$ has full column rank) and so  the
above procedure corresponds to interpolation.

As with collocation, this approximation technique leads to computational-cost
savings during the online stage if computing $\sampleMatT\f(t;\param)$ incurs
a flop count independent of $N$, i.e., $\f(t;\param)$ exhibits
$H$-independence.  Substituting  least-squares approximations
for the nonlinear functions into Eq.~\eqref{eq:lagrangeGal} yields the
	approximated reduced-order equations of motion
\begin{equation}
  \label{eq:lsqEq}
  \leftGappy{\MEmptyparam
	\ddot\qred}\Mparam\podstate \ddot\qred+
	\leftGappy{\CEmptyparam\dot\qred}\Cparam\podstate \dot\qred
	+ \leftGappy{\nabla_\q  \potEnEmptyParam}\nabla_\q 
	\potEn(\qRef+\podstate\qred;\param) =
	\leftGappy{\fextEmptyparam}\fextparamRom.
\end{equation}
Here, we have used the notation
\begin{equation}
\leftGappy{\fEmptyparam}\equiv 
\podstate^T\podArgsSub{\fEmptyparam}\left[\sampleMatT\podArgsSub{\fEmptyparam}\right]^+
\sampleMatT,
\end{equation}
and the subscript of $\leftGappy{}$ and $\podArgsNo$ denotes the function
for which the approximation has been constructed.

Unfortunately, this approximation method also
destroys the Lagrangian structure. As before, Lagrangian properties
\ref{prop:metric}--\ref{prop:diss} are lost because the reduced mass, stiffness, and
damping matrices are not symmetric.
However, Property \ref{prop:config} is preserved. Property \ref{prop:ext} is
also preserved, because the non-conservative external force can be derived by the
Lagrange--D'Alembert principle applied to the modified external force
$\tildefextparam =
\podfn\left[\sampleMatT\podArgsSub{\fext}\right]^+\sampleMatT\fextparam$ with
variations restricted to the configuration space $\Qredfull$.

\section{Efficient, structure-preserving model reduction}\label{sec:newMethod}

The main idea of the proposed approach is to directly approximate the
quantities defining the Lagrangian structure of the Galerkin-projection
reduced-order model, and subsequently derive the equations of motion. Section
\ref{sec:galerkin} enumerates these quantities for the simple mechanical
systems considered herein: the Riemannian metric $\metricRed$,
the potential-energy function $\potEnRed$, the semidefinite dissipation
function $\dissFunRedNo$, and the external force $\fextred$.  Approximations
to these ingredients should
1) preserve salient properties, 2) lead to
computationally inexpensive reduced-order-model simulations, and 3) incur
minimal approximation error.

To this end, we propose a model defined by
\begin{enumerate} [I.]
  \item a configuration space $\Qred = \RR{\nstate}$, which
	relates to the original coordinates by  Eq.\ \eqref{eq:relate},
	\item an approximated Riemannian metric $\metricRedApprox$,
	\item an approximated potential-energy function
	$\potEnRedApprox$,
	\item an approximated positive-semidefinite dissipation function
$\dissFunRedApproxNo$, and
	\item\label{ing:fextredApprox} an approximated external force $\fextredApprox$ derived from applying the
	Lagrange--D'Alembert principle to an approximated force $\fextApprox$ represented in
	the original coordinates, but
	limited to variations in the reduced configuration space $\Qredfull$. 
	\end{enumerate}

\noindent We can derive the equations of motion by applying
the forced Euler--Lagrange equation with these approximations:
 \begin{equation} \label{eq:eomApprox}
 \frac{d}{dt}\nabla_{\dot \qred}\Lredapprox (\qred,\dot \qred;\param)-
 \nabla_{\qred}\Lredapprox(\qred,\dot \qred;\param) +
 \nabla_{\dot \qred}\dissFunRedApprox{\dot \qred}=\fextredApproxparam,
  \end{equation} 
	where the approximated Lagrangian is defined as
 \begin{equation} \label{eq:redLagApprox}
 \tildeLred(\qred,\dot \qred;\param) \equiv\half\metricRedApprox(\dot\qred,\dot\qred;\param)
 - \potEnRedApprox(\qred;\param).
  \end{equation} 
	Note that Eq.\ \eqref{eq:eomApprox} approximates Eq.\
	\eqref{eq:reducedOrderEOM}, while Eq.\ \eqref{eq:redLagApprox} approximates
	Eq.\ \eqref{eq:redLagrangian}.

Figure \ref{fig:structurePres} depicts the strategy graphically.
 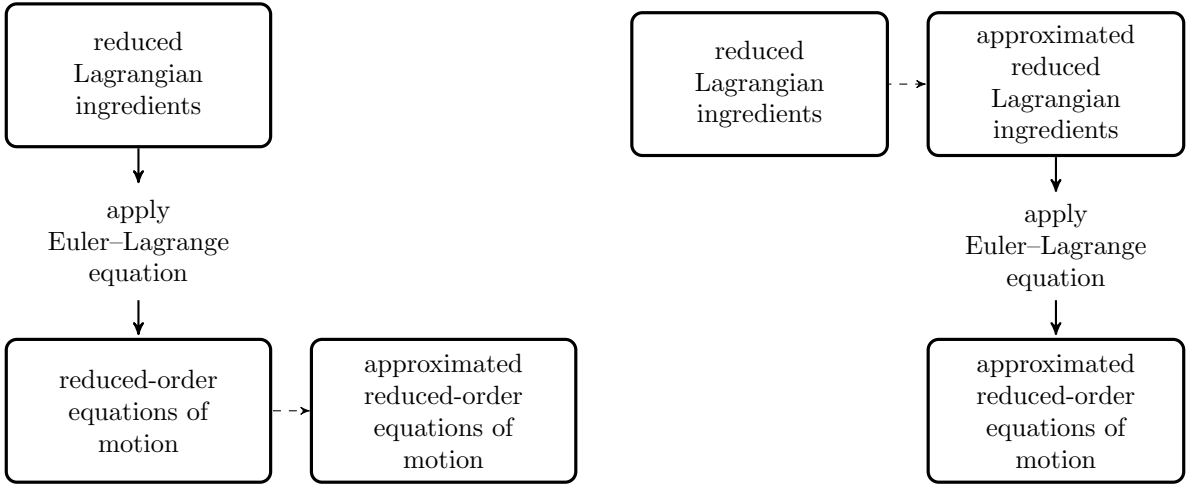
\begin{figure}[htbp] 
  \centering 
\begin{center}
\subfigure[\textbf{Existing complexity-reduction methods} (see Section
\ref{sec:existingMethods}). By approximating the equations of motion, such
methods destroy Lagrangian structure.]{
	\begin{tikzpicture}
  [node distance=.5cm,
  start chain=going below,,inner sep=5pt]
	 \node[punktchain, join, text width=9em,minimum height=1.9cm] (ingred) {reduced\\ Lagrangian\\ ingredients};
		\node[approxMeth,join, text width=9em](eulerLag){apply\\ Euler--Lagrange\\ equation};
		\node[punktchain,join, text width=9em,minimum height=1.9cm,minimum height=1.9cm](eom){reduced-order\\ equations of\\ motion};
		\begin{scope}[start branch=venstre,
			every join/.style={->, dashed, shorten <=0.3pt}, ]
			\node[punktchain, on chain=going right, join=by {->}, text width=9em,minimum height=1.9cm]
					(proposed) {approximated\\ reduced-order\\ equations of\\ motion};
		\end{scope}
	\end{tikzpicture}
}
\quad
\subfigure[\textbf{Proposed approach}. By approximating Lagrangian ingredients before
deriving the equations of motion, the approach preserves Lagrangian structure.]{
	\begin{tikzpicture}
  [node distance=.5cm,
  start chain=going below,]
	 \node[punktchain, join, text width=9em,minimum height=1.9cm] (ingred) {approximated\\ reduced\\ Lagrangian\\ ingredients};
		\begin{scope}[start branch=venstre,
			every join/.style={->, dashed, shorten <=0.3pt}, ]
			\node[punktchain, on chain=going left, join=by {<-}, text width=9em,minimum height=1.9cm]
					(proposed) {reduced\\ Lagrangian\\ ingredients};
		\end{scope}
		\node[approxMeth,join, text width=9em](eulerLag2){apply\\ Euler--Lagrange\\ equation};
		\node[punktchain,join, text width=9em,minimum height=1.9cm](approxEOM){approximated reduced-order\\ equations
		of\\ motion};
	\end{tikzpicture}
}
	\end{center}
\caption{Comparing existing complexity-reduction approaches with the proposed approach. A
dashed arrow implies a complexity-reduction approximation.} 
\label{fig:structurePres} 
\end{figure} 
 The next sections describe two proposed methods that align
with this strategy for structure preservation.
For reference, Table \ref{tab:methods} reports components of the equations of
motion for these methods, as well as for the model-reduction methods discussed
in the previous sections. 
 \begin{table}[htd] 
 \centering 
 \scriptsize{\begin{tabular}{|c|c|c|c|c|c|c|} 
  \hline 
\multirow{2}{*}{method} & mass & damping & potential-energy & external &
struct. & low\\
& matrix & matrix & gradient & force & pres.?& cost?\\
\hline
 Galerkin& $\podstate^T\Mparam\podstate$ & $\podstate^T\Cparam\podstate$ &
 $\podstate^T\nabla_\q  \potEn(\initialQ+\podstate\qred;\param)$ &
 $\podstate^T\fext $& yes& no\\
	& & & & & &\\
collocation & $\podstate^T\sampleMat \sampleMatT\Mparam\podstate$ &
$\podstate^T\sampleMat \sampleMatT\Cparam\podstate$ &
$\podstate^T\sampleMat \sampleMatT\nabla_\q 
\potEn(\initialQ+\podstate\qred;\param)$ &
$\podstate^T\sampleMat \sampleMatT\fext $ & no&yes\\
	& & & & & &\\
 gappy POD& $\leftGappy{\MEmptyparam
	\ddot\qred}\Mparam\podstate $ & $\leftGappy{\CEmptyparam\dot\qred}\Cparam\podstate$ & $\leftGappy{\nabla_\q  \potEnEmptyParam}\nabla_\q 
	\potEn(\initialQ+\podstate\qred;\param)$ &
	$\leftGappy{\fextEmptyparam}\fext$ & no& yes\\ 
	& & & & & &\\
\multirow{2}{*}{proposal 1 }&
\multirow{2}{*}{$\sparseMass^T\Mparam\sparseMass$ }&
$\alpha \sparseMass^T\Mparam\sparseMass +$
 & \multirow{2}{*}{$
\sparsePotEnNo^T\nabla_\q\potEn(\initialQ+\sparsePotEnNo\qred;\param)$ }&
\multirow{2}{*}{$\leftGappy{\fextEmptyparam}\fext$}& \multirow{2}{*}{yes
}& \multirow{2}{*}{yes}\\
&  &
$\beta\sparsePotEnNo^T\nabla_{\q\q}\potEn(\initialQ;\param)\sparsePotEnNo$
 & &
&  & \\
	& & & & & &\\
\multirow{2}{*}{proposal 2 }&
\multirow{2}{*}{$\sum\limits_{i=1}^\nM
\Mcoeffi(\param)\podstate^T\Mbasis{i}\podstate $}&
$\alpha \sum\limits_{i=1}^\nM
\Mcoeffi(\param)\podstate^T\Mbasis{i}\podstate +$
 & \multirow{2}{*}{$
\sparsePotEnNo^T\nabla_\q\potEn(\initialQ+\sparsePotEnNo\qred;\param)$ }&
\multirow{2}{*}{$\leftGappy{\fextEmptyparam}\fext$}& \multirow{2}{*}{yes
}& \multirow{2}{*}{yes}\\
&  &
$\beta\sparsePotEnNo^T\nabla_{\q\q}\potEn(\initialQ;\param)\sparsePotEnNo$
 & &
&  & \\
\hline
  \end{tabular} }
  \caption{Terms appearing in the equations of motion for various model-reduction techniques, including the two proposed structure-preserving methods.} 
  \label{tab:methods} 
  \end{table} 
 
\subsection{Riemannian-metric approximation
$\metricRedApprox$}\label{sec:matApprox}

The function $\metricRed:\left(\vVec _r,\w _r;\param\right) \mapsto \vVec _r^T\podstate^T\Mparam\podstate
	\w _r$ is defined by a low-dimensional symmetric
positive-definite matrix $\podstate^T\Mparam\podstate$. Thus, the task of
approximating this matrix is consistent with problem (P\ref{probMat}) of
Section \ref{sec:matrixApprox}; we therefore propose computing an approximated
Riemannian metric
$\metricRedApprox:\RR{\nstate}\times\RR{\nstate}\times\paramDomain\rightarrow
\RR{}$
as 
\begin{align} 
\label{eq:metricRedApprox}\metricRedApprox(\vVec _1,\vVec
_2;\param) \equiv&\vVec _1^T\matApproxM\vVec _2,
\end{align} 
where $\matApproxM$ is an $\nstate\times\nstate$ matrix that must be
symmetric and positive definite. The first method
(proposal 1 in Table \ref{tab:methods}) employs the reduced-basis
sparsification technique, i.e., it approximates $\matApproxM$  via
Eq.~\eqref{eq:spApproxSparse}. The second method (proposal 2 in
Table \ref{tab:methods}) employs matrix gappy POD and approximates this
matrix by Eq.~\eqref{eq:spApprox}. Procedures \ref{procedure:rbs} (Section
\ref{sec:sparseMat}) and \ref{procedure:matrixGappy} (Section
\ref{sec:lsrecon}) provide the offline and online steps to
implement these approximations.

\subsection{Potential-energy-function approximation
$\potEnRedApprox$}\label{sec:podEnRedApprox}

Noting that only $\nabla_\qred\potEnRedApprox $ appears in the reduced-order
equations of motion (see Eqs.\ \eqref{eq:eomApprox}--\eqref{eq:redLagApprox}),
problem (P\ref{probPotEn}) of Section \ref{sec:potEnGen} applies to this
scenario, and so we approximate the potential energy according to the method
described in that section. In particular, Eq.~\eqref{eq:potEnSparse} defines
the approximated reduced potential energy. Further, we set the reference
configuration to equilibrium $\qRefNo=\initialQNo$ to avoid the limitations
associated with other choices (see the discussion in Section
\ref{sec:potEnGen}). Procedure \ref{procedure:rbsPotEn}
of Section \ref{sec:procedure:rbsPotEn} describes the offline/online decomposition for implementing this
approximation.

\subsection{Dissipation-function approximation
$\dissFunRedNo$}\label{sec:diss}
To maintain the Rayleigh-damping structure, we simply approximate the damping
matrix as a linear combination of the approximated mass matrix and Hessian of
the potential at
equilibrium
 \begin{equation} 
\dissFunRedApprox{\vVec} = \frac{\rayleighM}{2}\vVec^T\matApproxM\vVec +
\frac{\rayleighK}{2}\vVec^T\nabla_{\qred\qred}\potEnRedApprox(0;\param)\vVec,
  \end{equation} 
	where $\rayleighM$ and $\rayleighK$ are the Rayleigh damping coefficients
	defined in Section \ref{sec:prob}.

\subsection{External-force approximation
$\fextredApprox$}\label{sec:extForceApp}
The following form of the approximated external force preserves 
structure, i.e., ensures it is derived from applying the
	Lagrange--D'Alembert principle to an approximated force $\fextApprox$ limited to variations in the reduced configuration space $\Qredfull$:
 \begin{equation} \label{eq:extForceApprox}
	\fextredApproxparam\equiv \podstate^T\fextApproxparam.
  \end{equation} 
Thus, the task of generating this
approximation can be reduced to computing $\fextApproxparam$---an approximation to the vector-valued function
$\fextparam$. That is, we assign no special mathematical properties to $\fext$
aside from the fact that it is a vector. One way to accomplish this is by the
DEIM/gappy POD approach described in 
Section \ref{sec:eim}. 


The error in this approximation 
can be bounded using a result derived from the error
in the gappy POD approximation of $\fext$ (e.g., see \cite[Appendix
D]{carlbergJCP}).  We obtain
\begin{equation}
\|\fextredApprox - \fextred\|_2
\leq \|\left(\I - \podArgsSub{\fext}\left[\sampleMatT
\podArgsSub{\fext}\right]^+\sampleMatT\right)\fext \|_2 \leq
\|\Rmat^{-1}\|_2\|\left(\I - \podArgsSub{\fext}\podArgsSub{\fext}^T\right)\fext\|_2,
\end{equation}
where $\podArgsSub{\fext}$ is an orthogonal basis used to represent the
external force, and $\sampleMatT\podArgsSub{\fext} = \Qmat \Rmat$ is the thin QR
matrix factorization. This result assumes that
$\sampleMatT\podArgsSub{\fext}$ has full rank.
Thus, the accuracy of this approximation relies both on the sampling matrix
$\sampleMatT$ and how close $\fext$ is to the range of $\podArgsSub{\fext}$. To
achieve accuracy, we compute $\podArgsSub{\fext}$ via POD, which minimizes
the average value of $\|\left(\I -
\podArgsSub{\fext}\podArgsSub{\fext}^T\right)\fext\|_2^2$ over the training
data.

\subsubsection{Exactness conditions}\label{sec:exactFext}
Exactness conditions are similar to those described in Section
\ref{sec:matGappyExact} for the matrix gappy POD approximation. 
In the general case where $\nsample < N$, if $\fext = 0$, then the approximation is exact, i.e., $\fextredApprox =
\fextred$. If instead $\fext$ has at least one non-zero entry, then sufficient
conditions for an exact approximation are 1)
$\fext\in\range{\podArgsSub{\fext}}$ and 2)
$\sampleMatT\podArgsSub{\fext}$ has full column rank.  The
first of these conditions holds, for example, when $\podArgsSub{\fext}$ is computed via
POD, the POD basis is not truncated, $\fext$ is independent of $\q $ and $\dot \q $, $\paramOnline
\in\paramTrain$, and if a snapshot of the external force was collected at the
considered time instance. The second of these can be enforced by 
the method for choosing $\sampleMat$, which is beyond the scope of this paper.
In the full-sampling case where $\nsample = N$, condition 2 holds
automatically, so we only require condition 1 in this case.

\subsubsection{Implementation}
Procedure \ref{procedure:externalForce} provides the offline and online steps for
implementing the external-force approximation.
\begin{algorithm2e}\label{procedure:externalForce}
\caption{\textsf{External-force approximation via gappy POD}}
\BlankLine
\tcm{\textbf{\textsf{Offline stage}}}
\BlankLine
	Collect snapshots of the external force $\snaps{\fext}\equiv\{\fextparam\
	|\ \param\in\paramTrain,\ t \in\tTrain(\param)\}$\;
	Compute a POD basis $\podArgsSub{\fext}$ using Algorithm \ref{PODSVD}
	with inputs $\snaps{\fext}$ and $\energyCrit_{\fext}\in\left[0,1\right]$.\;
	Determine the sampling matrix $\sampleMat$.\;
	Compute the low-dimensional matrix
	$\leftGappy{\fext}
	=\podstate^T\podArgsSub{\fext}\left[\sampleMatT\podArgsSub{\fext}\right]^+$.\;
\setcounter{AlgoLine}{0}
\BlankLine
\tcm{\textbf{\textsf{Online stage}} \textsf{(given $\paramOnline$)}}
\BlankLine
  Compute $\nsample\ll N$ entries of the external force
	$\sampleMatT\fextparamOn$.\;
  Compute the low-dimensional matrix--vector product 
	$\leftGappy{\fext}\left[\sampleMatT\fextparamOn\right]$.\;
\end{algorithm2e}

%

\section{Numerical experiments}\label{sec:experiments}


Although the Galerkin and proposed reduced-order
models have a \emph{theoretical} advantage over the gappy POD and collocation
reduced-order models in terms of preserving Lagrangian
structure, it is unclear if this translates to improved
numerical results in practice. This section investigates this question by applying the
model-reduction techniques to a practical problem: the clamped--free truss
structure shown in Figure \ref{fig:structure}.

\begin{figure}[htbp] \centering
\includegraphics[width=.6\textwidth]{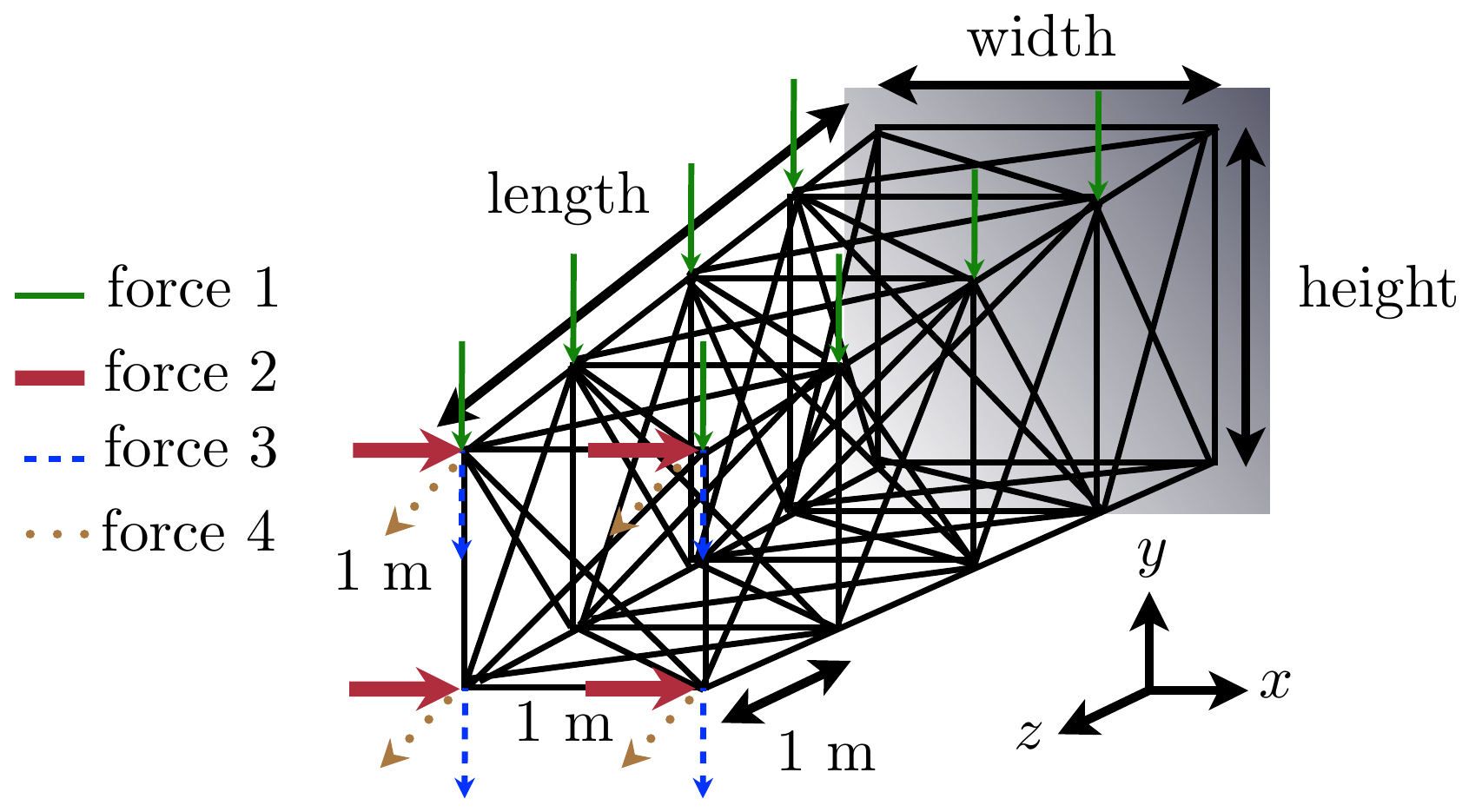}
\caption{Clamped--free parameterized truss structure}
\label{fig:structure}
\end{figure}
We set the material properties to those of aluminum, i.e., density
$\rho = 2700\ \mathrm{kg/m}^3$ and 
elastic modulus $E=62\times 10^9$ Pa.
The external force is composed of four components:
 \begin{equation} 
\forcing(\param,t) = \sum_{i=1}^4\forceTime_i(\param,t)\forceDist_i,
 \end{equation} 
where 
$\forceDist_i\in\RR{N}$, $i=1,\ldots, 4$ correspond to unit loads uniformly distributed
across designated nodes and $\forceTime_i:\paramDomain\times
\left[0,\lastT\right]\rightarrow\RR{}$, $i=1,\ldots, 4$
denote the component-force
magnitudes. Figure
~\ref{fig:structure} depicts the spatial distribution of the forces, which lead
to vectors $\forceDist_i$, $i=1,\ldots, 4$ through the finite-element
formulation described below. The parameterized, time-dependent magnitudes of these forces
are
\begin{gather} 
\forceTime_i(\param,t) = 
\begin{cases}
\forceMag_i\left(\param\right)\sin\left(\lambda_i(\param)\left(t - \lastT/4\right)\right),\quad
t\geq \lastT/4\\
0,\quad \mbox{otherwise}
\end{cases}
,
\end{gather} 
where $\forceMag_i:\paramDomain\rightarrow \RR{}$ and
$\lambda_i:\paramDomain\rightarrow \RR{}$, $i=1,\ldots, 4$ denote the maximum force
magnitudes and force frequencies, respectively.
Similarly, the initial condition is composed of
four components
 \begin{equation} 
\initialCondition = \sum_{i=1}^4\initialCondMag_i(\param)\initialCond_i,
 \end{equation} 
 where $\initialCond_i$ is the steady-state displacement of the truss subjected
 to load $\forceDist_i\forceMag_i\left(\paramNom\right)$ with $\paramNom =
 \left(0,\ldots,0\right)$ denoting the nominal point in parameter space. The equilibrium
 configuration is simply the undeformed truss represented by $\initialQ = 0$;
 thus, the configuration space is centered at equilibrium.

The truss is parameterized by 16 parameters
$\param\equiv\left(\param_1,\ldots,\param_{16}\right)\in\left[-1,1\right]^{16}$
that affect the geometry, initial condition, and applied force as described in
Table \ref{tab:parameterization}.



 \begin{table}[htd] 
 \centering 
 \scriptsize{\begin{tabular}{|c|c|c|c|c|c|c|} 
  \hline 
\multirow{3}{*}{length (m)} & bar & \multirow{3}{*}{width (m)}& \multirow{3}{*}{height (m)}&initial
condition
& external-force & external-force  \\
& cross-sectional & & & max magnitude (N)
&magnitude & frequency
 \\
&area ($\mathrm{m}^2$) & & &$\initialCondMag_i$ & $\forceMag_i$ &$\lambda_i$ \\
\hline
$200 + 50\param_1$ & $0.0025(1 + 0.5\param_2)$ & $10(1 + \param_3)$ & $10(1 +
\param_4)$&$\fNomi(1
+ 0.5 \param_{i+4})$ & $\fNomi (1 + 0.5\mu_{i+8})$& $3\omega_0(1 +
0.5\param_{i+12})$ \\
\hline 
  \end{tabular} }
	\caption{Effect of parameters on truss geometry, initial conditions, and
	applied forces. Here, $\fNomi$, $i=1,\ldots,4$ denote the nominal force
	magnitudes (to be specified within each experiment) and $\omega_0$ denotes the
	lowest-magnitude eigenvalue of the structure at the nominal point
	$\paramNom$.}
\label{tab:parameterization}
	\end{table}

The problem is discretized by the finite-element method. The model consists
of sixteen three-dimensional bar elements per bay with three degrees of
freedom per node; this results in 12 degrees of freedom per bay. We consider a
problem with 250 bays, which leads to $N=3\times 10^3$ degrees of freedom in the
full-order model. The bar elements model geometric nonlinearity, which results
in a high-order nonlinearity in the potential-energy $\potEn$. 

This discretization leads to a model that corresponds to a Lagrangian
dynamical system, with configuration manifold $\Q = \RR{N}$, Riemannian metric
$g(\vVec , \w ; \param) = \vVec ^T\Mparam \w $, nonlinear potential-energy
function $\potEn$ and dissipation function $\dissFun{\dot \q }\equiv
\frac{1}{2}\dot \q ^T\Cparam \dot \q$. Here, $\Cparam = \alpha \Mparam + \beta
\nabla_{\q\q}V(0;\param)$ corresponds to Rayleigh damping. Here, $\alpha$ and
$\beta$ are chosen such that the damping ratio is a specified value $\zeta$
for the uncoupled ODEs associated with the
smallest two eigenvalues of the matrix pencil
$\left(\M(\paramNom),\nabla_{\q\q}V\left(0;\paramNom\right)\right)$
 \cite{chowdhury2003computation}.

To numerically solve the Lagrangian equations of motion in the time interval
$[0,\lastT]$ with $\lastT=25$ seconds, we employ the
implicit midpoint rule (a symplectic integrator). This ensures that the
numerical solution will yield symplectic time-evolution maps in the
conservative case. We employ a globalized Newton solver with a More--Thuente
linesearch \cite{poblano} to solve the system of nonlinear algebraic equations
arising at each time step. Convergence of the Newton iterations is declared
when the residual norm reaches $10^{-6}$ of its value computed using a zero
acceleration and the values of the displacement and velocity at the beginning
of the timestep.  The linear system arising at each Newton iteration is solved
directly.

The experiments compare the performance of four reduced-order models:
Galerkin projection (Section \ref{sec:galerkin}), collocation (Section
\ref{sec:coll}), and gappy POD (Section \ref{sec:eim}), and the proposed
structure-preserving methods. All reduced-order models (ROMs) employ the same POD
reduced basis $\podstate$, which is computed by applying Algorithm \ref{PODSVD}
with snapshots of the configuration variables and an energy criterion
$\energyCrit_\q\in\left[0,1\right]$ specified within each experiment. The POD bases $\podf$ employed by the gappy POD
approach (see Section \ref{sec:eim}) are generated in the same way. In all cases, snapshots are only collected for the
first half of the time interval at the training points; as a result, the
second half of the time interval can be considered predictive---even for the training
set. 

Reduced-order models with complexity reduction employ the same sampling matrix
$\sampleMat$, which is generated using GNAT's greedy sample-mesh algorithm
\cite{carlbergJCP}. These models are also implemented using the sample-mesh
concept \cite{carlbergJCP}.  
To solve optimization problems \eqref{eq:optMatchMat}, 
we use the Poblano toolbox for unconstrained optimization \cite{poblano}. The
initial guess for each of these problems is chosen as
$\sampleMat^T\sampleMat \podstate$.  In practice, we always found the
constraints to be inactive at the unconstrained solution to
\eqref{eq:optProblem2}; therefore, this reduces to a linear
least-squares problem that we solve directly.

To compare the performance of the reduced-order models, we will consider the
response quantity of interest to be the $y$-displacement of the bottom-left node of the
end face of the truss in Figure \ref{fig:structure}; we denote this
(parameterized, time-dependent)
quantity by $y\in\RR{}$. The reported
errors will be a normalized 1-norm (in time) of the error in this quantity:
\begin{equation} \label{eq:error}
\mathrm{error} =
\frac{\sum\limits_{t\in\tTrain(\paramOnline)}|y_{\mathrm{ROM}}(t;\paramOnline)-y_{\mathrm{HFM}}(t;\paramOnline)|}{|\tTrain(\paramOnline)|\left(\max\limits_{t\in\tTrain(\paramOnline)}y_{\mathrm{HFM}}(t;\paramOnline)-\min\limits_{t\in\tTrain(\paramOnline)}y_{\mathrm{HFM}}(t;\paramOnline)\right)}.
\end{equation} 
Here, $y_{\mathrm{ROM}}$ denotes the response computed by a reduced-order
model, $y_{\mathrm{HFM}}$ is the high-fidelity `truth' response, and $\tTrain(\paramOnline)\subset\left[0,\lastT\right]$ denotes the
time instances selected by the time integrator for online point
$\paramOnline$.\footnote{We employ this error measure because it is
insensitive to shifts in the average value of the displacement, unlike other
measures such as the average 1-norm.} In addition to the error in
Eq.~\ref{eq:error}, we will compare the speedup achieved by the reduced-order
models, measured as the ratio of the reduced-order-model simulation time to
the full-order-model simulation time. All computations are carried out in
Matlab on a Mac Pro with 2 $\times$ 2.93 GHz 6-Core Intel Xeon processors and
64 GB of memory.
\subsection{Conservative case}\label{sec:conservative}
We first consider the conservative case characterized by zero damping
$\zeta=0$ and no external
forces $\param_i=-2$ for $i=9,\ldots,16$. 
This scenario is particularly interesting, as the
full-order model corresponds to a conservative Lagrangian dynamical system
characterized by energy conservation and symplectic time-evolution maps, and
our method can also be interpreted as preserving
Hamiltonian structure (see Appendix \ref{app:Hamiltonian}).
Because we numerically solve the equations of motion using the implicit
midpoint rule, which is a symplectic integrator, the
numerical solution is also characterized by a sympletic time-evolution map. This will
also hold for reduced-order models that preserve Lagrangian structure,
i.e., the Galerkin reduced-order model and the two proposed techniques (see
Table \ref{tab:methods}). Note also that the dynamics of undamped, unforced structures
are typically quite stiff, which often leaves reduced-order models prone to
instabilities.
As a result, we are free to vary parameters $\param_i$, $i=1,\ldots, 8$,
which affect only the geometry and initial condition. We set the nominal
forces that affect the initial condition to $\fNomnum{1} = \fNomnum{2} = 2\mathrm{kg}
\times 9.81 \mathrm{m/s}^2$ and $\fNomnum{3} = \fNomnum{4} = 0.4 \mathrm{kg} \times 9.81 \mathrm{m/s}^2$.

We first perform a timestep-verification study for the nominal point
$\paramNom$ characterized by $\paramNom_i=0$, $i=1,\ldots 8$ to ensure we employ
an appropriate timestep in the numerical experiments. Results are shown in
Figure \ref{fig:consTimestepVerification}. A timestep size of $\Delta t =
0.008$ seconds yields an observed convergence rate in the time-averaged tip
displacement of $1.98$, which is close to the asymptotic rate of convergence of the
implicit midpoint rule, and an approximated error in the time-averaged tip
displacement using Richardson
extrapolation of $5.16\times 10^{-7}$. We can therefore declare this to be an
appropriate timestep size for the numerical experiments. Further, we note that
the average number of Newton iterations per timestep is $3.15$, so the
geometric nonlinearity in the potential-energy function is significant.

\begin{figure}[htbp]
\begin{centering}
{\includegraphics[width=.45\textwidth]{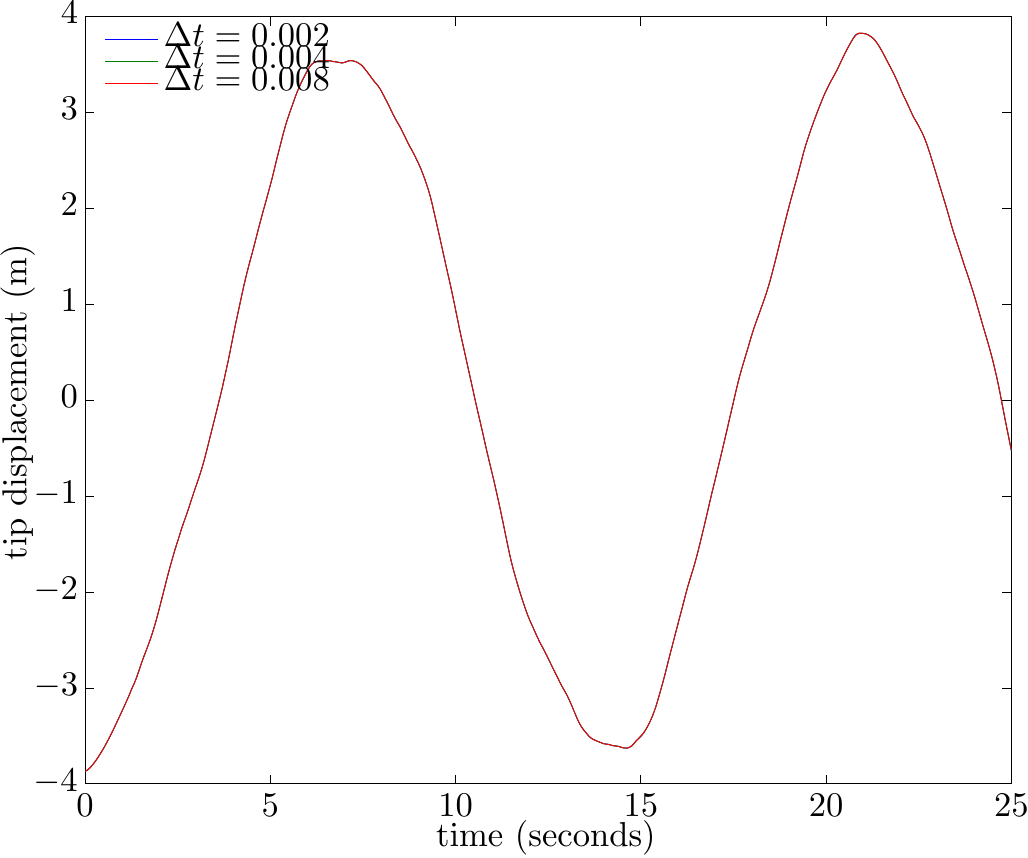}}\quad
{\includegraphics[width=.45\textwidth]{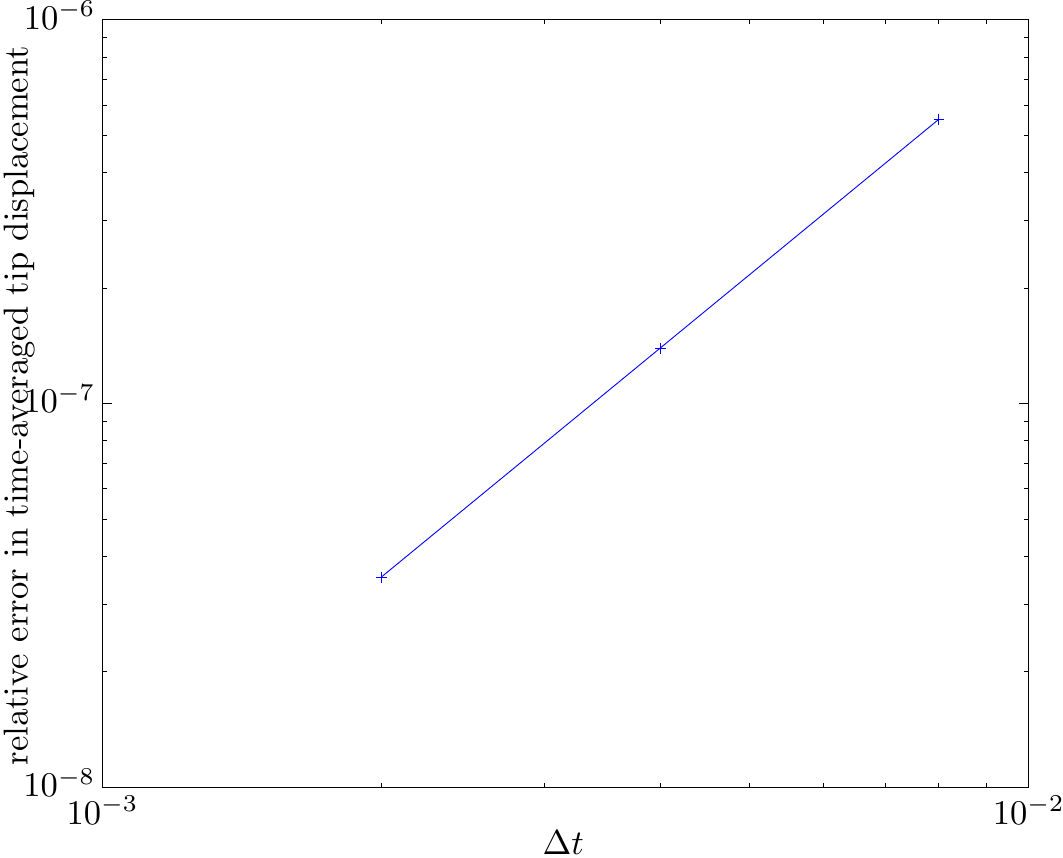}}
\caption{Conservative case: Timestep verification study at the nominal
point. A timestep of $\Delta t = 0.008$ seconds gives a
convergence rate of $1.98$ and approximate relative error of $5.16\times
10^{-7}$. Note that the three responses in the left figure nearly overlap.}
\label{fig:consTimestepVerification}
\end{centering}
\end{figure}

\subsubsection{Fixed parameters}
We now test the model-reduction techniques in a fixed-parameters scenario. That
is, we employ the nominal point in the parameter space for both the training and online
points: $\paramTrain = \paramNom$ and $\paramOnline = \paramNom$. Recall
that we only collect snapshots for the first half of the time interval, so the
second half can be considered a predictive regime. Note that 
the two proposed structure-preserving methods are the same for this case: they
both exactly approximate the reduced mass matrix when the parameters are fixed.

The POD reduced basis $\podstate$ is generated using an energy criterion of
$\energyCrit_\q = 1-10^{-5}$ in
Algorithm \ref{PODSVD} of Appendix \ref{app:POD}; this leads
to a basis dimension of only $\nstate = 11\ll N$. The gappy POD-based
reduced-order model employs an energy criterion of $1$ (i.e., no truncation) for its
reduced bases $\podf$ (see Section \ref{sec:eim}). Figures
\ref{fig:consNonPredict} and \ref{fig:consNonPredictPerformError} report
results for the reduced-order models as the number of sample indices
varies.\footnote{In all response plots, a `flat line' indicates that the
nonlinear solver failed to converge after 500 Newton iterations at three
different time steps.}
	
First, note that the Galerkin reduced-order model is accurate (relative error
of $5.42\%$); however, it yields a speedup of only 1.69.
This is to be expected, as it preserves Lagrangian structure, but has no
complexity-reduction mechanism (see Section \ref{sec:bottleneck}). In addition, the proposed reduced-order
model---which also preserves structure, yet has a complexity-reduction
mechanism---yields a stable and accurate response regardless of the number of
sample points chosen. For example, $0.4\%$ sampling yields a relative error of
7.3\% and a speedup of 207.0. Sampling $2\%$ of the indices yields an error
of $0.71\%$ and a speedup of 34.5, and sampling $5\%$ of the indices leads to
$0.48\%$ error and a speedup of 15.7. Note that sampling beyond $5\%$ does not
improve the method's accuracy; however, it degrades the speedup, as it
requires computing more entries of the vector-valued functions.
	
The other complexity-reducing
reduced-order models (gappy POD and collocation) are \emph{always} unstable
except for collocation in
the full-sampling case, when it is equivalent to Galerkin. This 
clearly highlights the practical benefits of preserving structure in model reduction,
as existing structure-destroying complexity-reduction methods failed, even in the relatively
simple scenario of fixed parameter values.\footnote{We will show in Section
\ref{sec:noncon_fixed} that introducing dissipation improves the performance
of both collocation and gappy POD. Note that gappy POD was also unstable for
other attempted energy criteria of $\energyCrit_\f = 1-10^{-9}$ and
$\energyCrit_\f = 1-10^{-8}$.}

\begin{figure}[htbp]
\centering
\subfigure[0.4\% sampling]
{\includegraphics[width=.32\textwidth]{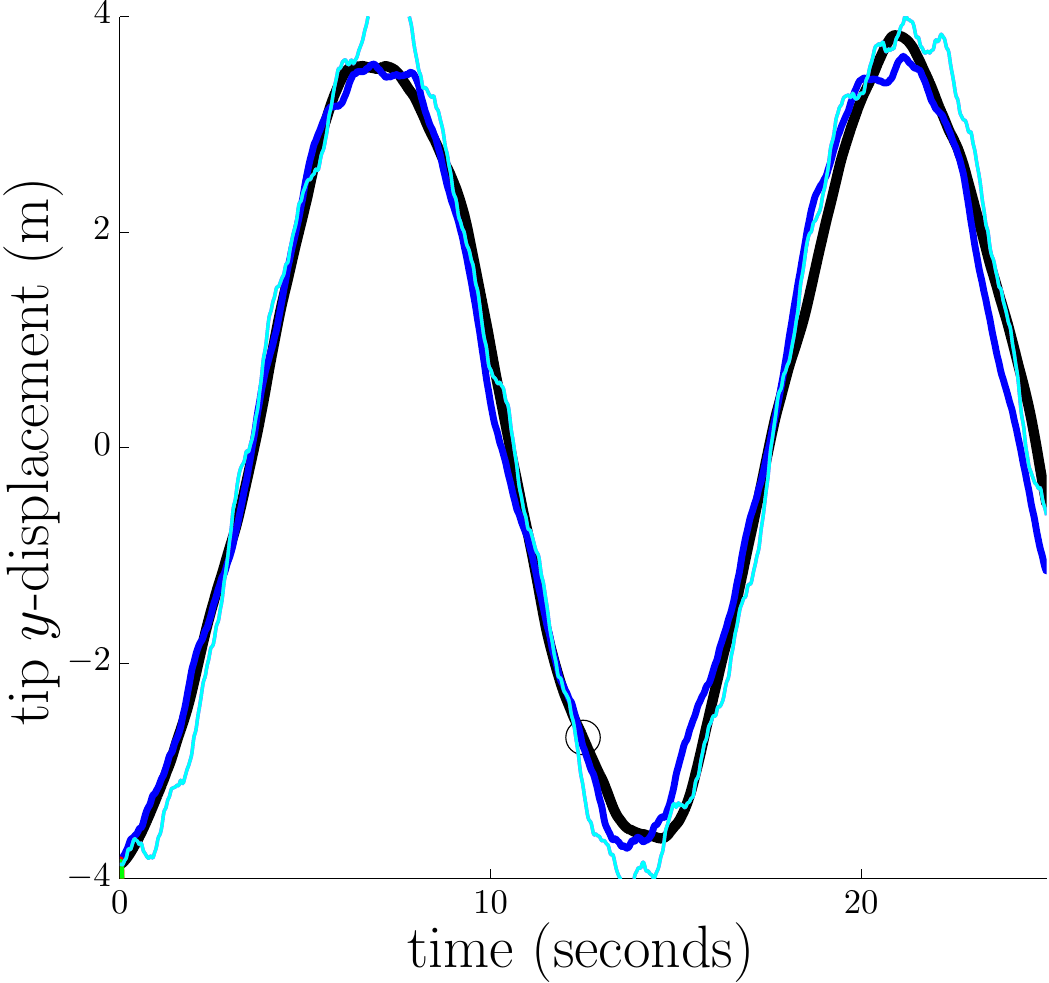}}
\subfigure[2\% sampling]
{\includegraphics[width=.32\textwidth]{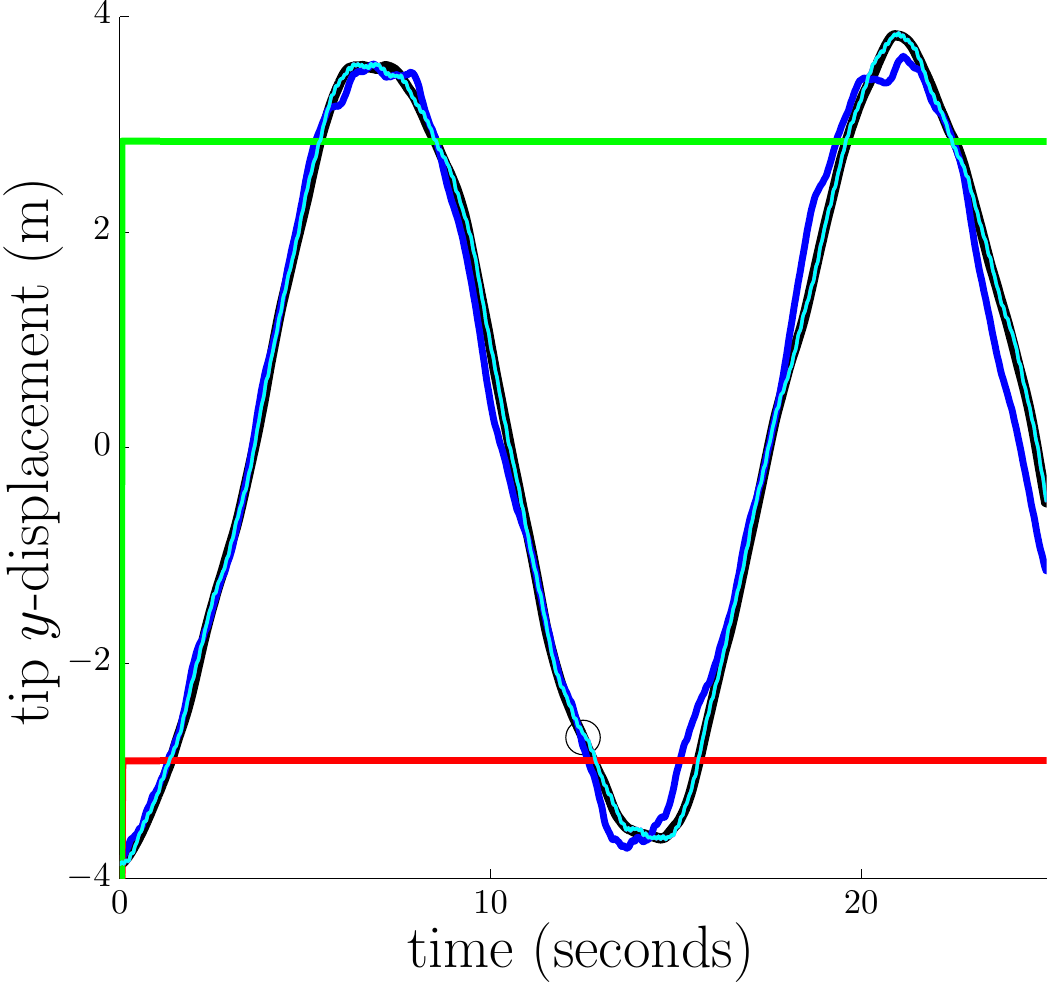}}
\subfigure[100\% sampling]
{\includegraphics[width=.32\textwidth]{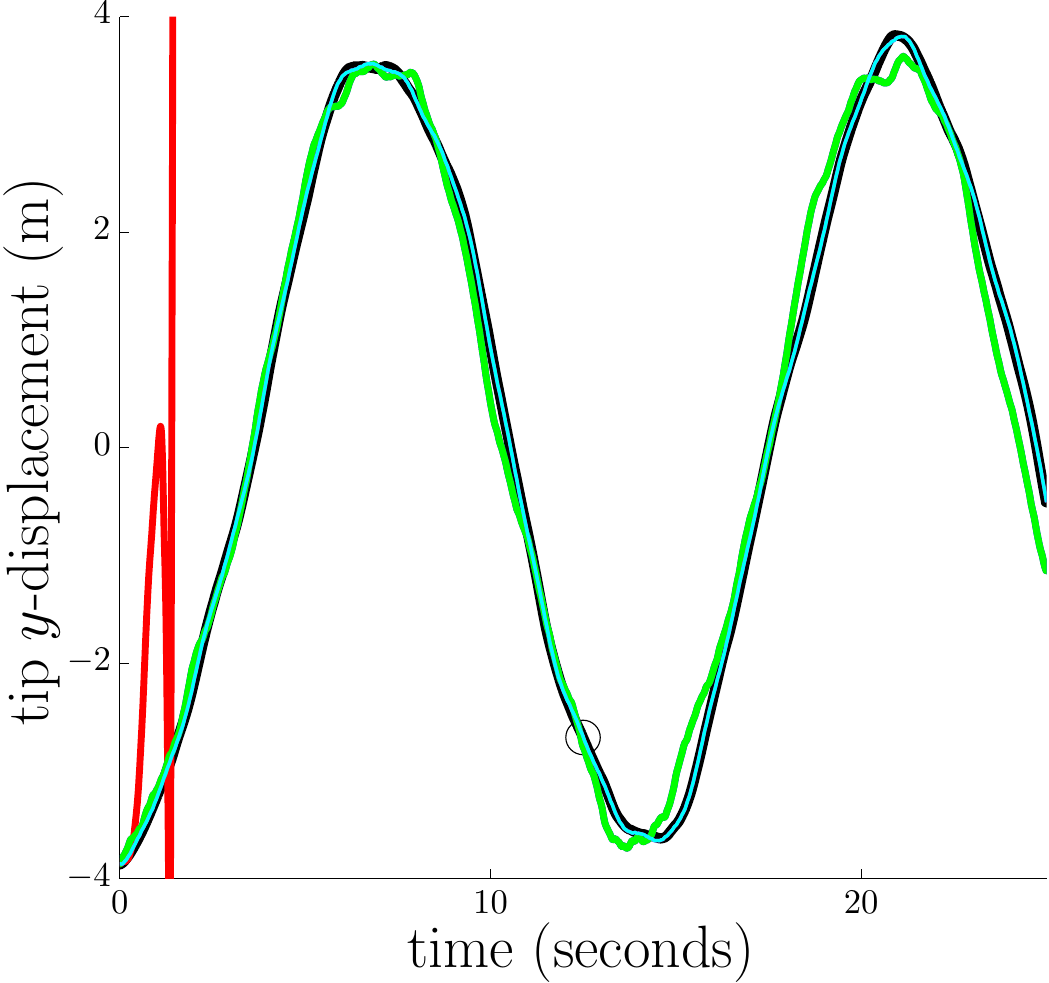}}
\caption{Conservative, fixed-parameters case: reduced-order model responses as
a function of sampling percentage $\nsample/N\times 100\%$. Legend: full-order model (black),
Galerkin ROM (dark blue), structure-preserving ROM (light blue), gappy POD ROM
(red), collocation ROM (green), end of training time interval (black circle).}\label{fig:consNonPredict}
\end{figure}
							  
\begin{figure}[htbp]
\centering
\subfigure
{\includegraphics[width=.48\textwidth]{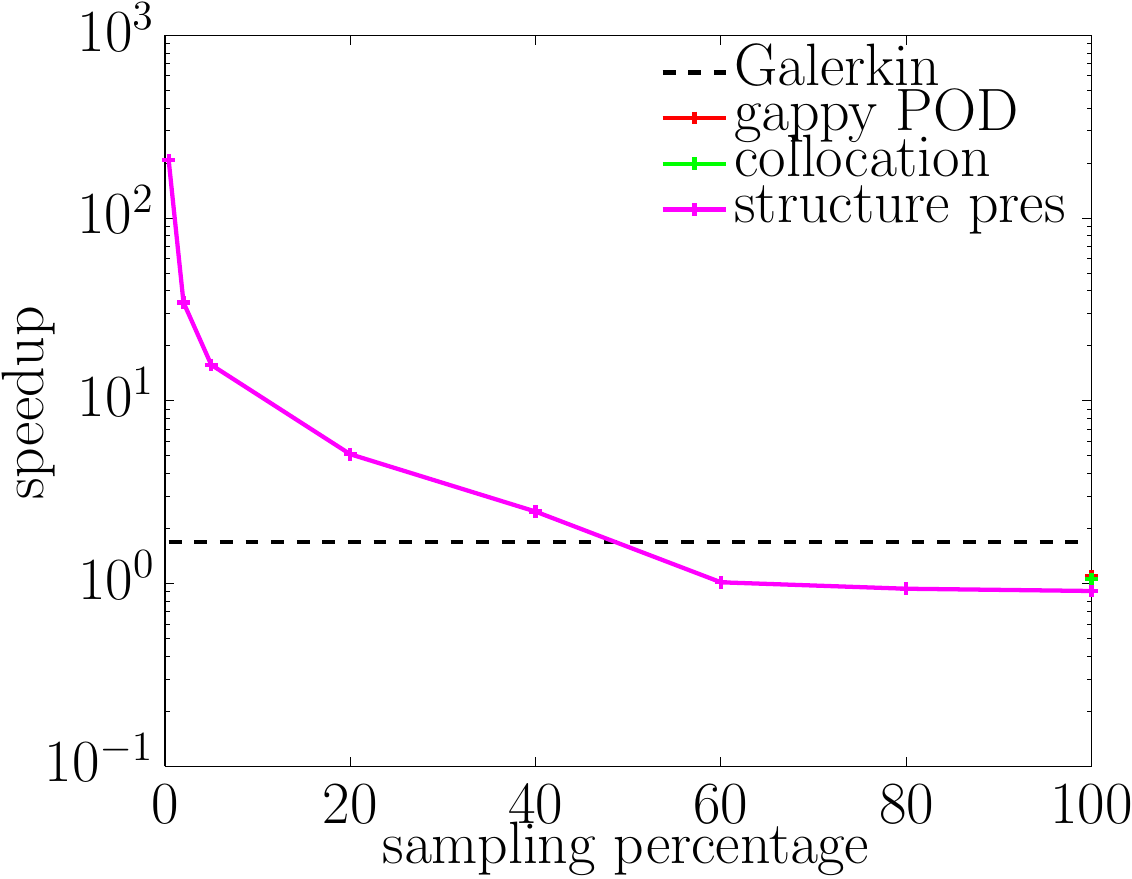}}
\subfigure
{\includegraphics[width=.48\textwidth]{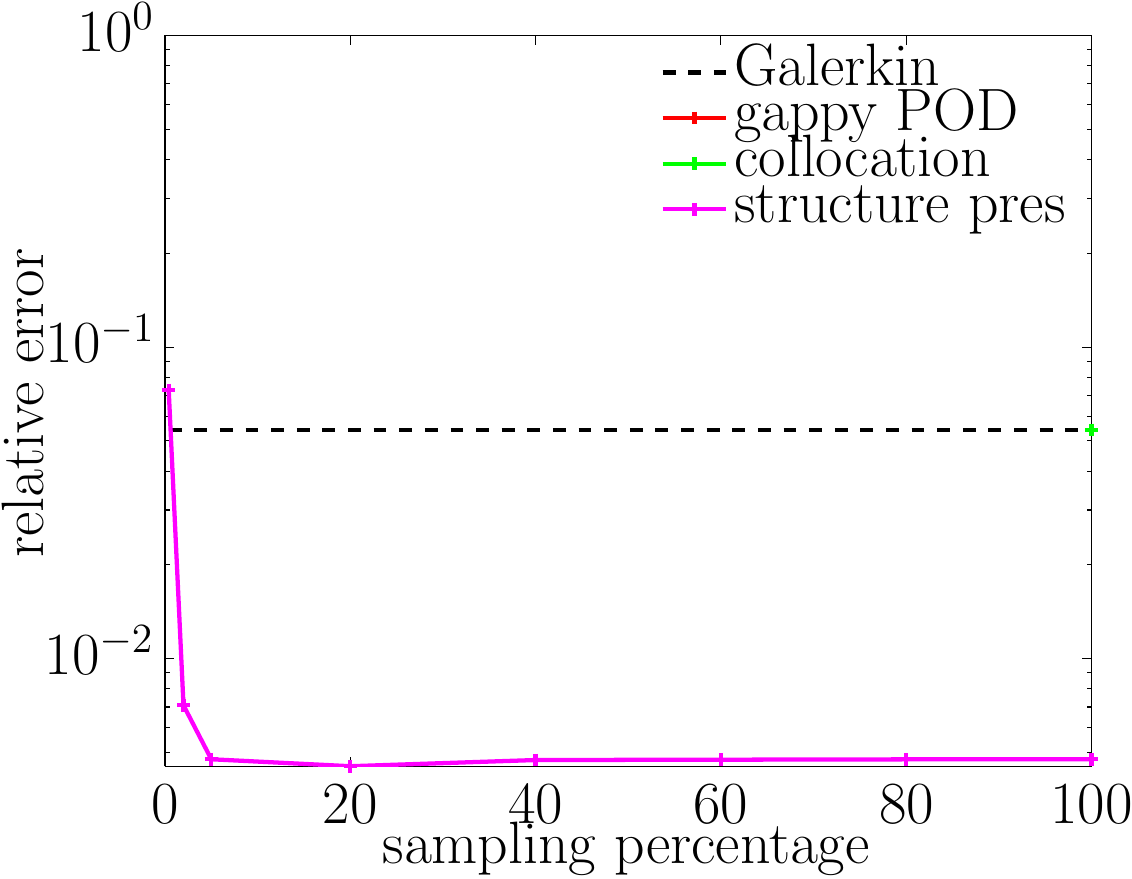}}
\caption{Conservative, fixed-parameters case: reduced-order model performance as
a function of sampling percentage $\nsample/N\times 100\%$. Missing data
points for gappy POD and collocation ROMs indicate unstable responses.}\label{fig:consNonPredictPerformError}
\end{figure}

\subsubsection{Varying parameters}
We now consider a fully predictive scenario with $\paramOnline\not\in \paramTrain$.
We use $\nTrain = 6$ training points and determine $\paramTrain$ using Latin hypercube
sampling. The online points are subsequently chosen randomly in the parameter
space.
Figure \ref{fig:consPredict} depicts the tip displacement for the
training points. Note that the responses are
significantly different from one another. The two proposed
structure-preserving reduced-order models will now be different from one another,
as the parameters are varying, which  means that the parameterized mass matrix
will be approximated differently by the two techniques (see Table
\ref{tab:methods}). \begin{figure}[htbp] 
\centering 
\includegraphics[width=7cm]{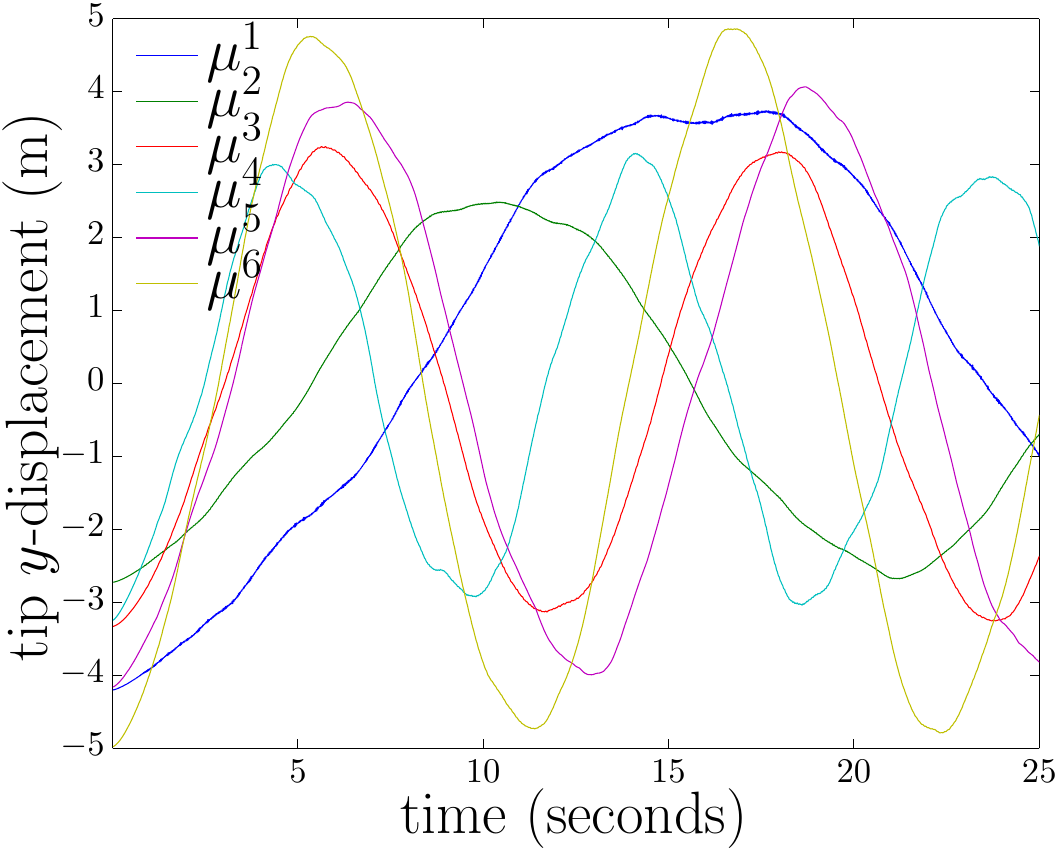} 
\caption{Conservative, varying-parameters case: tip displacement for the training
set $\paramTrain$. } 
\label{fig:consPredict} 
\end{figure} 

The reduced-order models employ a POD reduced basis with a truncation energy
criterion of $\energyCrit_\q = 1-10^{-6}$, which yields a basis dimension of
$\nstate = 147\ll N$. Again, the gappy POD-based reduced-order model employs a
truncation criterion of $\energyCrit_\f = 1$ for its reduced bases. Figure
\ref{fig:consPredictROM} reports the tip displacements generated by the
reduced-order models for the three randomly chosen online points, and
Figure \ref{fig:consPredictPerformError} reports the speedup and errors
achieved by the reduced-order models as a function of the number of sample indices.

Again, note that the Galerkin reduced-order model is stable and accurate, as
it generates relative errors of 18.7\%, 14.5\%, and 9.16\% at the three online
points, respectively. However, it yields discouraging speedups of 0.81
(i.e., the simulation was \emph{slower} than for the full-order model), 1.61,
and 1.32 at these points. The proposed structure-preserving methods are always
stable and quite accurate. They yield nearly the same performance, although
method two (which employs the matrix gappy POD approximation) generates lower
errors for online points with $4.9\%$ sampling.  From Figure
\ref{fig:consPredictROM}, note that the high-frequency oscillations that
characterize the proposed methods' responses are smoothed out when the
sampling percentage reaches 20\%. In particular, proposed method 2 generates
speedups of 15.9, 28.5, and 26.2 and relative errors of 11.6\%, 13.0\%, and
11.6\% for 4.9\% sampling. For 20\% sampling, the method generates speedups of
4.84, 9.82, and 7.72, and relative errors of 1.51\%, 5.83\%, and 1.09\%.

In this example, the gappy POD reduced-order model is unstable for all sampling percentages,
and the collocation reduced-order model is only stable for 100\% sampling (at which point it
is mathematically equivalent to the Galerkin reduced-order model). This is not surprising, as
these methods do not preserve problem structure, nor do they guarantee energy
conservation. This poor performance can be attributed to the stiff dynamics
that characterize the considered conservative Lagrangian dynamical system,
which lead to instabilities for both reduced-order models.

This example strongly showcases the practical importance of preserving Lagrangian
structure: the proposed structure-preserving reduced-order models are the only models
that yield both \emph{fast} and \emph{accurate} results.


\begin{figure}[htbp] \centering
\subfigure[$4.9\%$
sampling]
{\includegraphics[width=.3\textwidth]{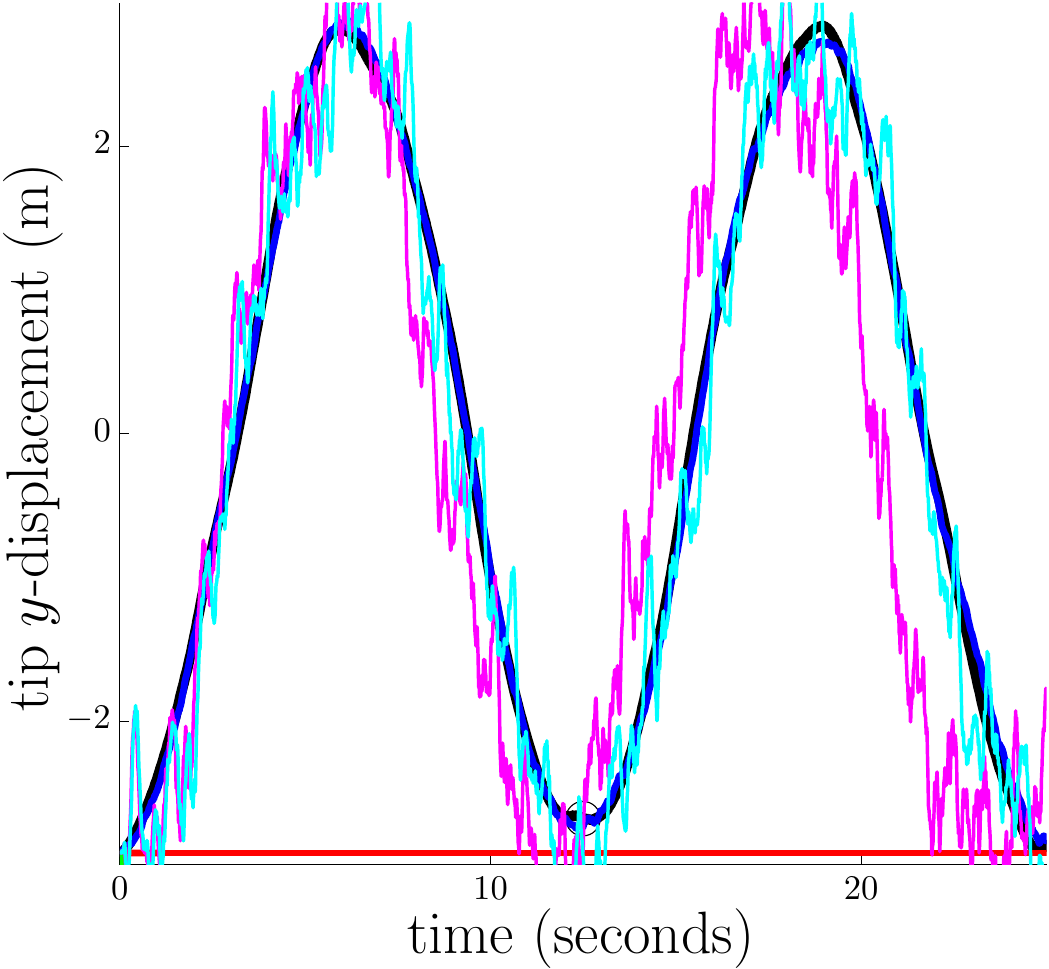}
\includegraphics[width=.3\textwidth]{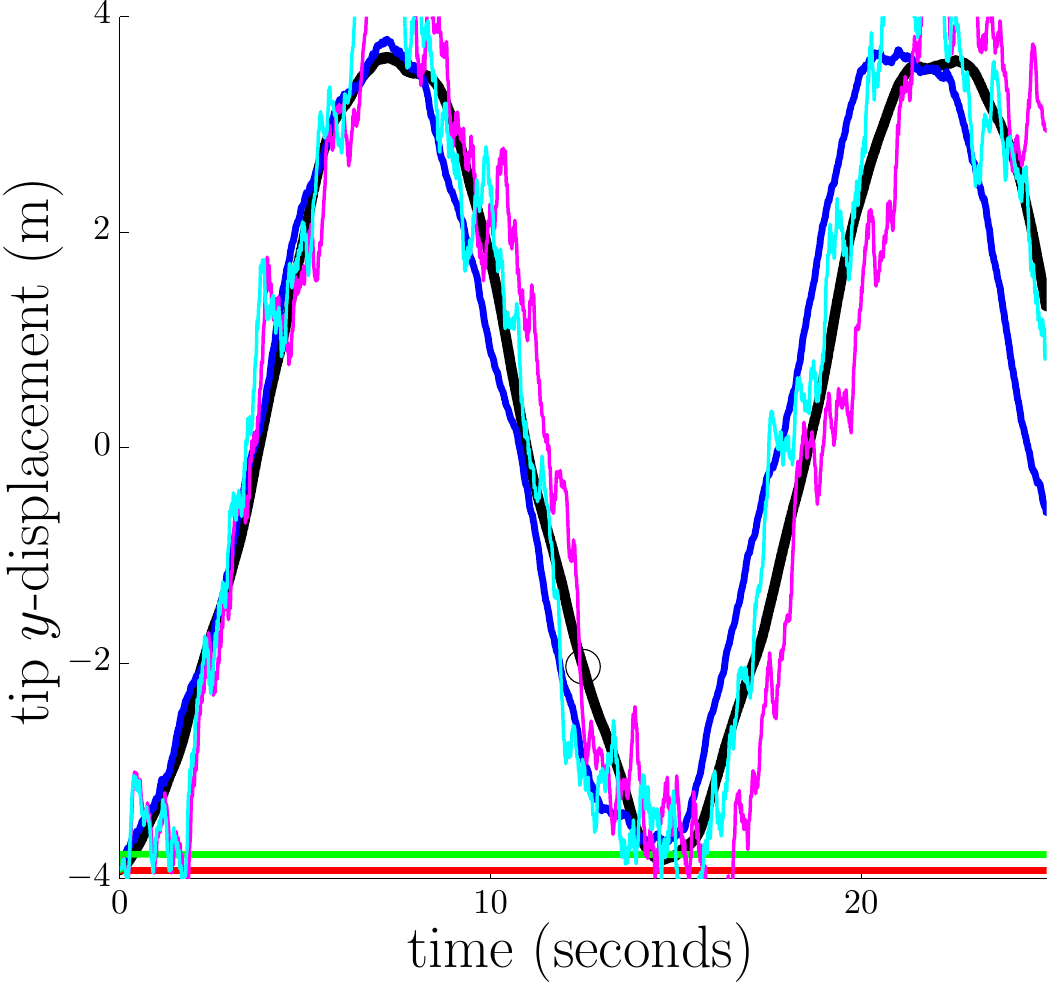}
\includegraphics[width=.3\textwidth]{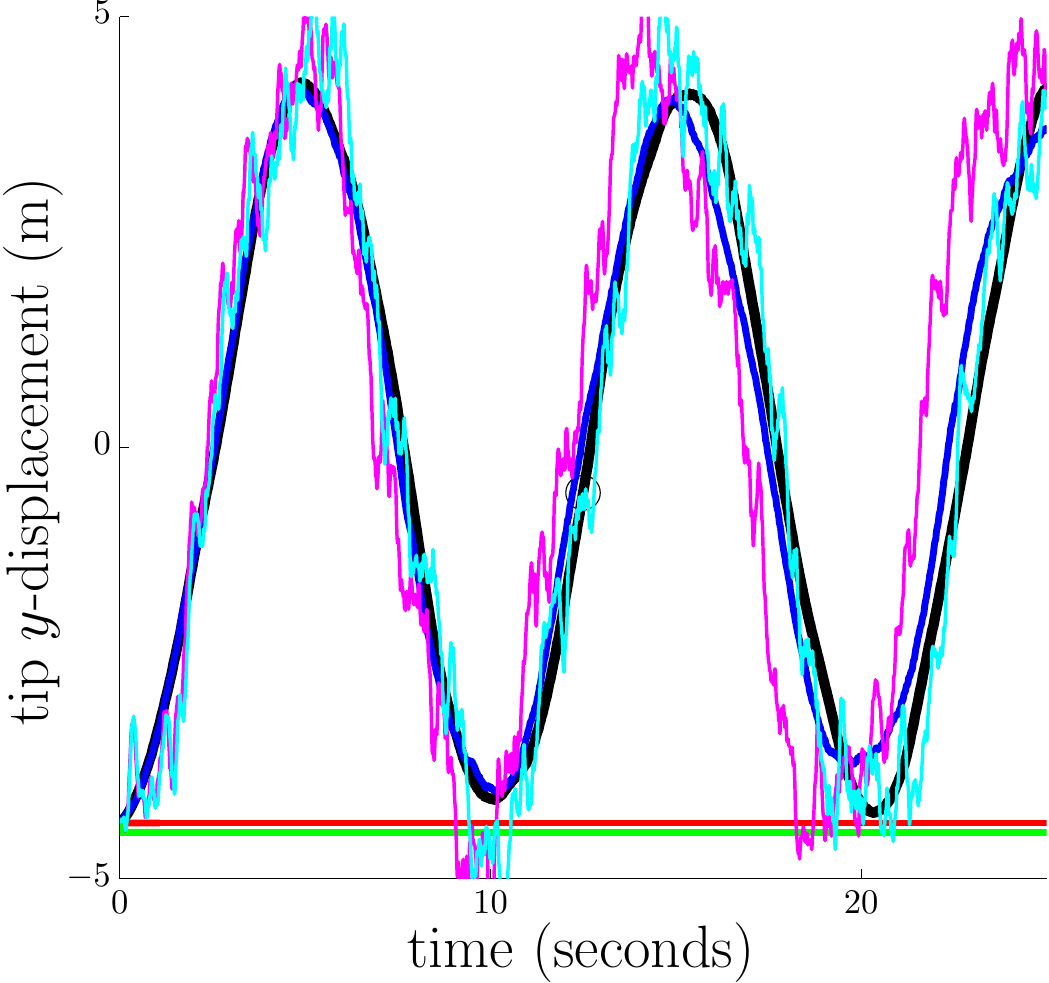}}
\subfigure[$20\%$ sampling]
{\includegraphics[width=.3\textwidth]{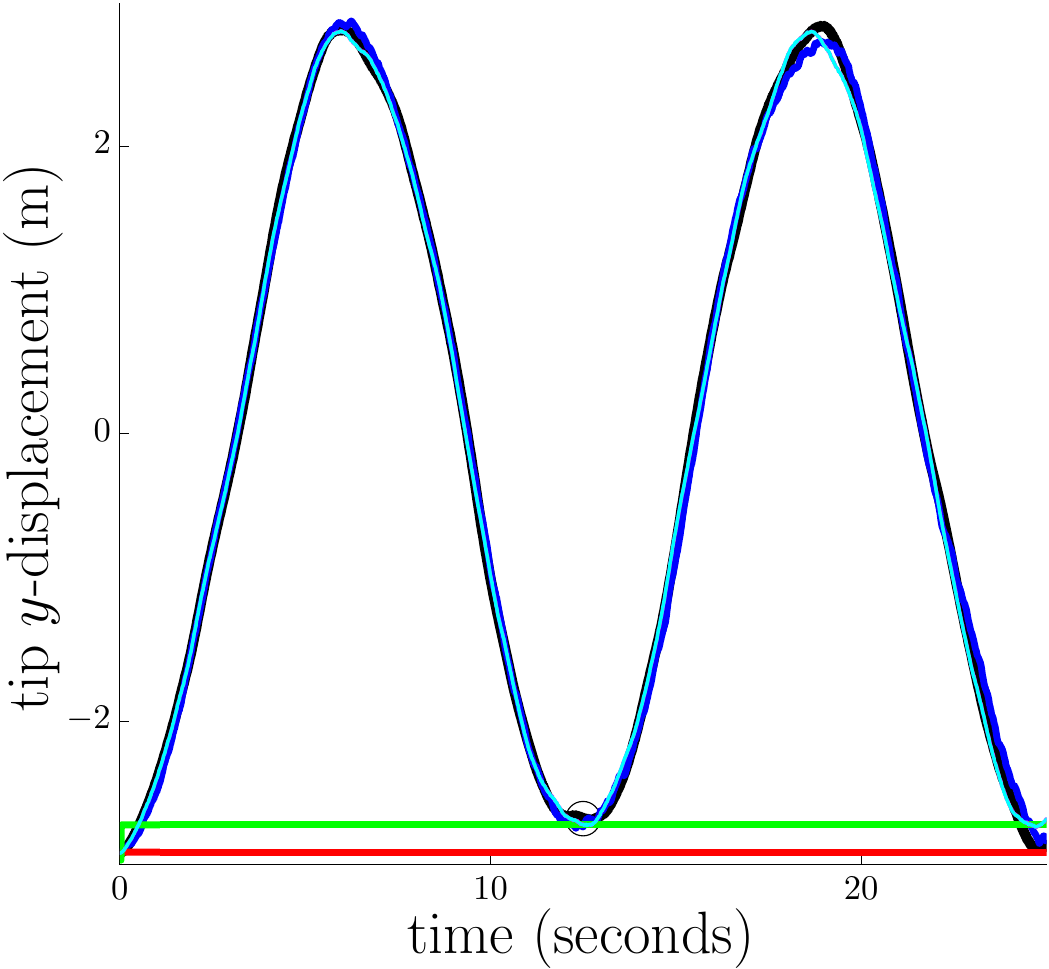}
\includegraphics[width=.3\textwidth]{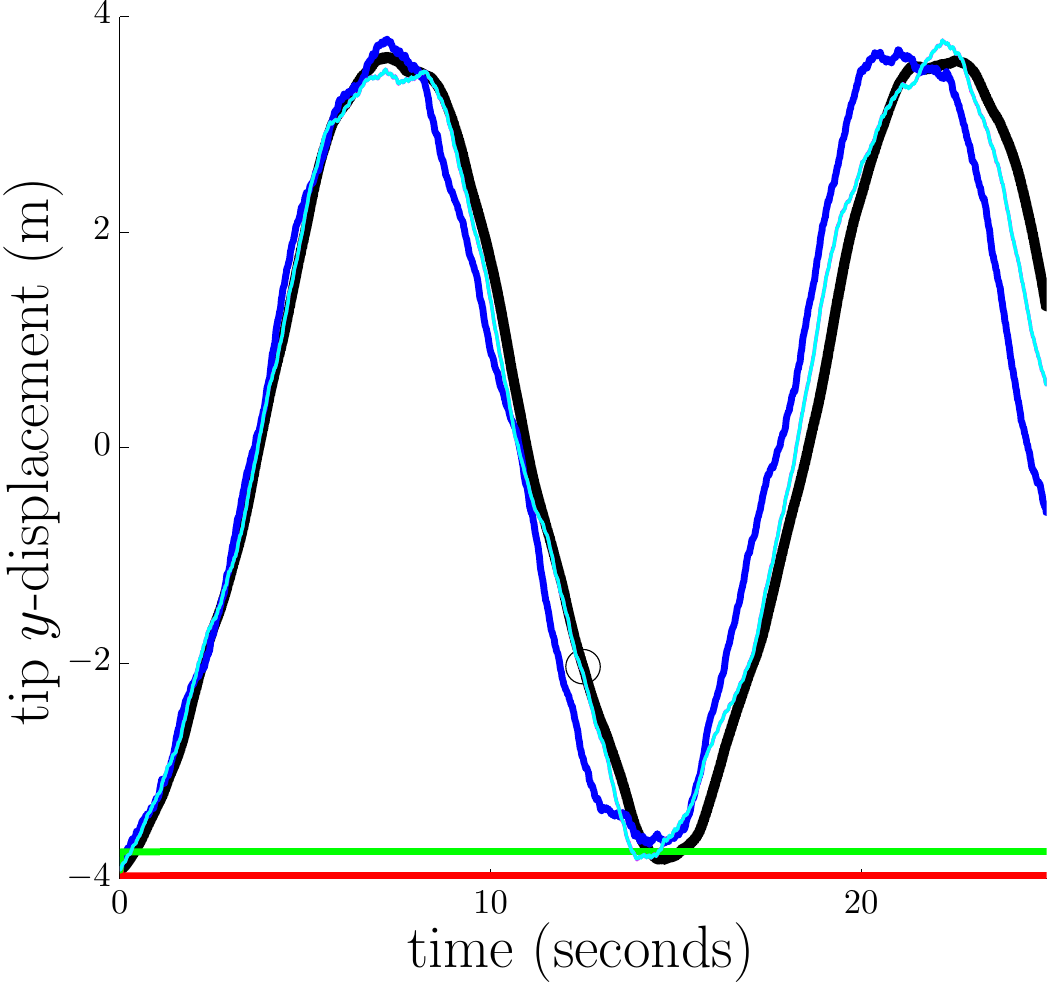}
\includegraphics[width=.3\textwidth]{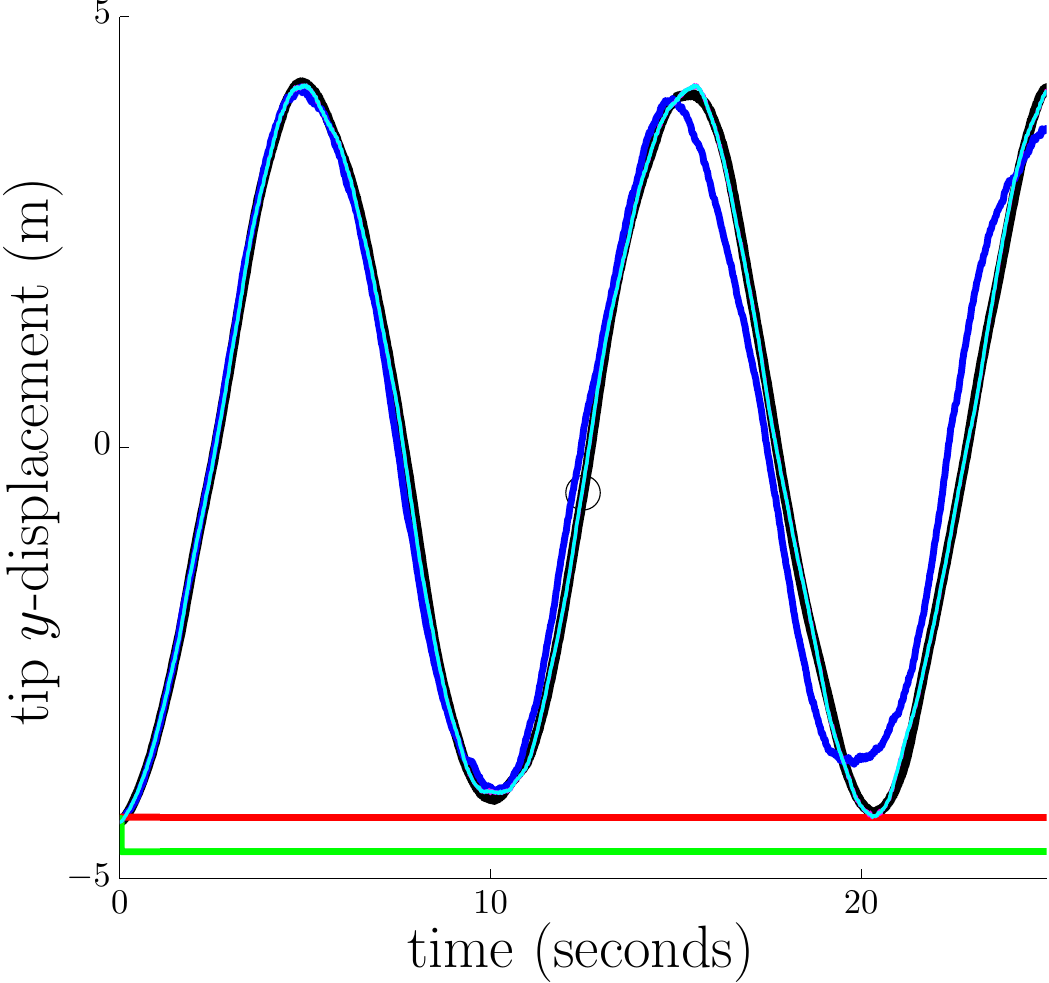}}
\subfigure[$100\%$ sampling]
{\includegraphics[width=.3\textwidth]{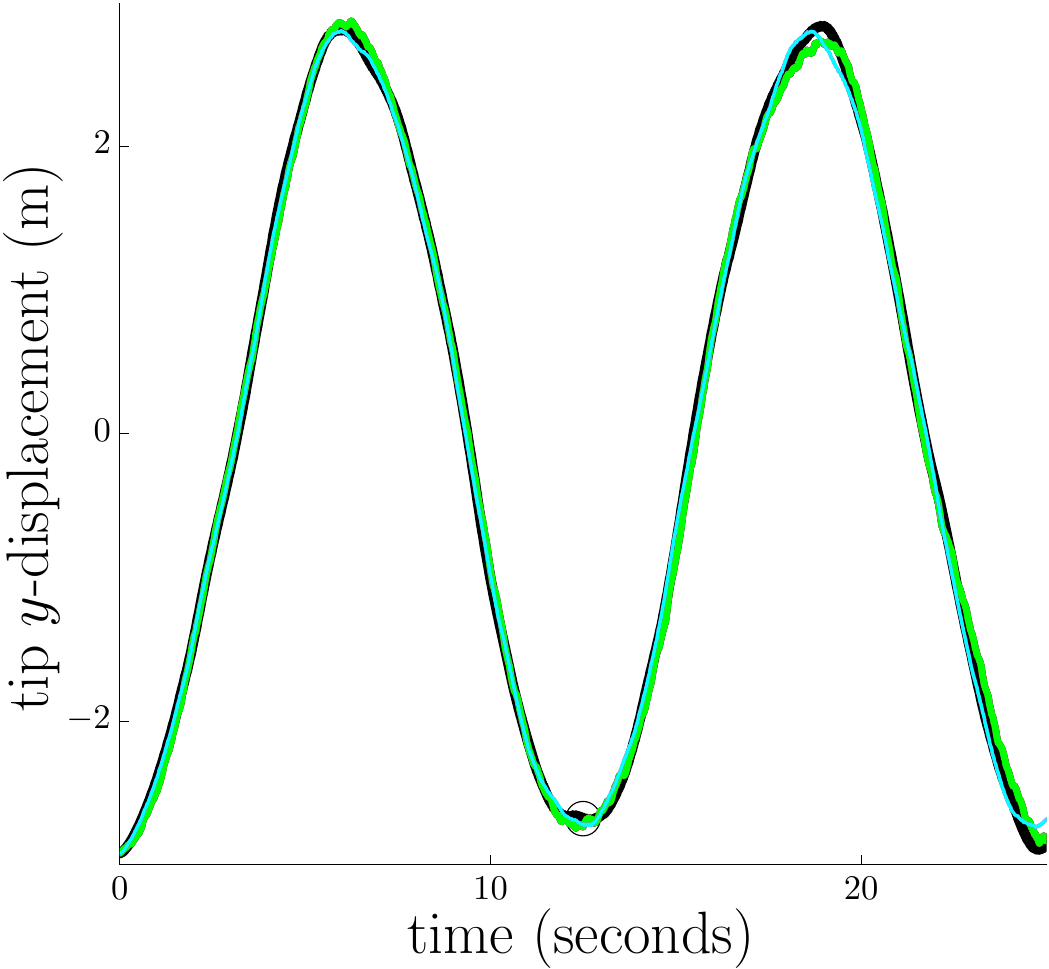}
\includegraphics[width=.3\textwidth]{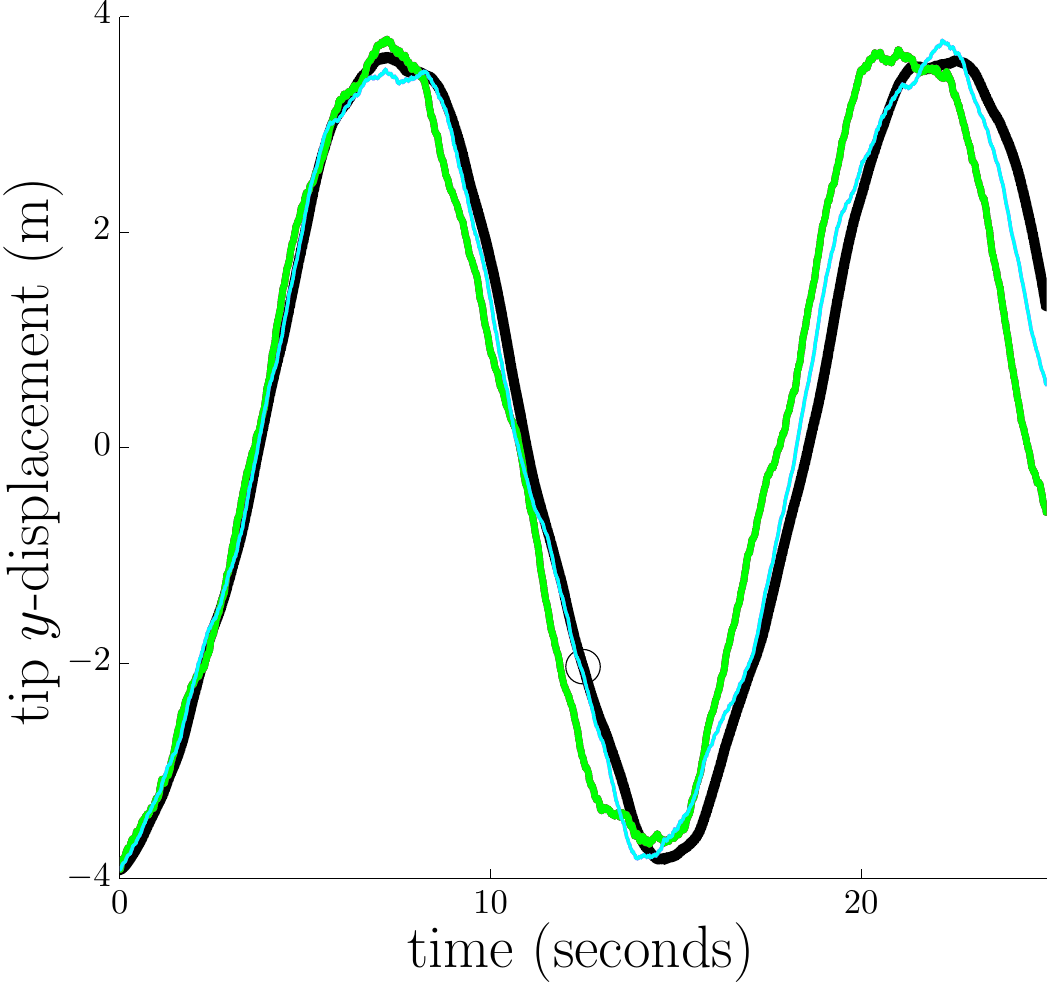}
\includegraphics[width=.3\textwidth]{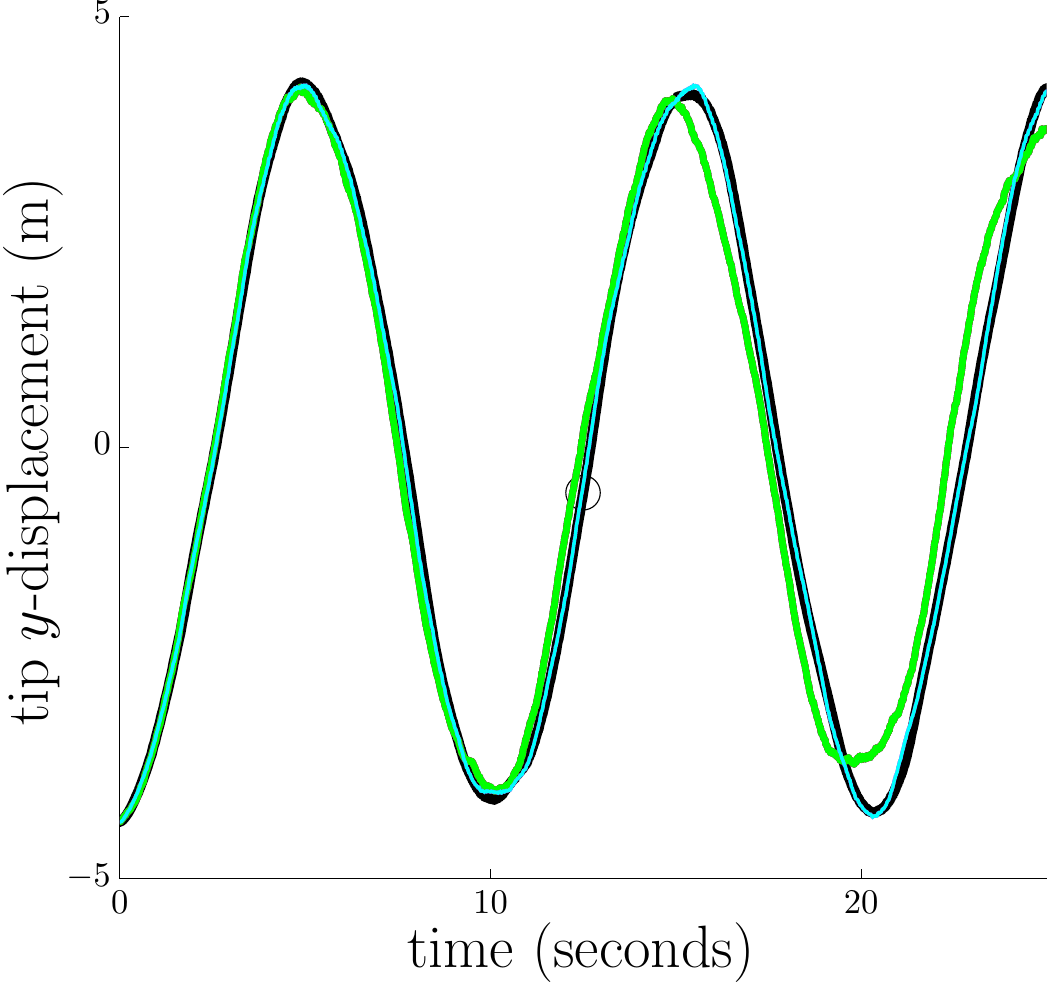}}
\caption{Conservative, varying-parameters case: reduced-order model responses as
a function of sampling percentage $\nsample/N\times 100\%$ for three randomly chosen online points. Legend: full-order model (black),
Galerkin ROM (dark blue), structure-preserving ROM method 1 (magenta), structure-preserving ROM method 2 (light blue), gappy POD ROM
(red), collocation ROM (green), end of training time interval (black circle).}\label{fig:consPredictROM}
\end{figure}

\begin{figure}[htbp]
\centering
\subfigure[online point 1]
{\includegraphics[width=.48\textwidth]{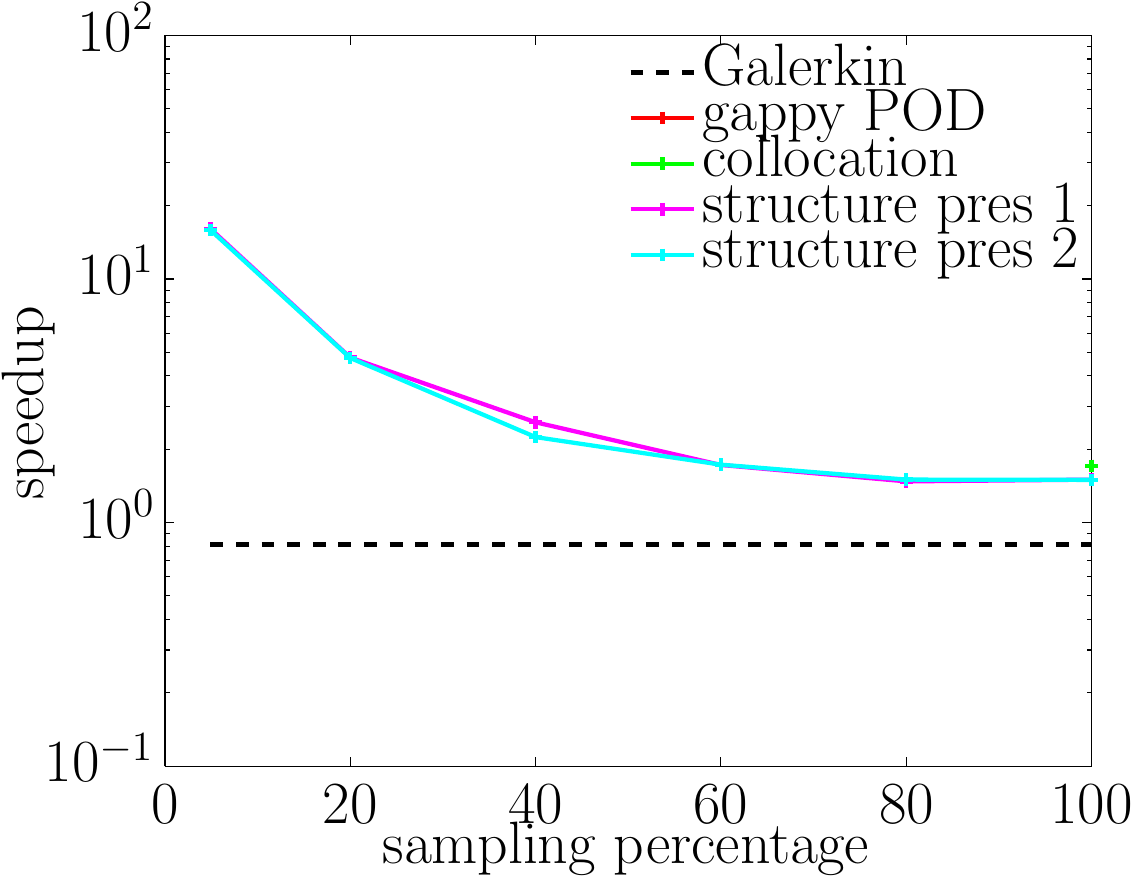}
\includegraphics[width=.48\textwidth]{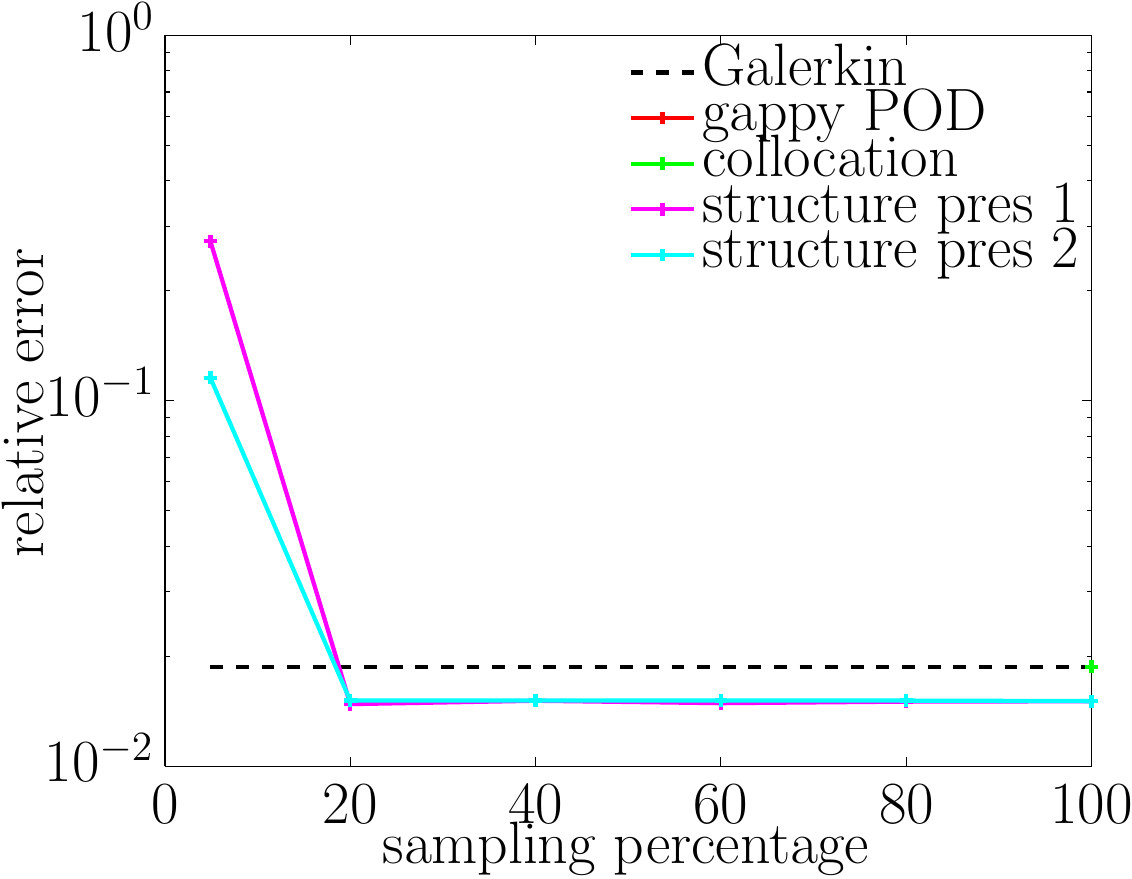}}
\subfigure[online point 2]
{\includegraphics[width=.48\textwidth]{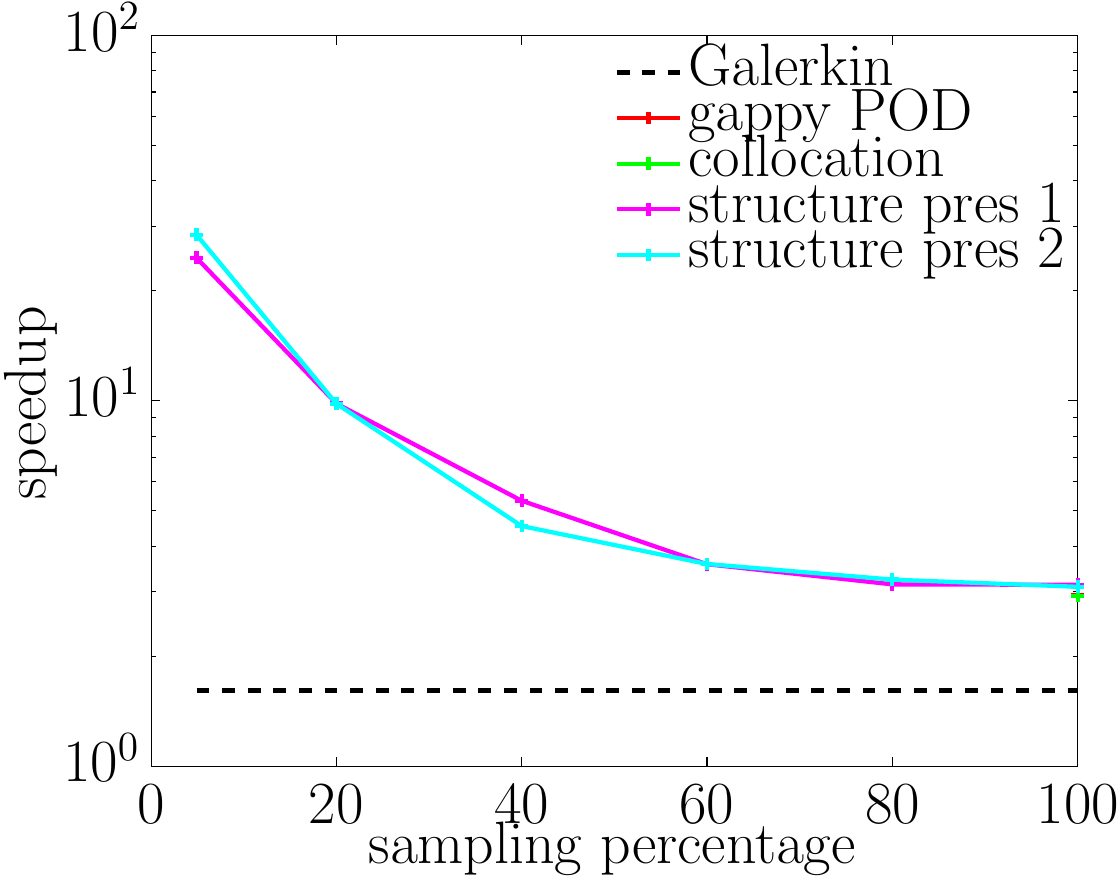}
\includegraphics[width=.48\textwidth]{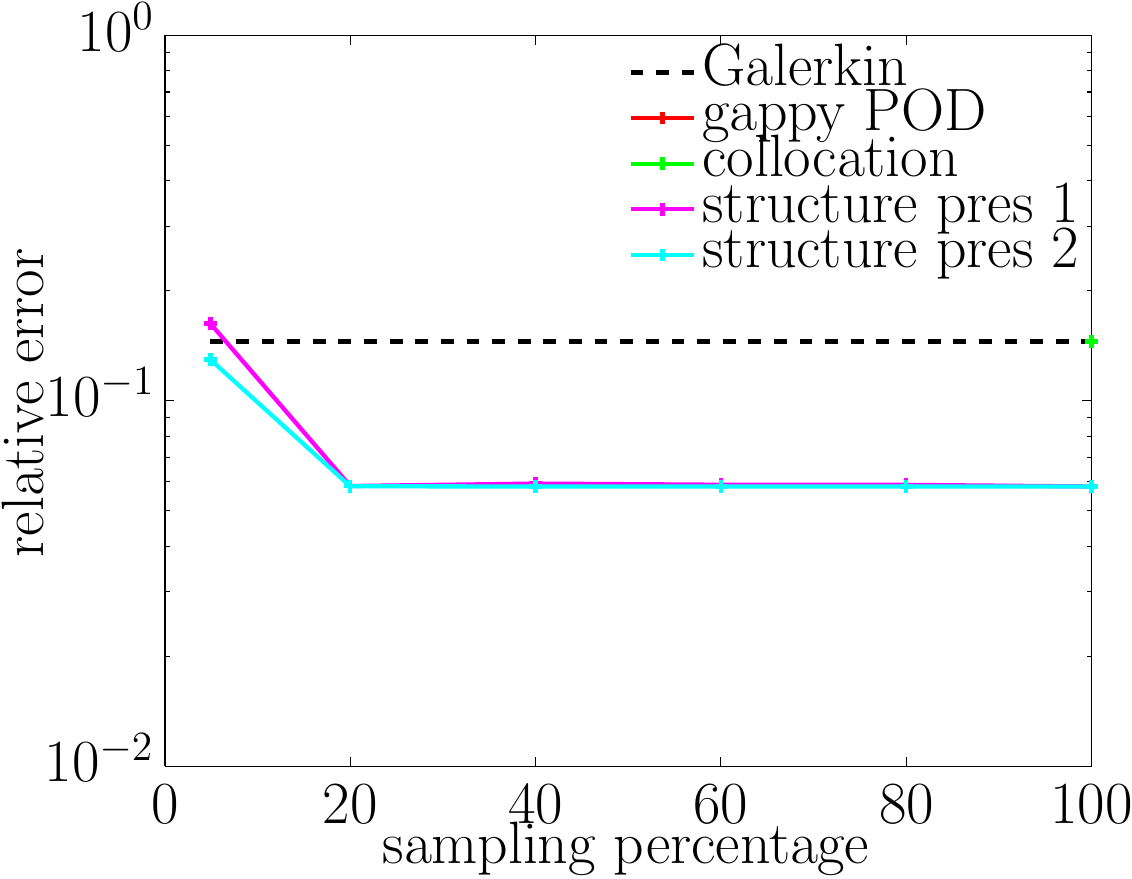}}
\subfigure[online point 3]
{\includegraphics[width=.48\textwidth]{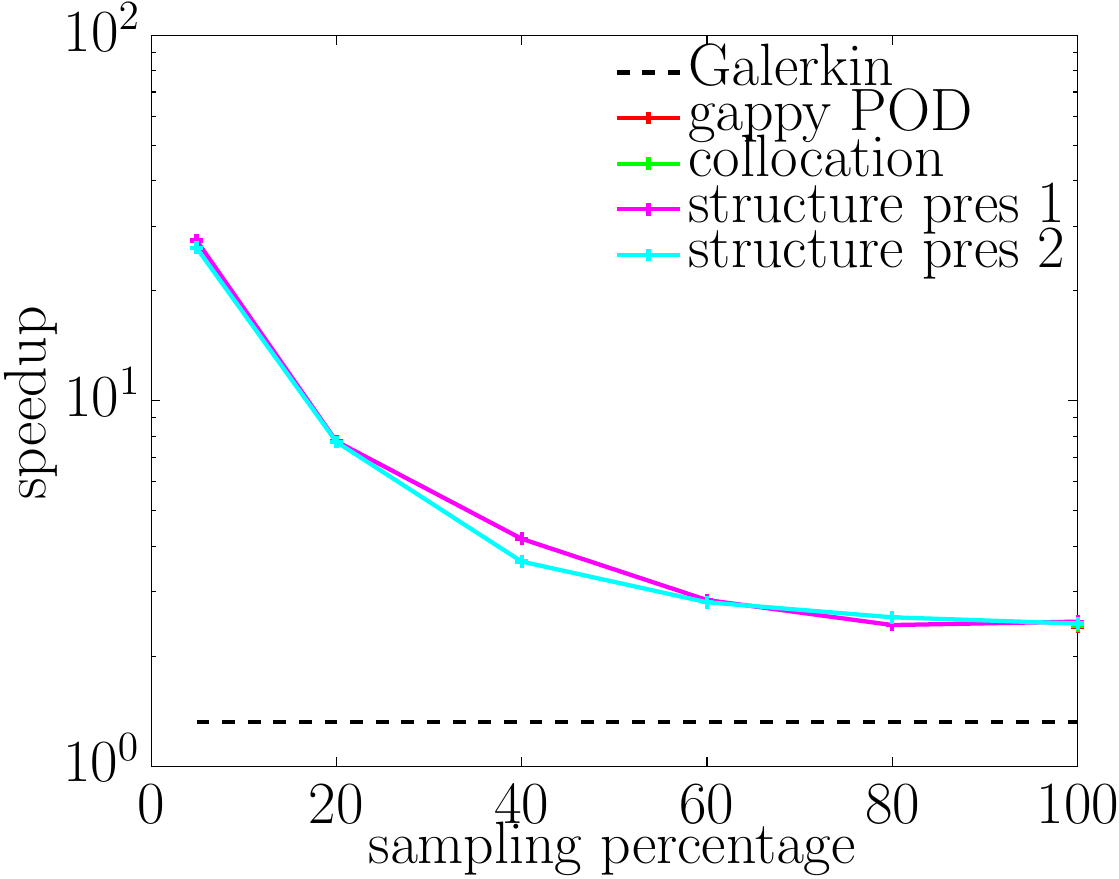}
\includegraphics[width=.48\textwidth]{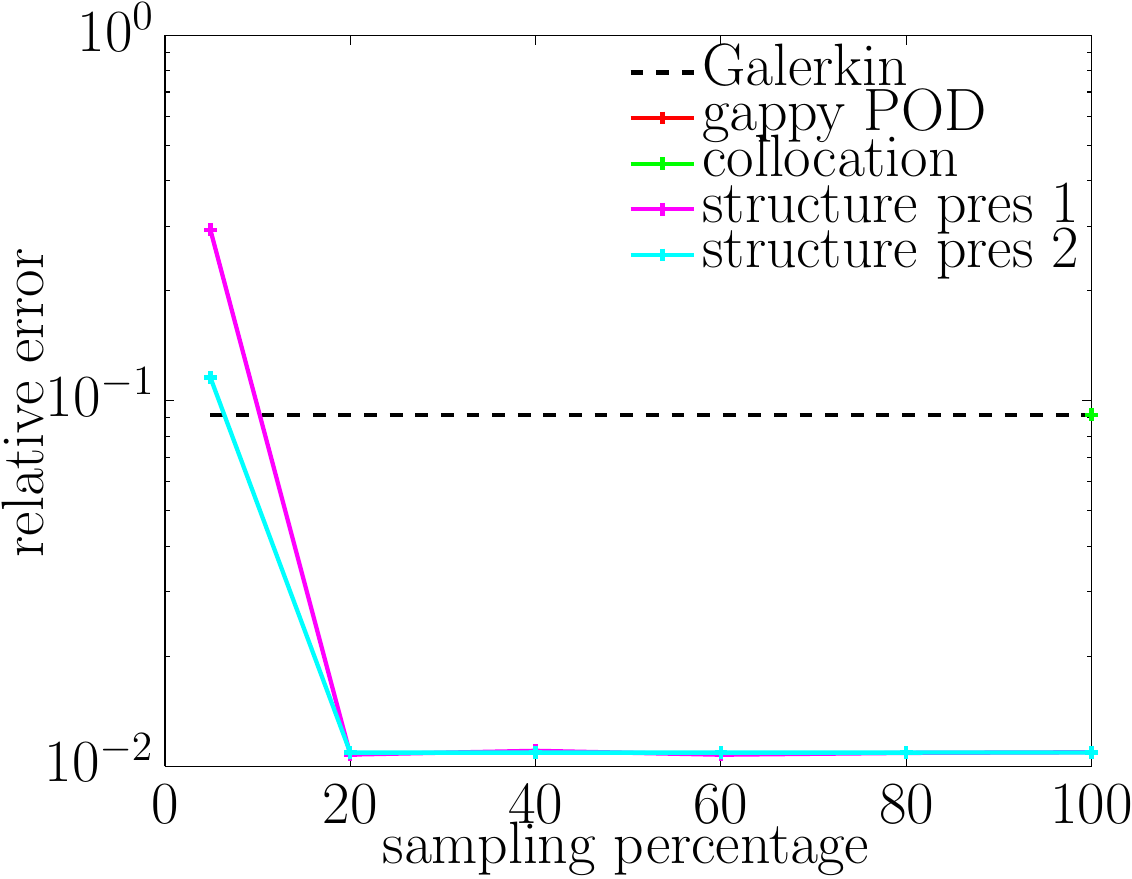}}
\caption{Conservative, varying-parameters case: reduced-order model performance as
a function of sampling percentage $\nsample/N\times 100\%$.}\label{fig:consPredictPerformError}
\end{figure}


\subsection{Non-conservative case}\label{sec:noncon_fixed}
We now consider the non-conservative case in which the non-conservative
dissipative and external forces are nonzero. That is, we set
$\zeta=\mathrm{sin}(5^\circ)$
and all parameters $\param_i$, $i=1,\ldots,16$ are free to vary.  We again set
the nominal forces to $\bar f_1 = \bar f_2 = 2\mathrm{kg} \times 9.81
\mathrm{m/s}^2$ and $\bar f_3 = \bar f_4 = 0.4 \mathrm{kg} \times 9.81
\mathrm{m/s}^2$.

As before, we perform a timestep-verification study for the nominal point
$\paramNom$ characterized by $\paramNom_i=0$, $i=1,\ldots 16$ to discover 
an appropriate timestep. 
 A timestep size of $\Delta t =
0.1$ seconds leads to 
an approximated error using Richardson
extrapolation of $1.07\times 10^{-4}$. We can therefore declare this to be an
appropriate timestep size for the numerical experiments. Further, we note that
the average number of Newton iterations per timestep is $2.56$, so the
nonlinearity remains significant.


\subsubsection{Fixed parameters}

We again test the different methods in the fixed-parameters case where
$\paramTrain = \paramNom$ and $\paramOnline = \paramNom$. As above, we only
collect snapshots for the first half of the time interval, and the proposed
structure-preserving methods yield the same results.

The POD reduced basis $\podstate$ is generated using an energy criterion of
$\energyCrit_\q = 1-10^{-5}$, which leads
to a basis dimension of $\nstate = 6\ll N$. The gappy POD-based
reduced-order model employs an energy criterion of $1$ for its
reduced bases $\podf$. Figures
\ref{fig:nonConsNonPredict} and \ref{fig:nonconsNonPredictPerformError} report
results for the reduced-order models as the number of sample indices varies.

Again, the Galerkin reduced-order model is accurate, with a relative error of
1.57\%, but produces a speedup of only 1.33. The proposed structure-preserving
method is always stable as expected. Its performance is dependent upon the
sampling percentage, with (arguably) the best performance achieved for 2\%
sampling (6.28\% error and 36.5 speedup). For 0.2\% sampling, the method
produces 16.1\% error and a speedup of 251; 20\% sampling leads to 5.39\%
error and a speedup of 4.6.

The gappy POD reduced-order model is unstable for 0.2\%, 2\%, and 5\%
sampling, but stabilizes at 20\%; compared to the conservative case, this
stability likely results from less stiff dynamics due to the presence of
damping. This yields its best performance of 1.53\% error, but only a 4.1
speedup.\footnote{A truncation criterion of $1$ yielded the best performance
for Gappy POD. For $\energyCrit_\f = 1-10^{-9}$, Gappy POD was unstable for
	all sampling percentages. It was also unstable for all sampling percentages
	when it employed an energy criterion of $\energyCrit_\f = 1-10^{-8}$.} The
	collocation reduced-order model is stable only for full sampling, when it is
	equivalent to Galerkin.

\begin{figure}[htbp]
\centering
\subfigure[0.2\% sampling]
{\includegraphics[width=0.32\textwidth]{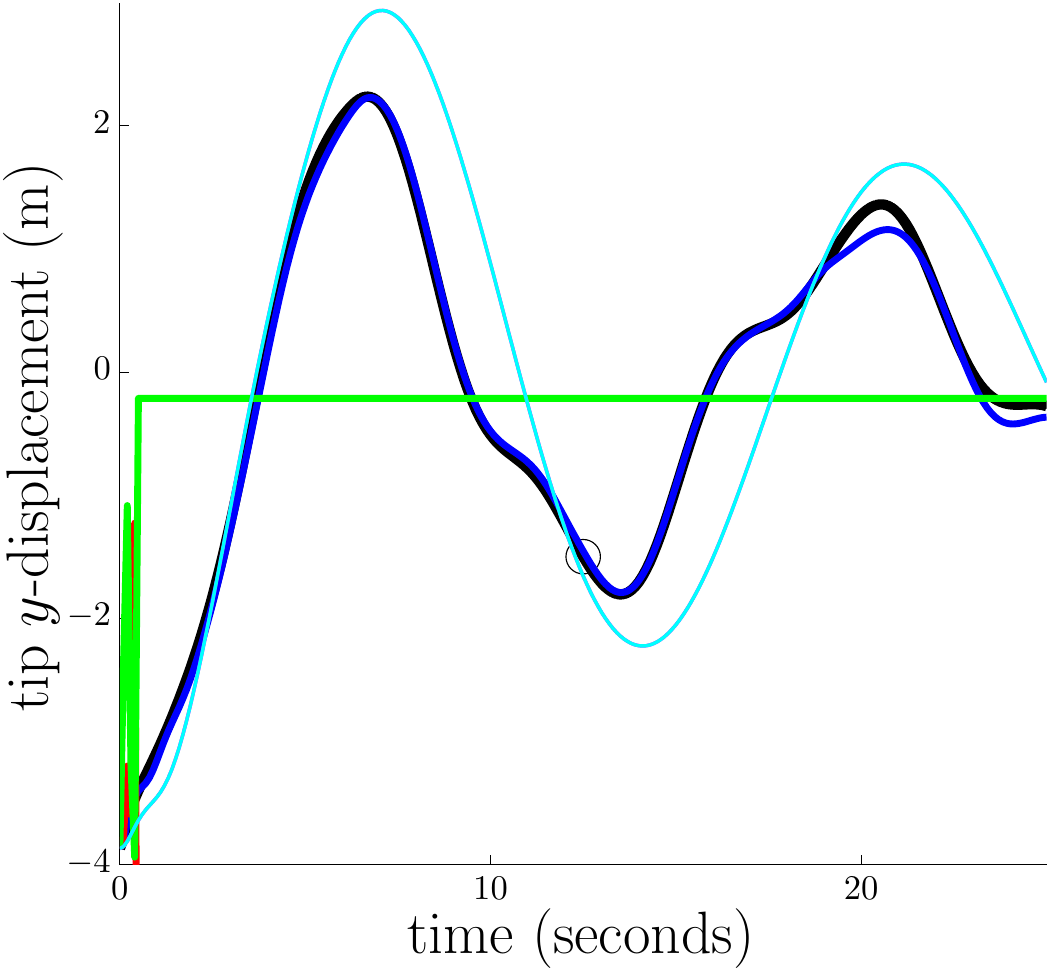}}
\subfigure[5\% sampling ]
{\includegraphics[width=0.32\textwidth]{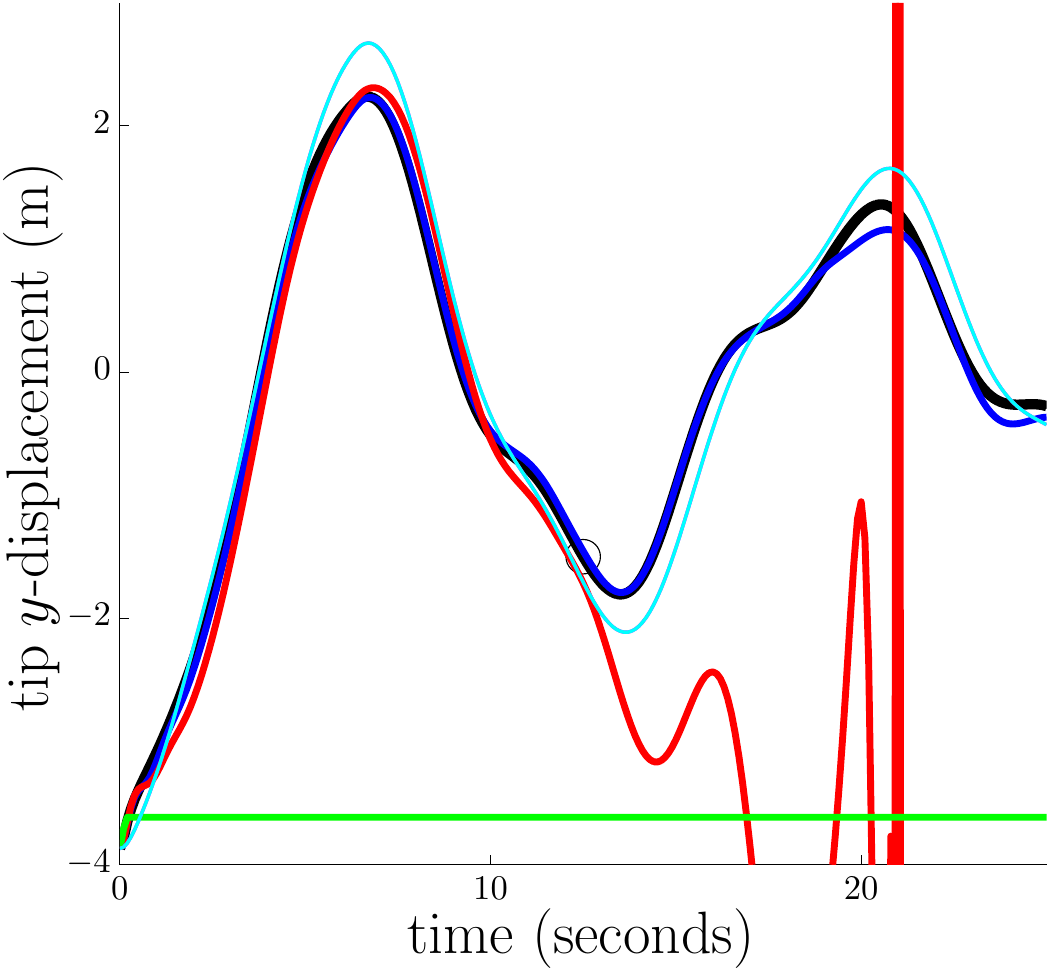}}
\subfigure[20\% sampling]
{\includegraphics[width=0.32\textwidth]{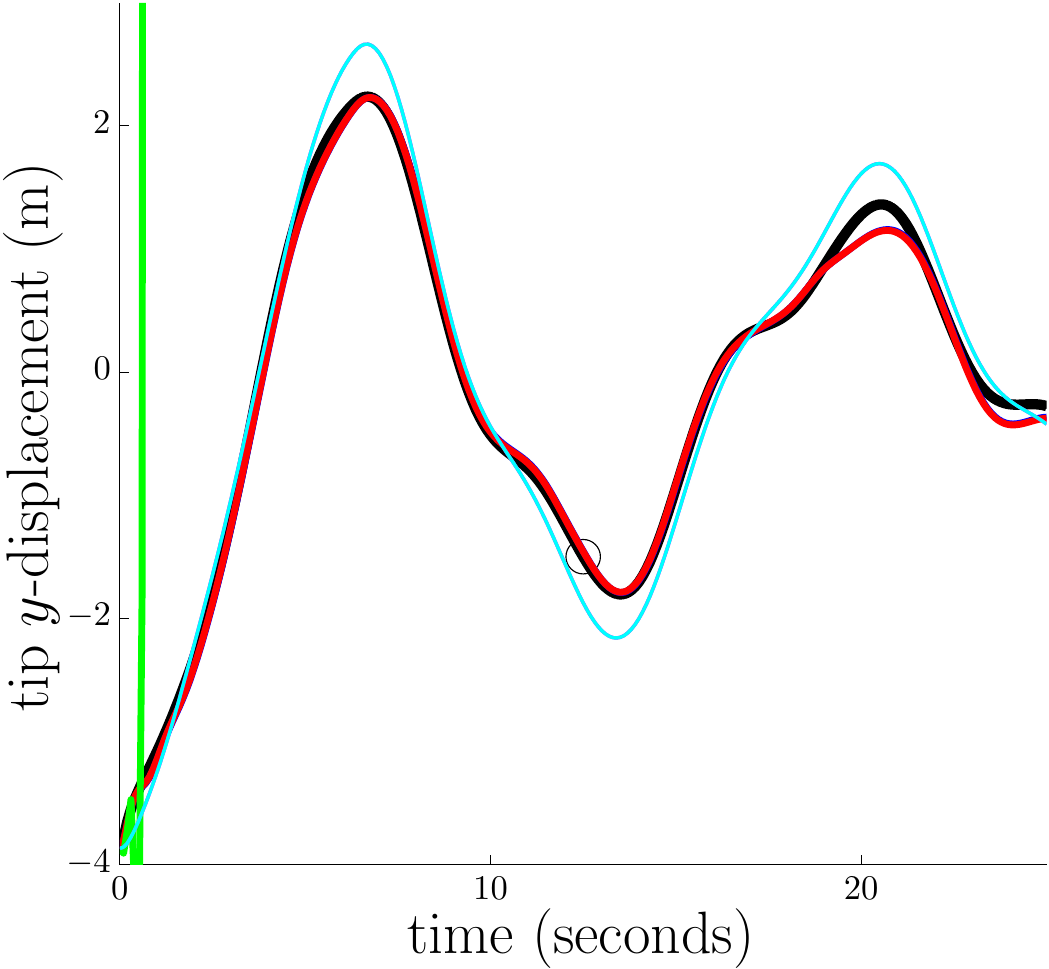}}
\caption{Non-conservative, fixed-parameters case: reduced-order model responses as
a function of sampling percentage $\nsample/N\times 100\%$. Legend: full-order model (black),
Galerkin ROM (dark blue), structure-preserving ROM (light blue), gappy POD ROM
(red), collocation ROM (green), end of training time interval (black circle).}\label{fig:nonConsNonPredict}
\end{figure}

\begin{figure}[htbp]
\centering
{\includegraphics[width=.48\textwidth]{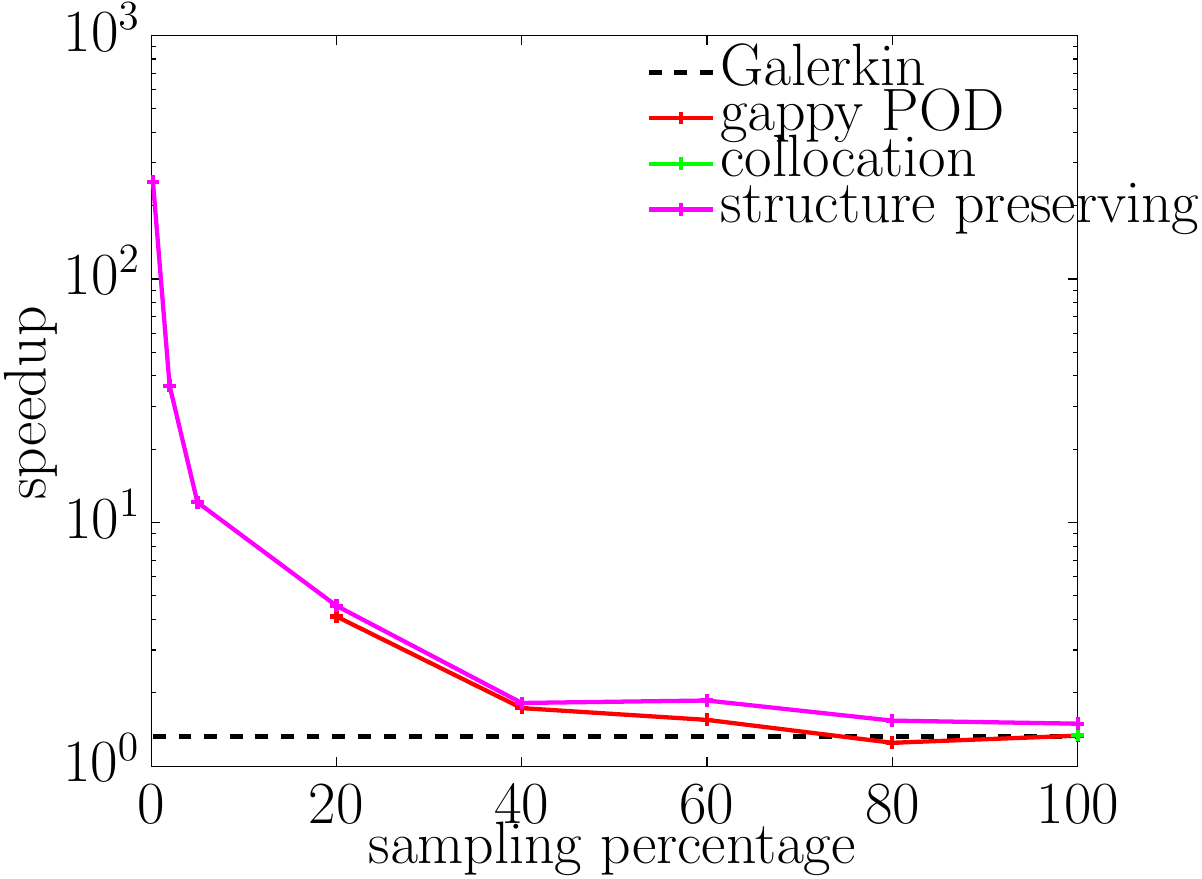}}
{\includegraphics[width=.48\textwidth]{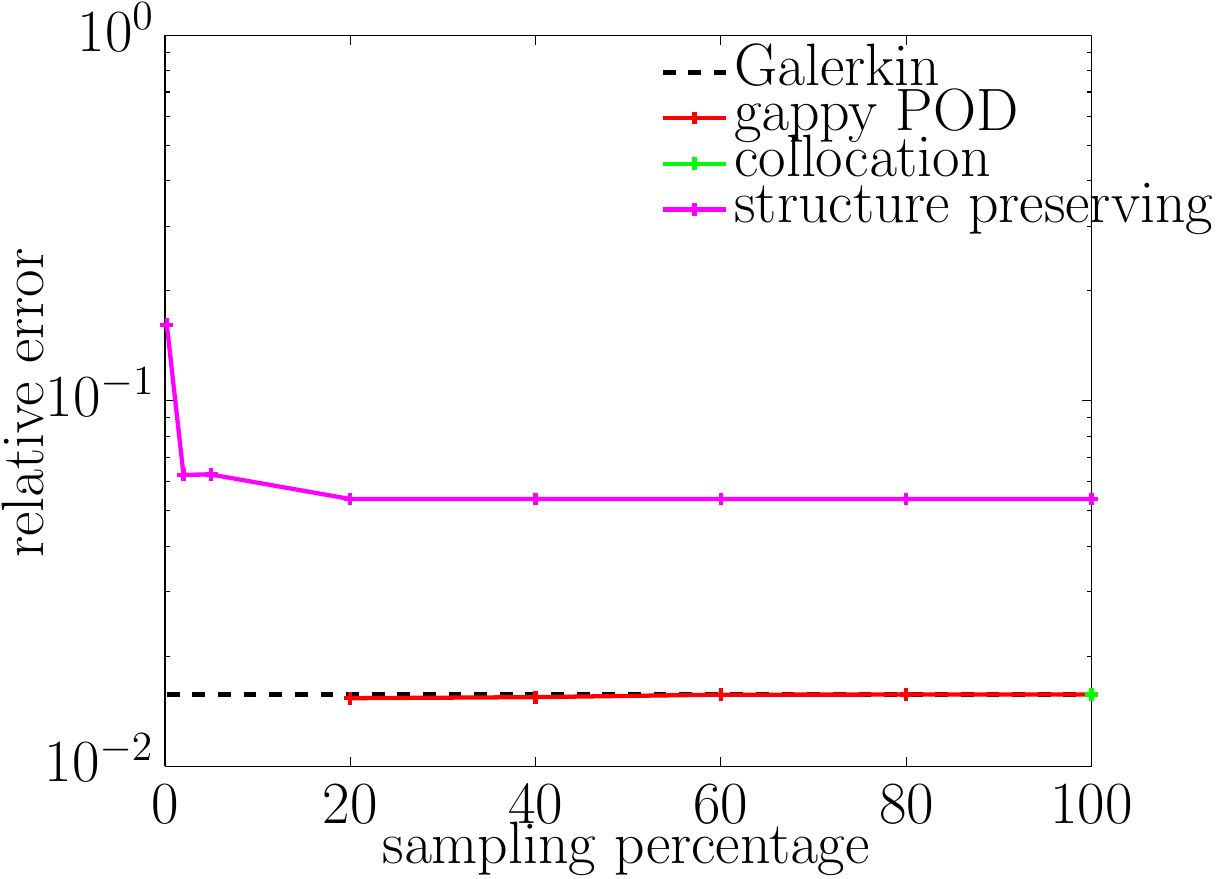}}
\caption{Non-conservative, fixed-parameters case: reduced-order model performance as
a function of sampling percentage $\nsample/N\times 100\%$. Missing data
points for gappy POD and collocation ROMs indicate unstable responses.}\label{fig:nonconsNonPredictPerformError}
\end{figure}

\subsubsection{Varying parameters}\label{sec:nonconPredict}
We now consider the parameter-varying case where $\paramOnline\not\in \paramTrain$.
We again employ $\nTrain = 6$ training points and determine $\paramTrain$ using Latin hypercube
sampling. We choose the online points randomly in the parameter space. Figure
\ref{fig:nonconsPredict} shows the tip displacement for the
training points; clearly, the responses are
significantly different from one another.

\begin{figure}[htbp] 
\centering 
\subfigure[Section \ref{sec:noncon_fixed} experiments ]{
\includegraphics[width=7cm]{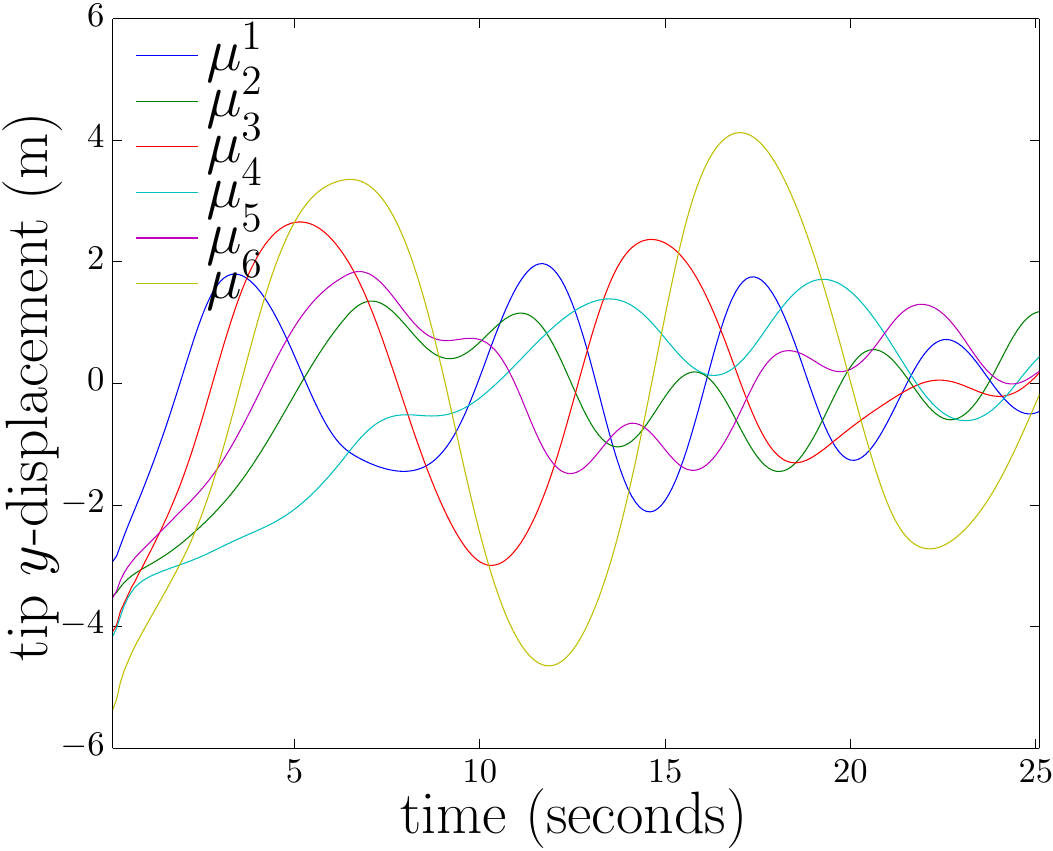}
\label{fig:nonconsPredict} 
}
\subfigure[Section \ref{sec:highlyNonlinear} experiments (higher
nonlinearity)]{
\includegraphics[width=7cm]{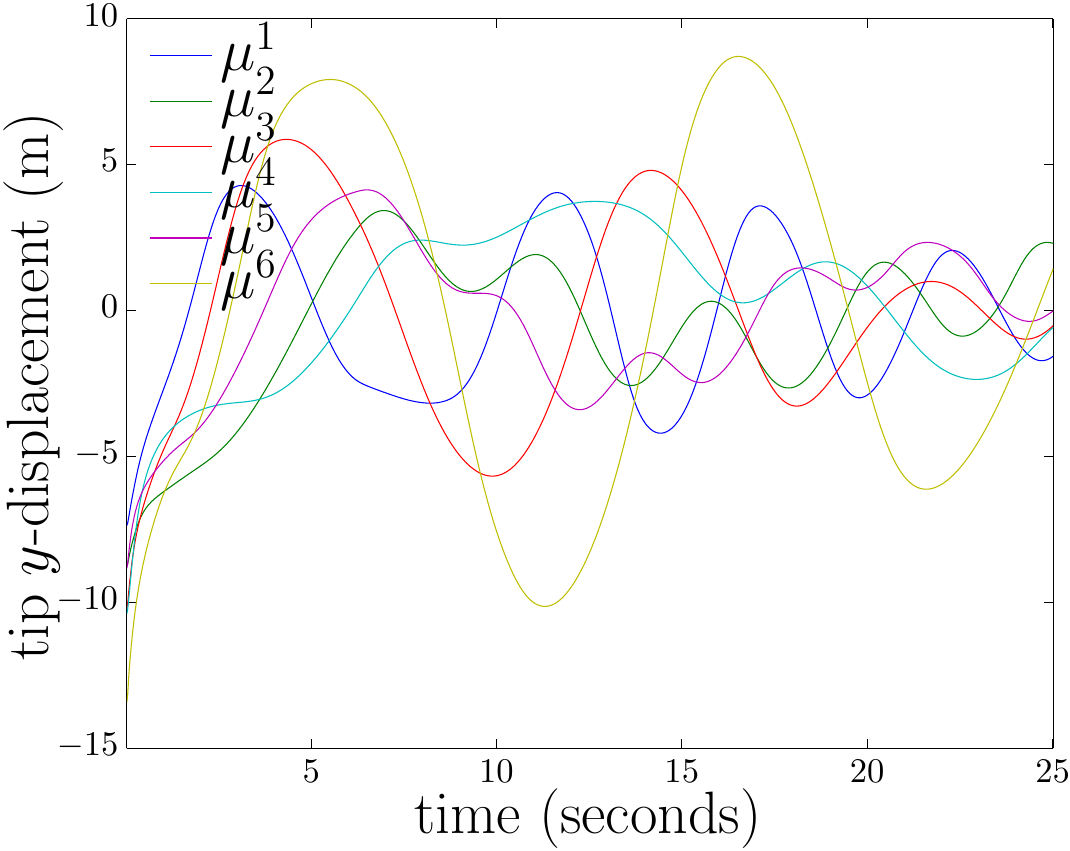}
\label{fig:nonconsNonlinPredict} 
}
\caption{Non-conservative, parameter-varying case: tip displacement for the training
set $\paramTrain$ for two sets of experiments. } 
\end{figure} 

Because we are in a fully predictive scenario, the two proposed
structure-preserving reduced-order models again yield different results. All
reduced-order models employ an energy criterion of $\energyCrit_\q =
1-10^{-5}$, which leads to a basis dimension of $\nstate = 12$. We employ
$\energyCrit_\f = 1$ for the gappy POD reduced-order model.

Figures \ref{fig:nonConsPredictROM} and \ref{fig:nonconsPredictPerformError}
report the results for this predictive study at the online points. At all
three points, Galerkin is accurate (relative errors of 9.8\%, 7.5\%, and 13.5\%), but does
not yield significant speedups (speedups of 1.2, 1.4, and 1.1). As is apparent from
the plots, the two proposed structure-preserving methods yield nearly the same
performance. At 0.4\% sampling, method 1 produces relative errors of 11.0\%,
2.82\%, and 
10.3\% and speedups of 73.3, 96.3, and 82.3. At 2\% sampling, method 1 yields
relative errors of 10.9\%, 4.38\%, and  7.97\% and speedups of 19.2, 21.6, and 16.8.

In this example, gappy POD does not stabilize until 40\% sampling, at which
point the speedup is less than 1. Thus, gappy POD does not yield performance
improvement for this problem. Collocation stabilizes at 80\%
sampling, and also fails to generate any performance improvement.

\begin{figure}[htbp] \centering
\subfigure[predictions at three randomly chosen online points, $0.4\%$
sampling]
{\includegraphics[width=.3\textwidth]{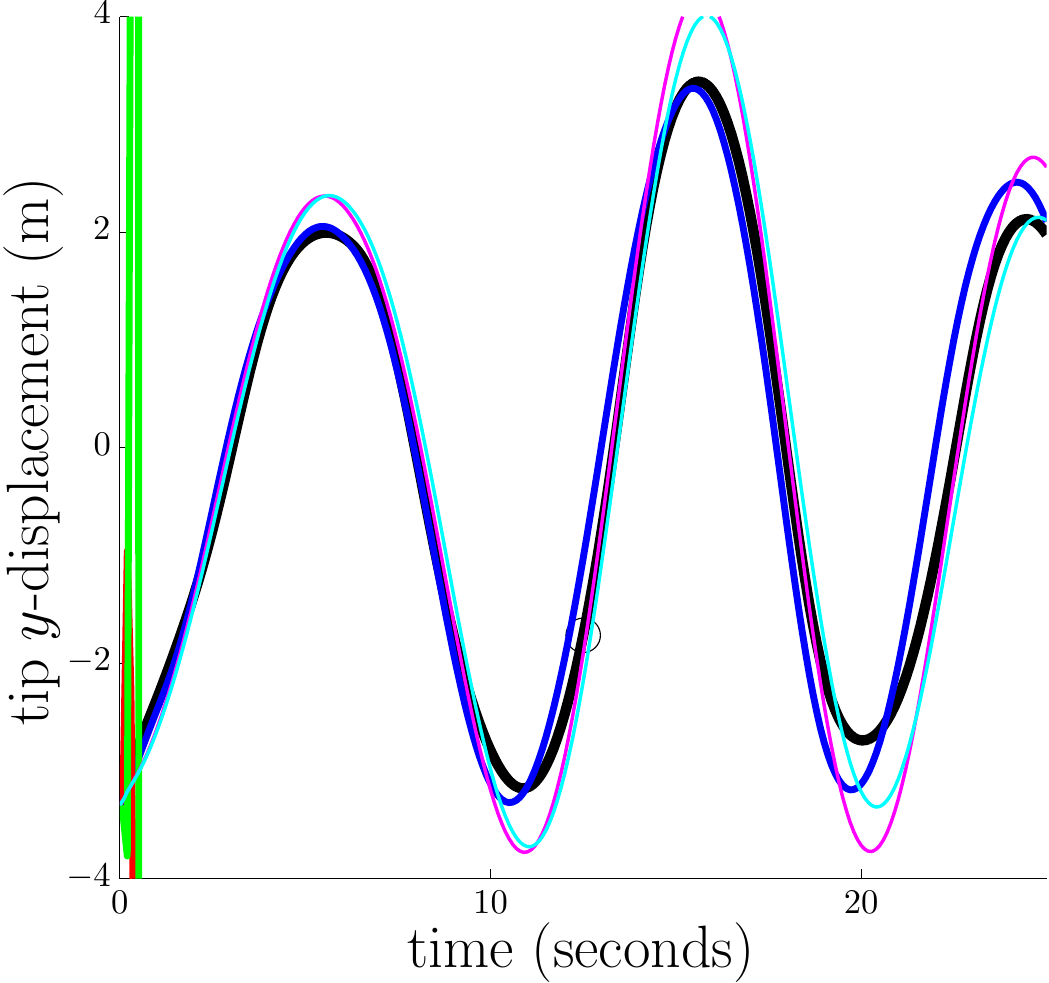}
\includegraphics[width=.3\textwidth]{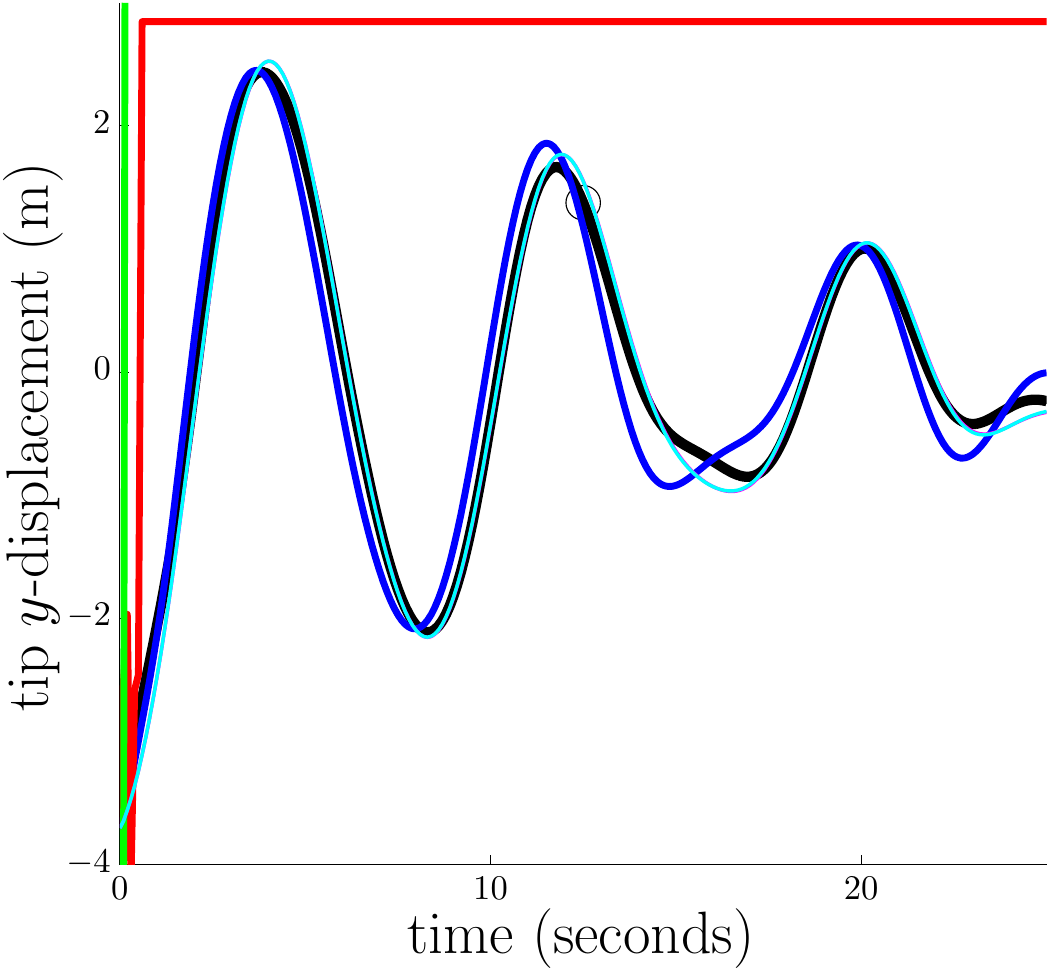}
\includegraphics[width=.3\textwidth]{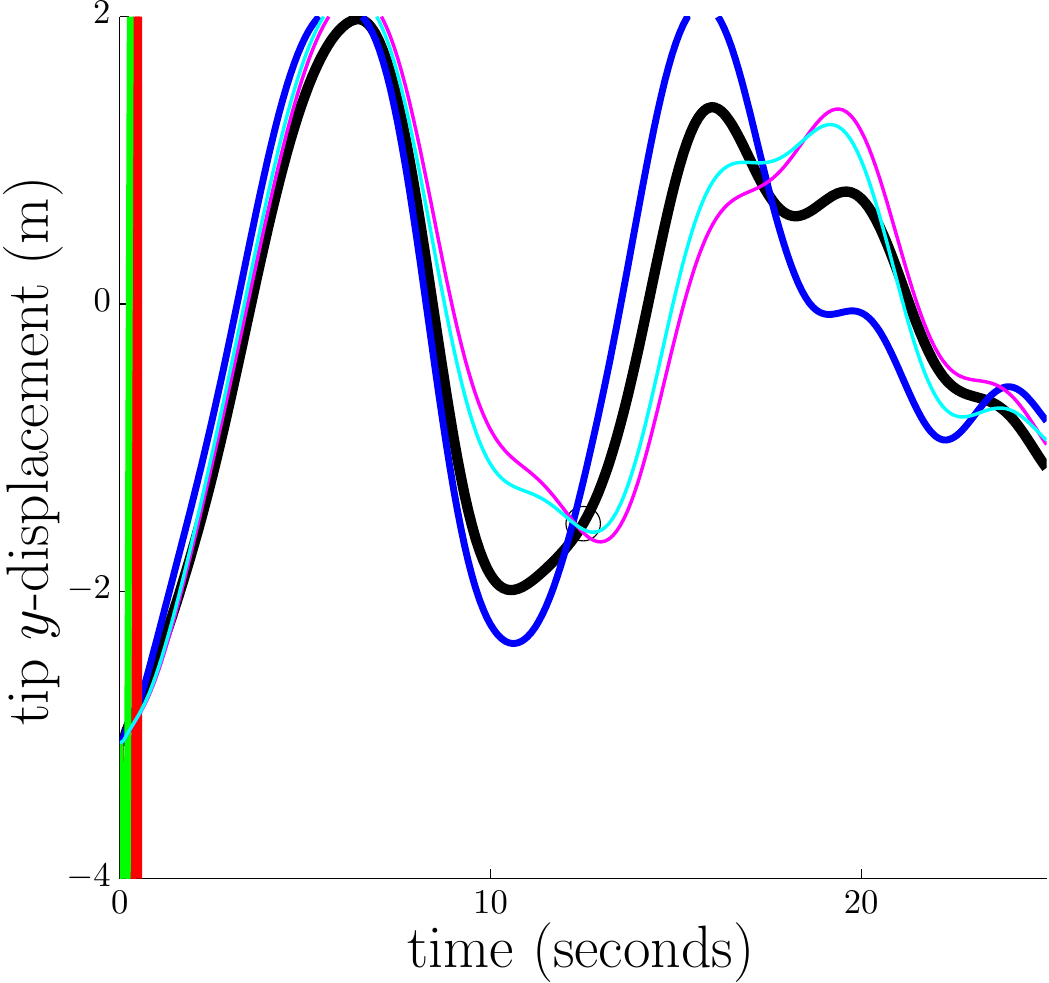}}
\subfigure[predictions at three randomly chosen online points, $40\%$ sampling]
{\includegraphics[width=.3\textwidth]{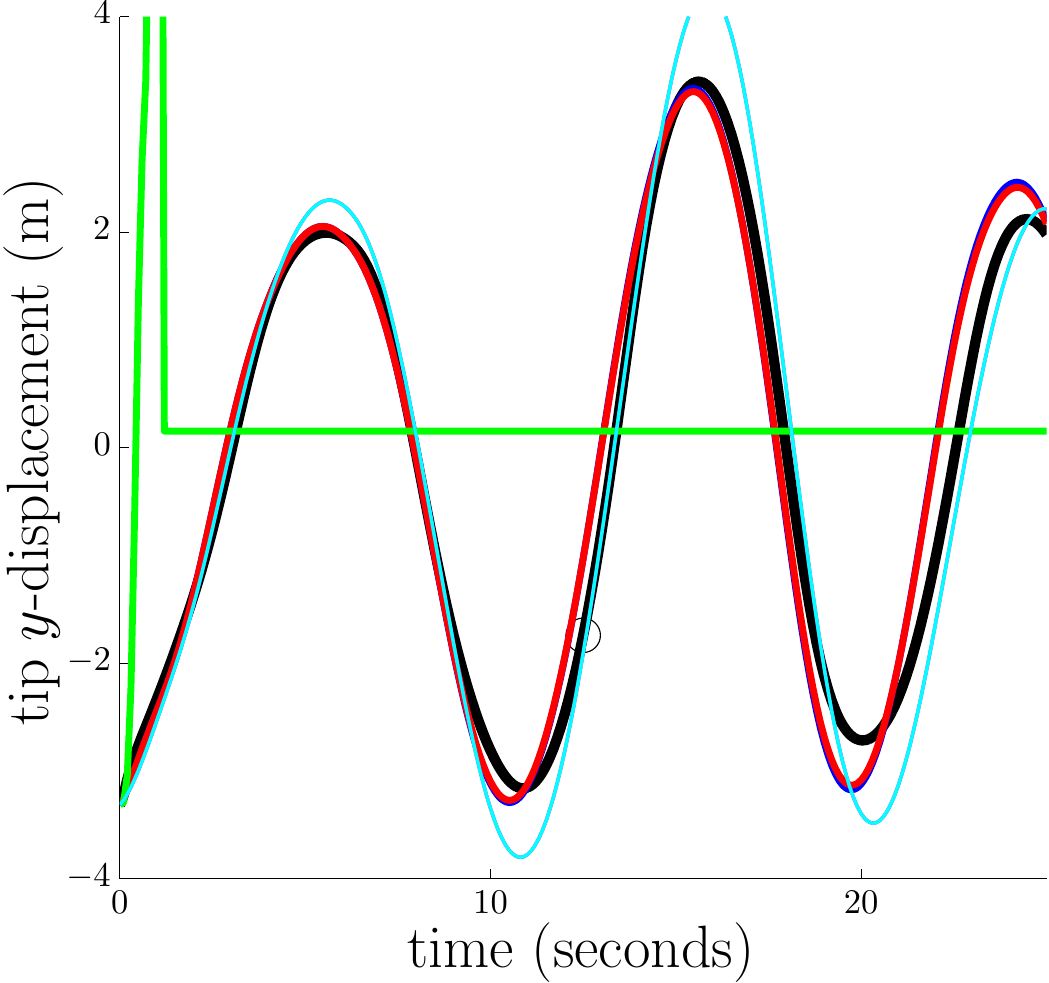}
\includegraphics[width=.3\textwidth]{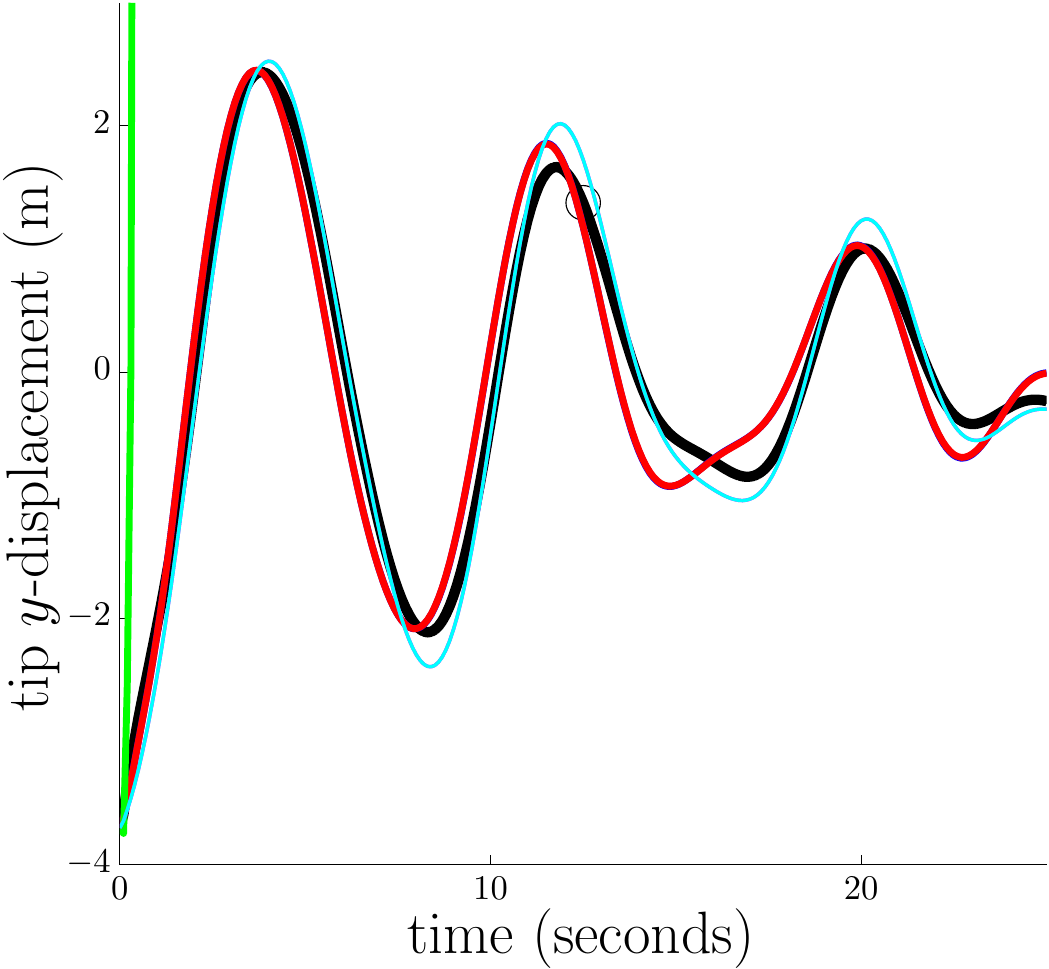}
\includegraphics[width=.3\textwidth]{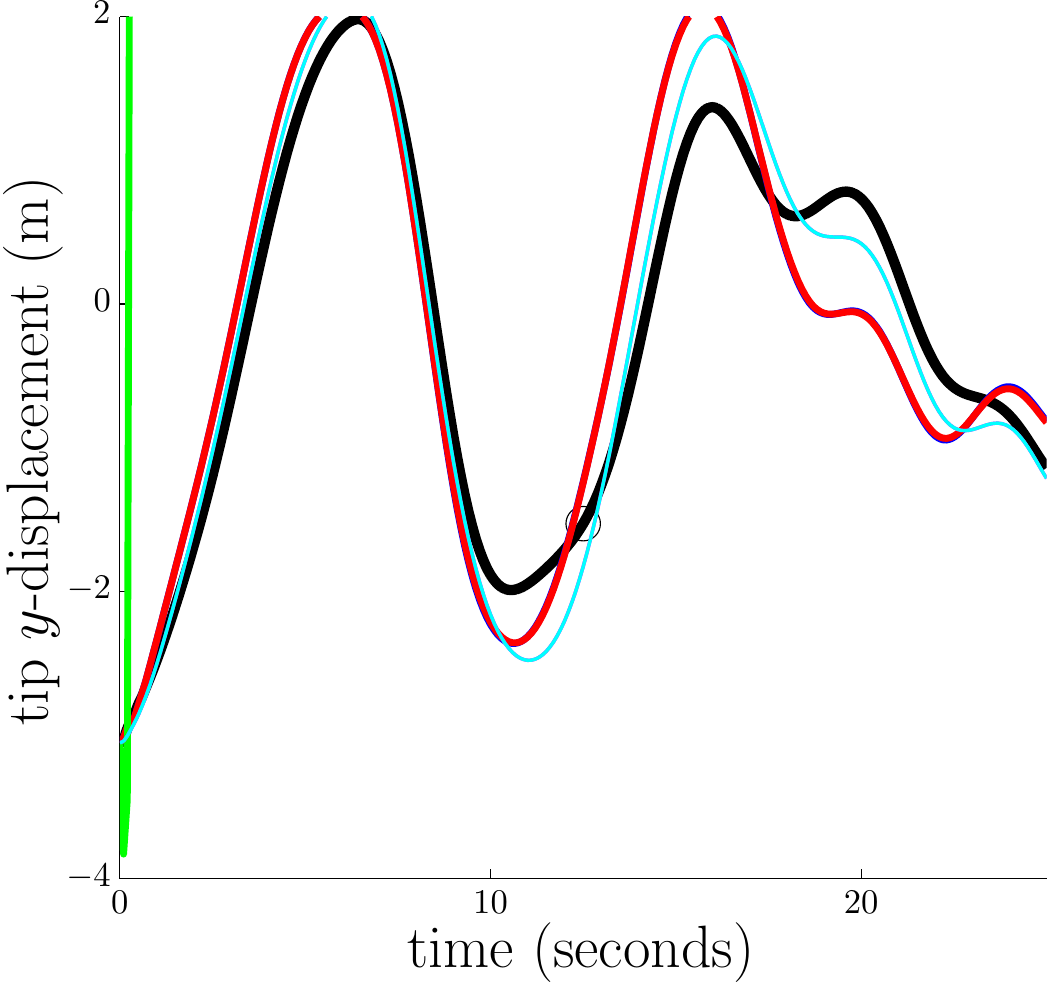}}
\subfigure[predictions at three randomly chosen online points, $80\%$ sampling]
{\includegraphics[width=.3\textwidth]{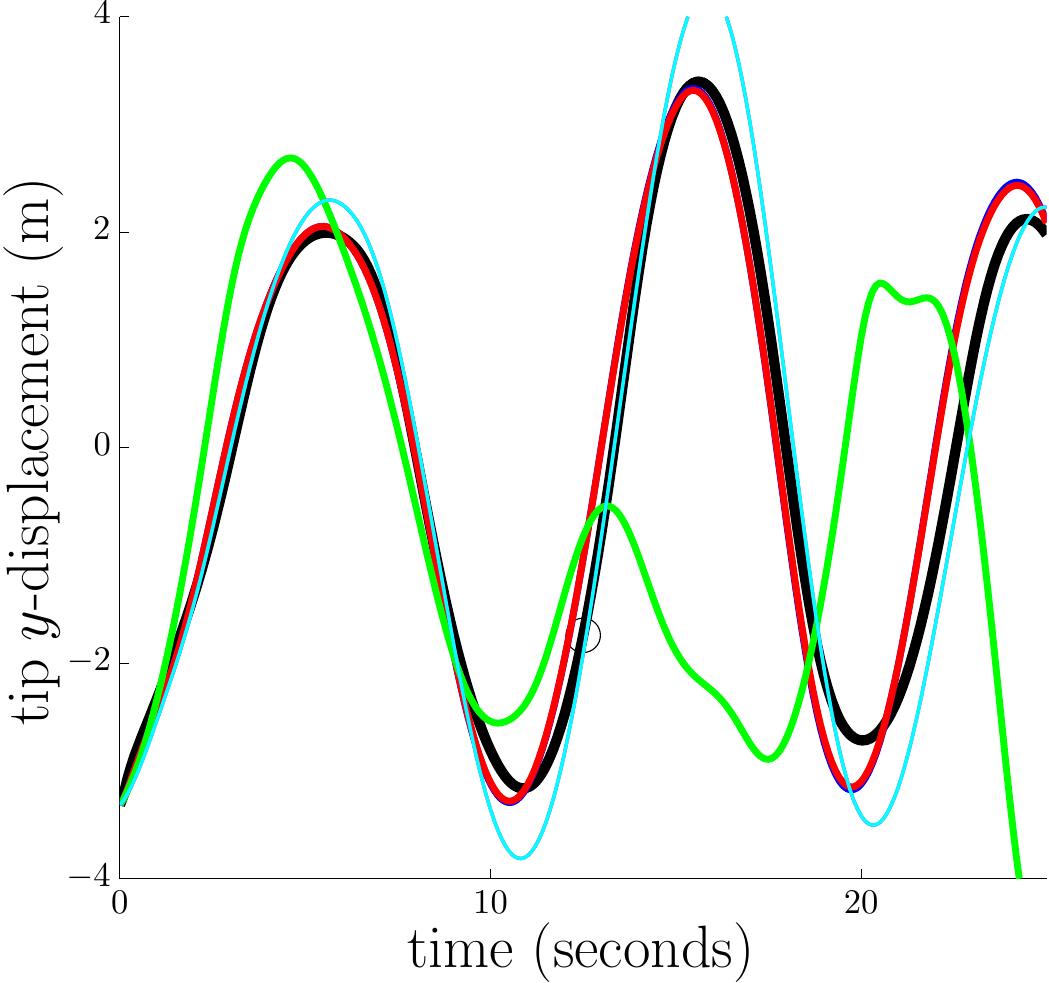}
\includegraphics[width=.3\textwidth]{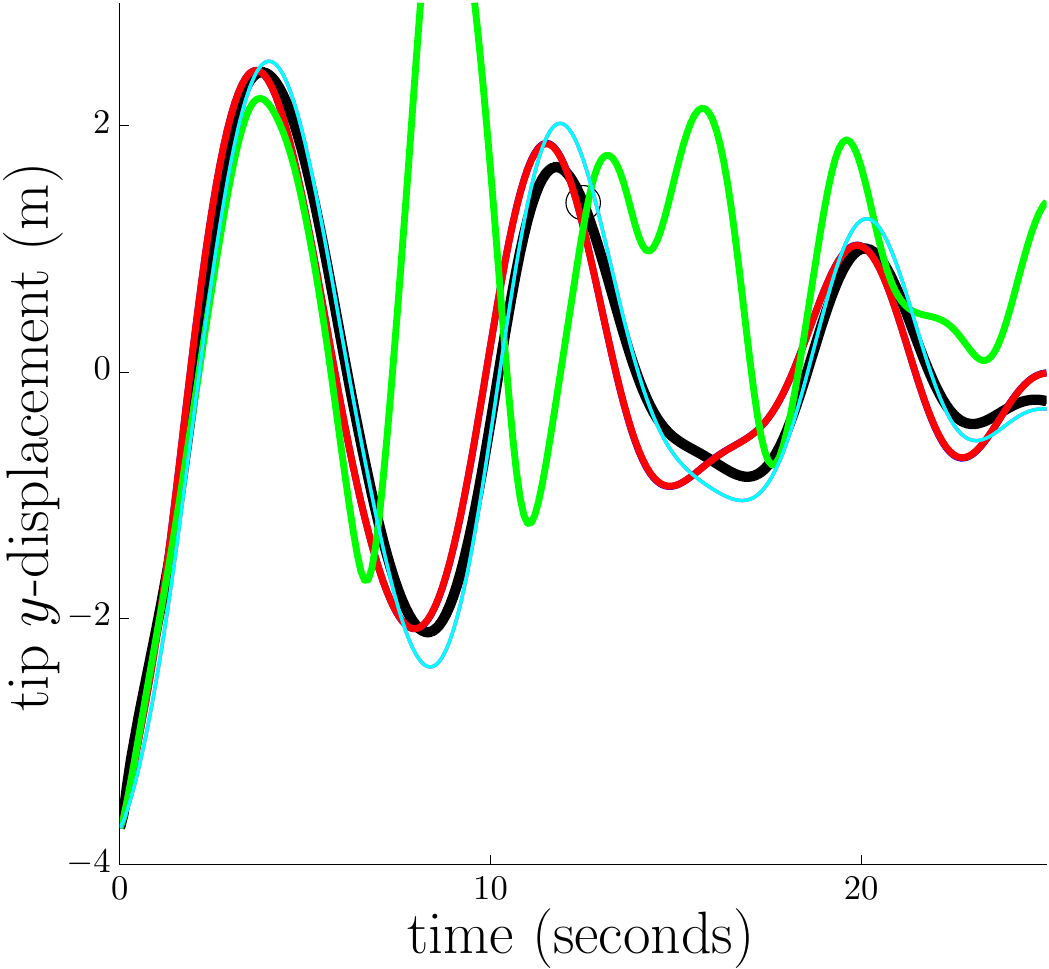}
\includegraphics[width=.3\textwidth]{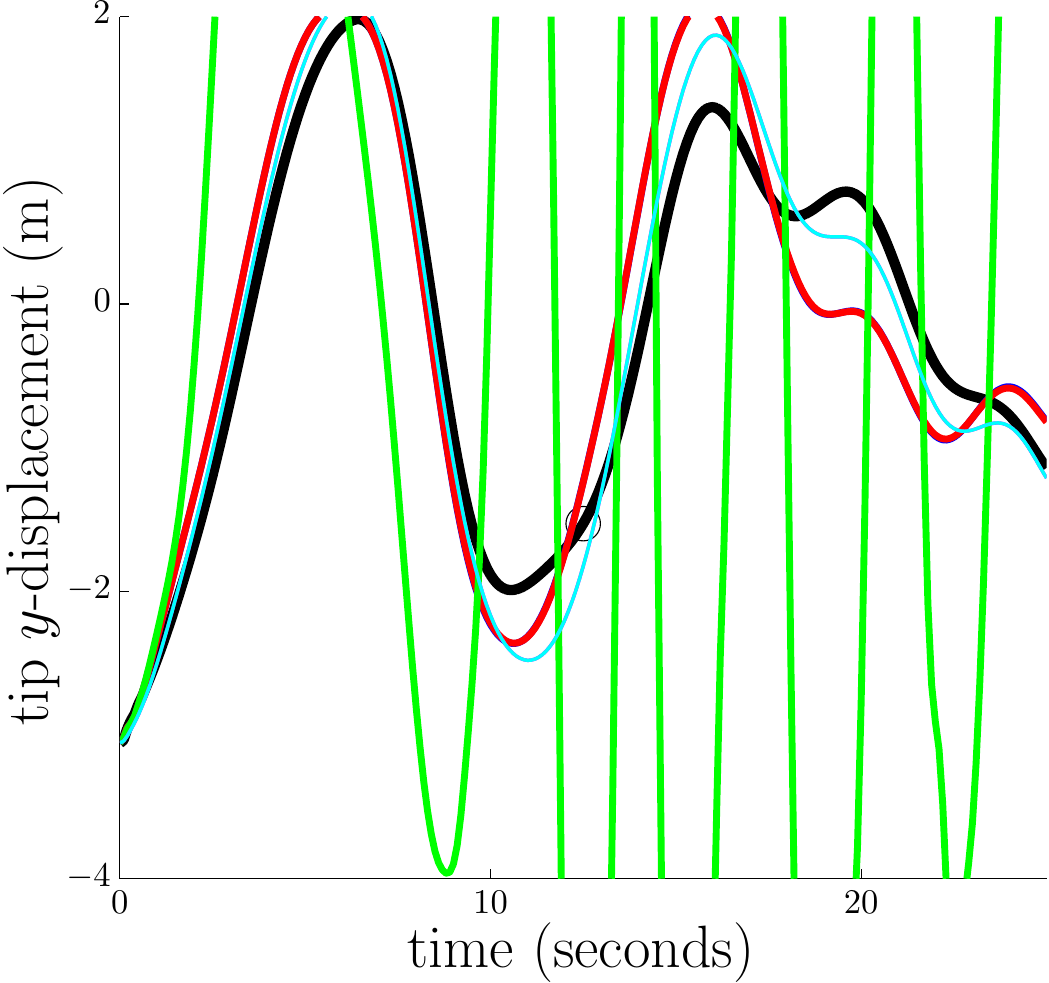}}
\caption{Non-conservative, parameter-varying case: reduced-order model responses as
a function of sampling percentage $\nsample/N\times 100\%$. Legend: full-order model (black),
Galerkin ROM (dark blue), structure-preserving ROM method 1 (magenta), structure-preserving ROM method 2 (light blue), gappy POD ROM
(red), collocation ROM (green), end of training time interval (black circle).}\label{fig:nonConsPredictROM}
\end{figure}

\begin{figure}[htbp]
\centering
\subfigure[online point 1]
{\includegraphics[width=.48\textwidth]{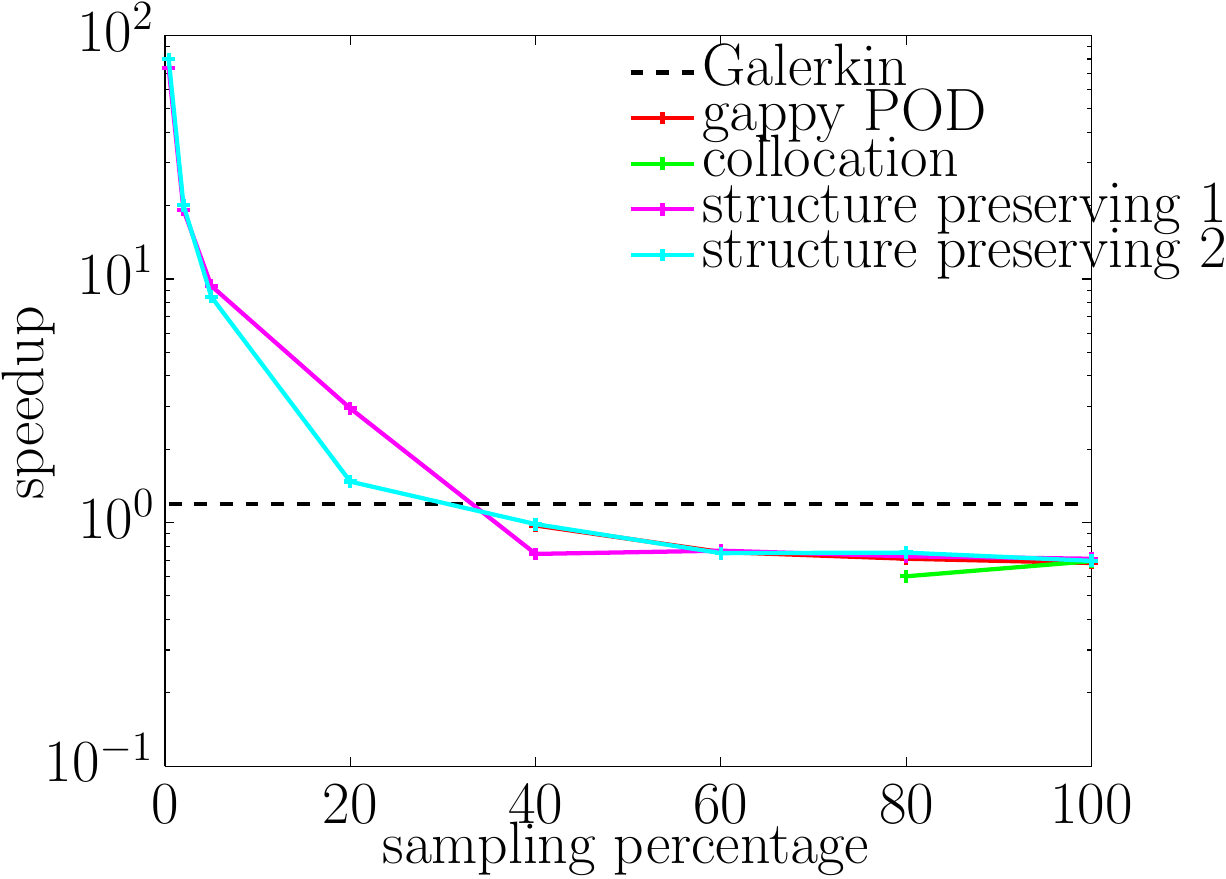}
\includegraphics[width=.48\textwidth]{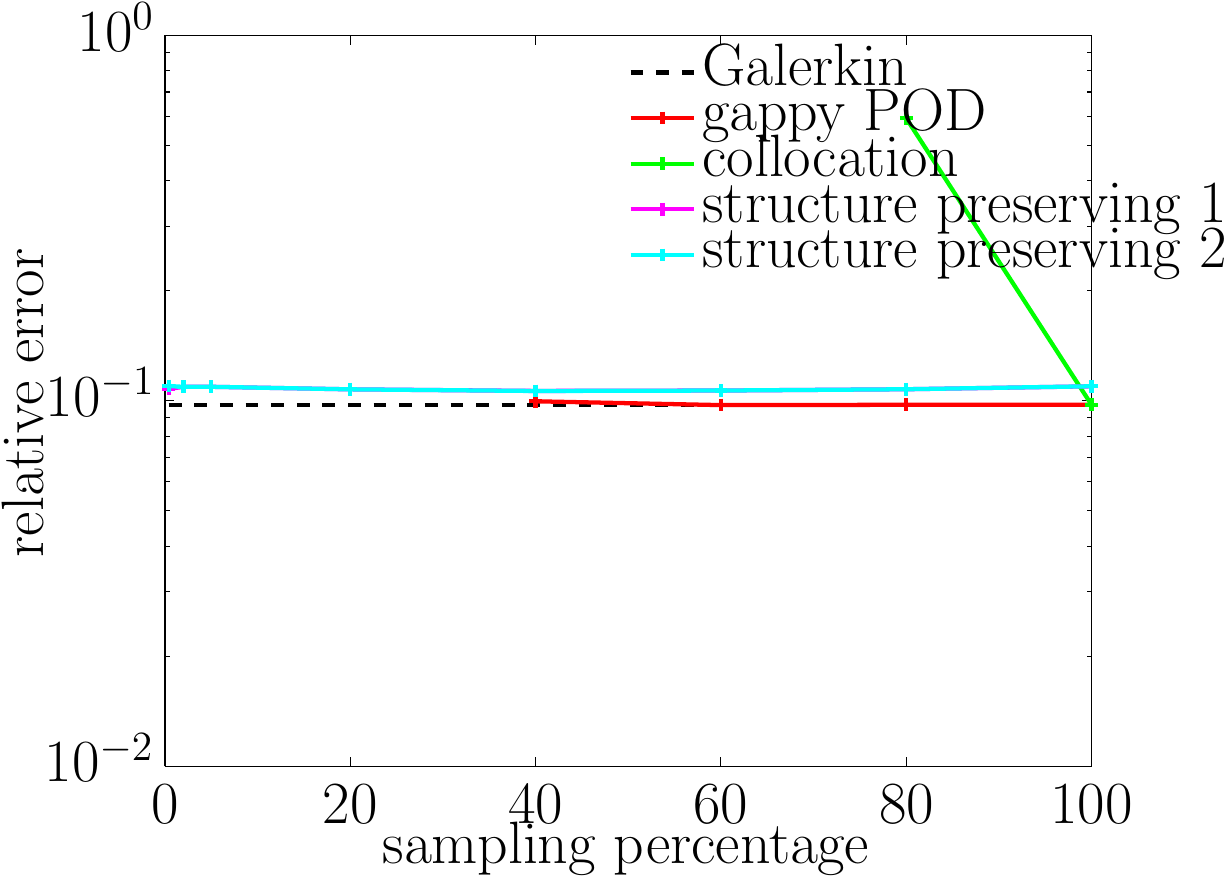}}
\subfigure[online point 2]
{\includegraphics[width=.48\textwidth]{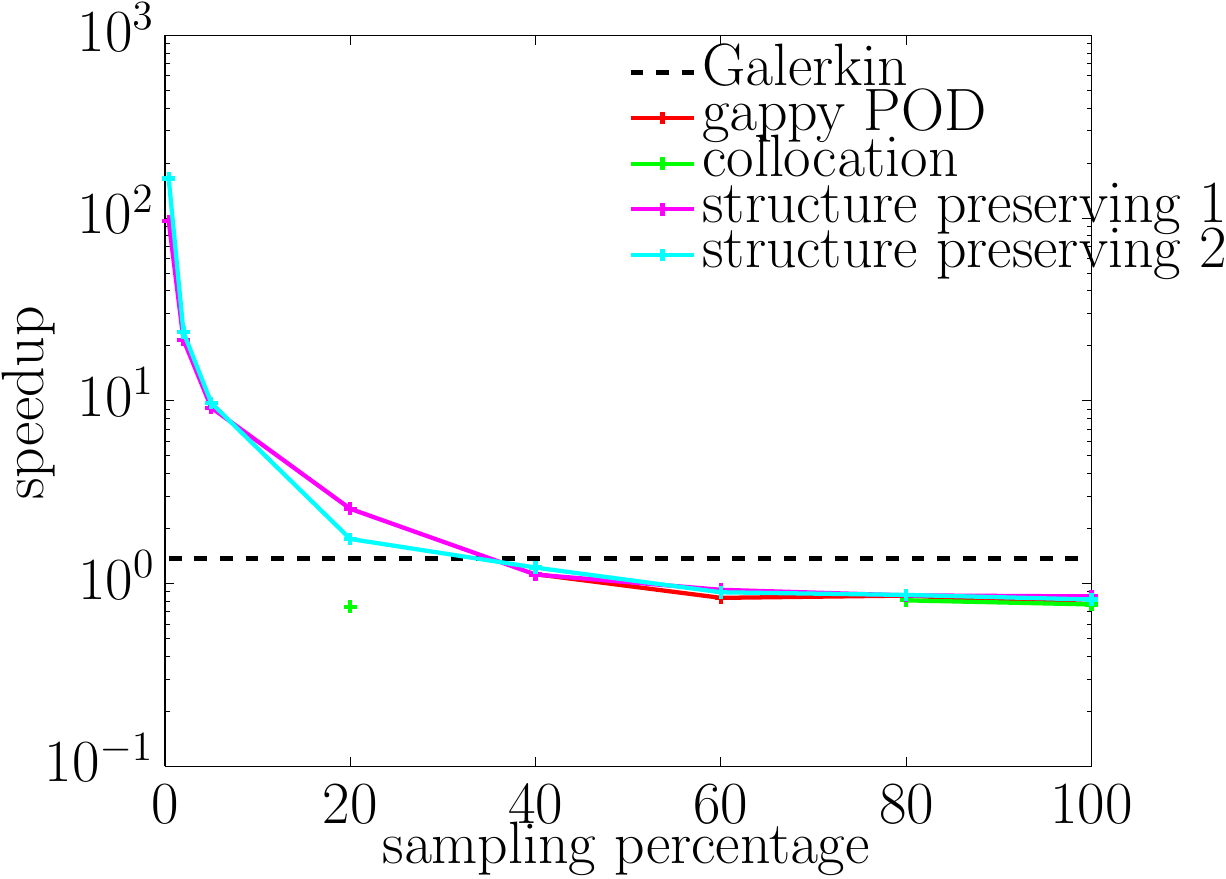}
\includegraphics[width=.48\textwidth]{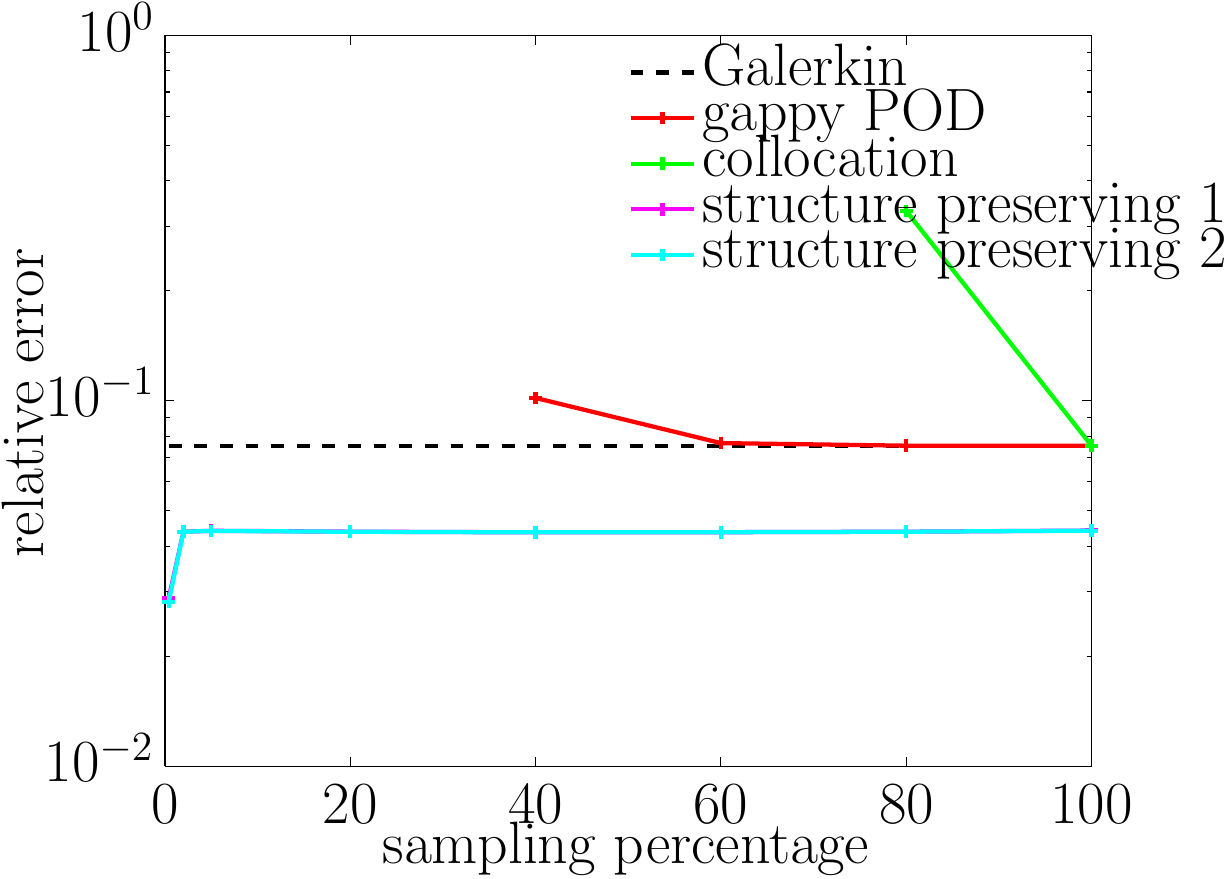}}
\subfigure[online point 3]
{\includegraphics[width=.48\textwidth]{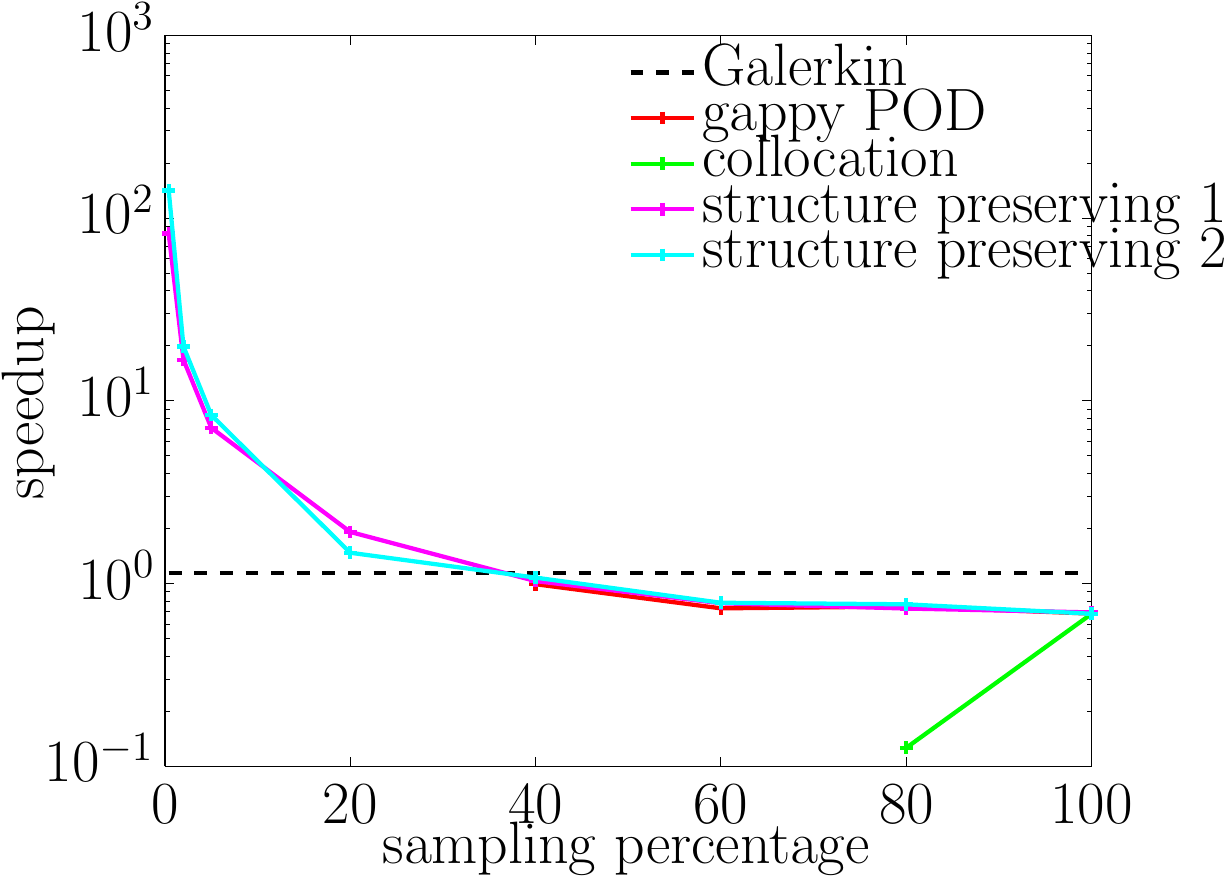}
\includegraphics[width=.48\textwidth]{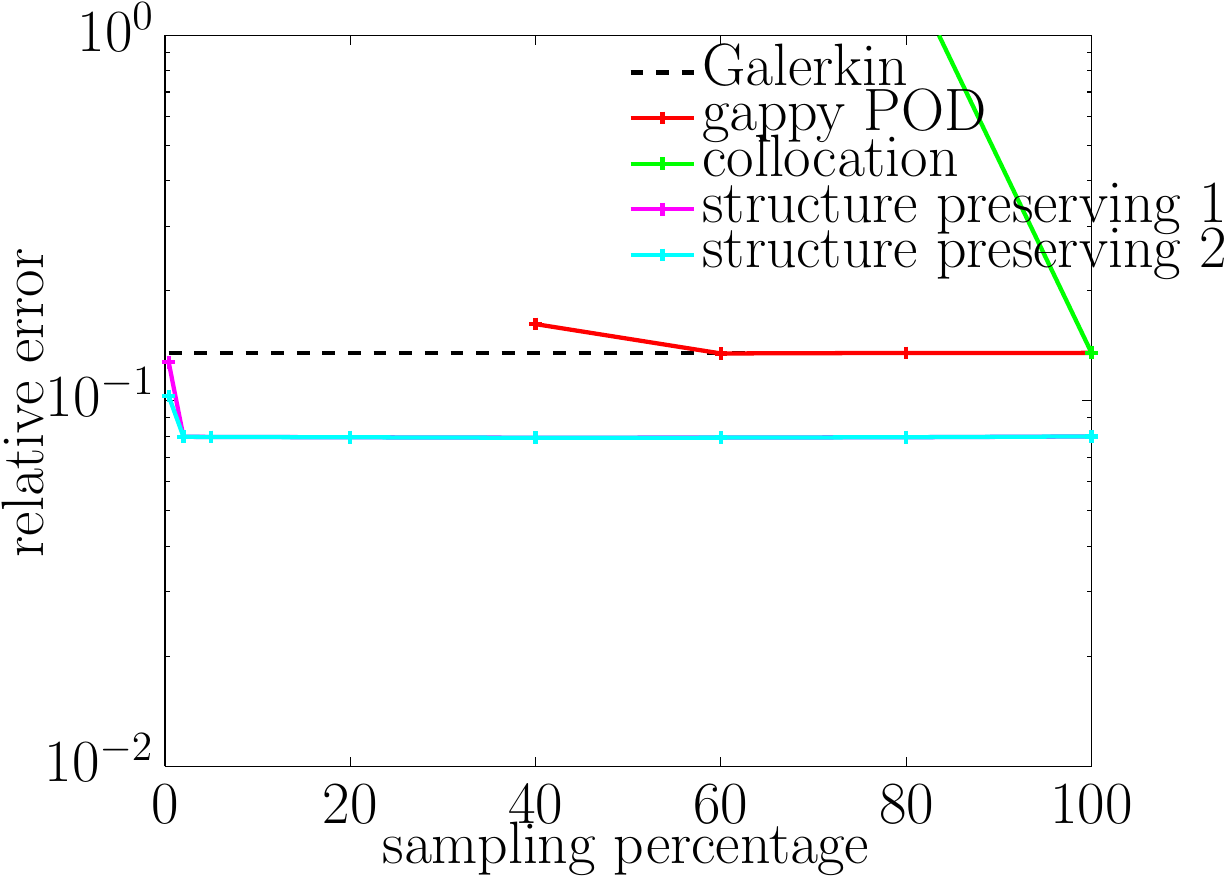}}
\caption{Non-conservative, parameter-varying case: reduced-order model performance as
a function of sampling percentage $\nsample/N\times 100\%$.}\label{fig:nonconsPredictPerformError}
\end{figure}

\subsection{Effect of nonlinearity}\label{sec:highlyNonlinear}
We now aim to characterize the dependence of problem nonlinearity on the
proposed methods' performances. Recall from Section \ref{sec:potEnGen} that
the potential-energy approximation is computed by matching the gradient of the potential
energy to first order about the equilibrium configuration $\initialQOn$.
In the presence of stronger nonlinearity, we expect the configuration to
deviate further from equilibrium, which should degrade the accuracy of the
approximation.

To numerically assess the effect of nonlinearity, we repeat the experiments
from Section \ref{sec:nonconPredict} using the same training and online
points, but we increase the nominal forces 
by a factor of 2.5 to $\fNomnum{1} = \fNomnum{2} = 5\mathrm{kg}
\times 9.81 \mathrm{m/s}^2$ and $\fNomnum{3} = \fNomnum{4} = 1 \mathrm{kg} \times 9.81 \mathrm{m/s}^2$.
We first perform a timestep-verification study for the nominal point
$\paramNom$. As expected,
a smaller timestep size of $\Delta t =
0.025$ seconds is required, as it
corresponds to an approximated error using Richardson
extrapolation of $3.62\times 10^{-4}$. 


Figure \ref{fig:nonconsNonlinPredict} displays the tip displacement for the
training points. Note that the responses are similar to those for the previous
study (see Figure \ref{fig:nonconsPredict}), but have larger magnitudes and
therefore imply a stronger geometric nonlinearity.  The reduced-order models
employ a POD reduced basis of dimension $\nstate = 14$, which was obtained by
an energy criterion of $\energyCrit_\q = 1-10^{-5}$; gappy POD uses
$\energyCrit_\f=1$ for its nonlinear-function bases.

Figures \ref{fig:nonConsNonlinPredictROM} and
\ref{fig:nonconsNonlinPredictPerformError} report the reduced-order models'
performances for this problem. As in the previous case, Galerkin is accurate
(relative errors of 8.3\%, 3.0\%, and 10.0\% at the online points), but does
not generate significant speedups (1.67, 1.71, and 1.0). The proposed
structure-preserving techniques again yield very similar results to each
other; however, the errors are significantly larger than in the the
experiments from Section \ref{sec:nonconPredict} characterized by a
less severe nonlinearity. For 0.5\% sampling, proposed method 1 yields
relative errors of 21.3\%, 11.1\%, and 15.9\% at the online points and speedups
of 116.4, 160, and 98.9. Thus, increasing the nonlinearity in the problem does
have a deleterious effect on the methods' performances. 

However, it is important to note that other complexity-reducing reduced-order
models fail to generate significant performance improvement on this more
highly nonlinear problem. In particular collocation is always unstable for a
sampling percentage less than 60\%, and gappy POD is always unstable when the
percentage is less than 80\%. As a result, the best speedup obtained by either
of the methods is only 2.77 (collocation for 60\% sampling for online
point 2).


\begin{figure}[htbp] \centering
\subfigure[predictions at three randomly chosen online points, $0.5\%$
sampling]
{\includegraphics[width=.3\textwidth]{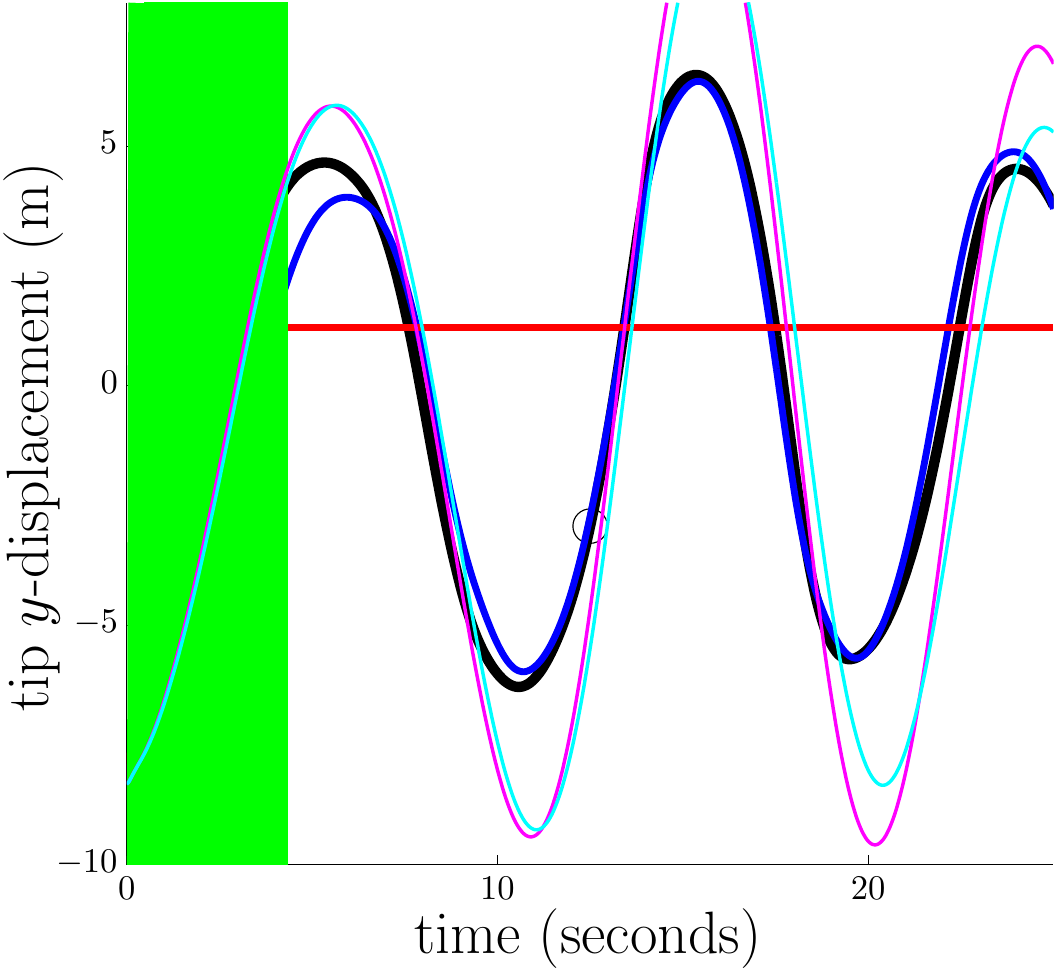}
\includegraphics[width=.3\textwidth]{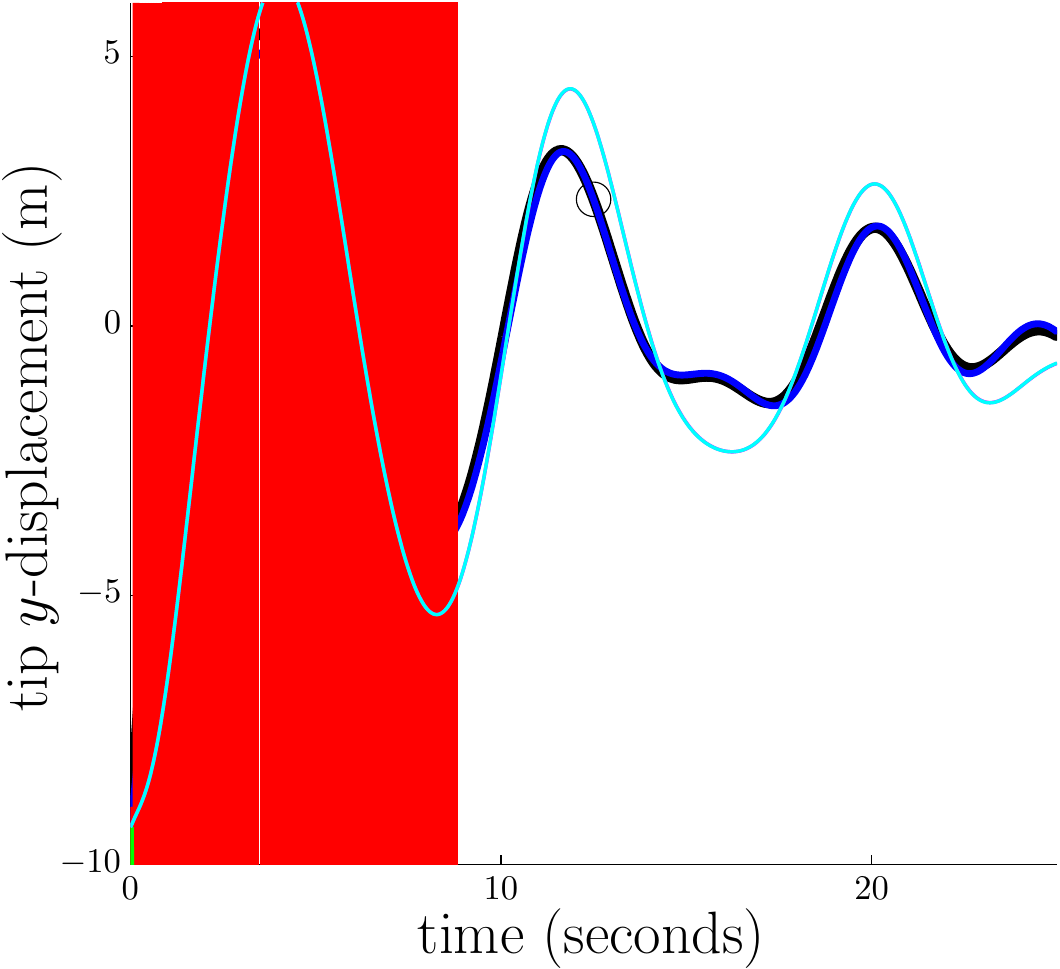}
\includegraphics[width=.3\textwidth]{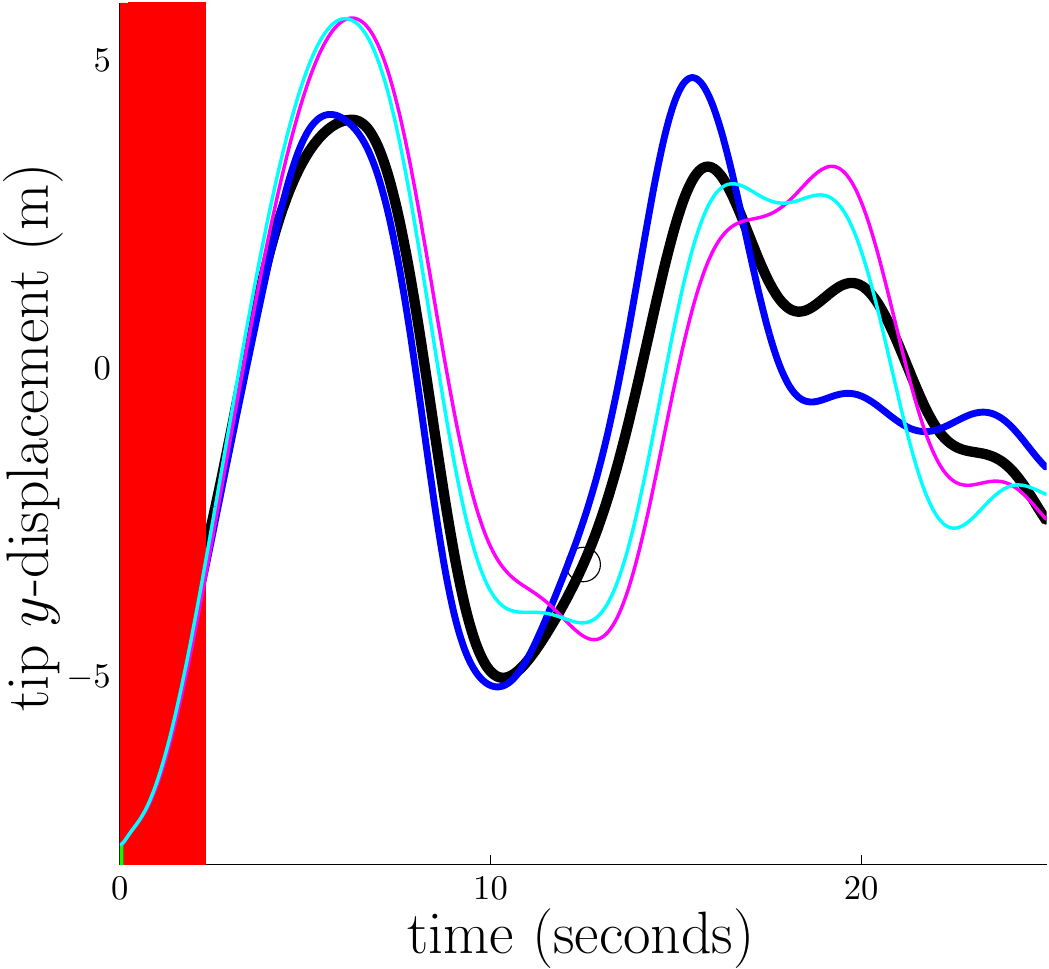}}
\subfigure[predictions at three randomly chosen online points, $60\%$ sampling]
{\includegraphics[width=.3\textwidth]{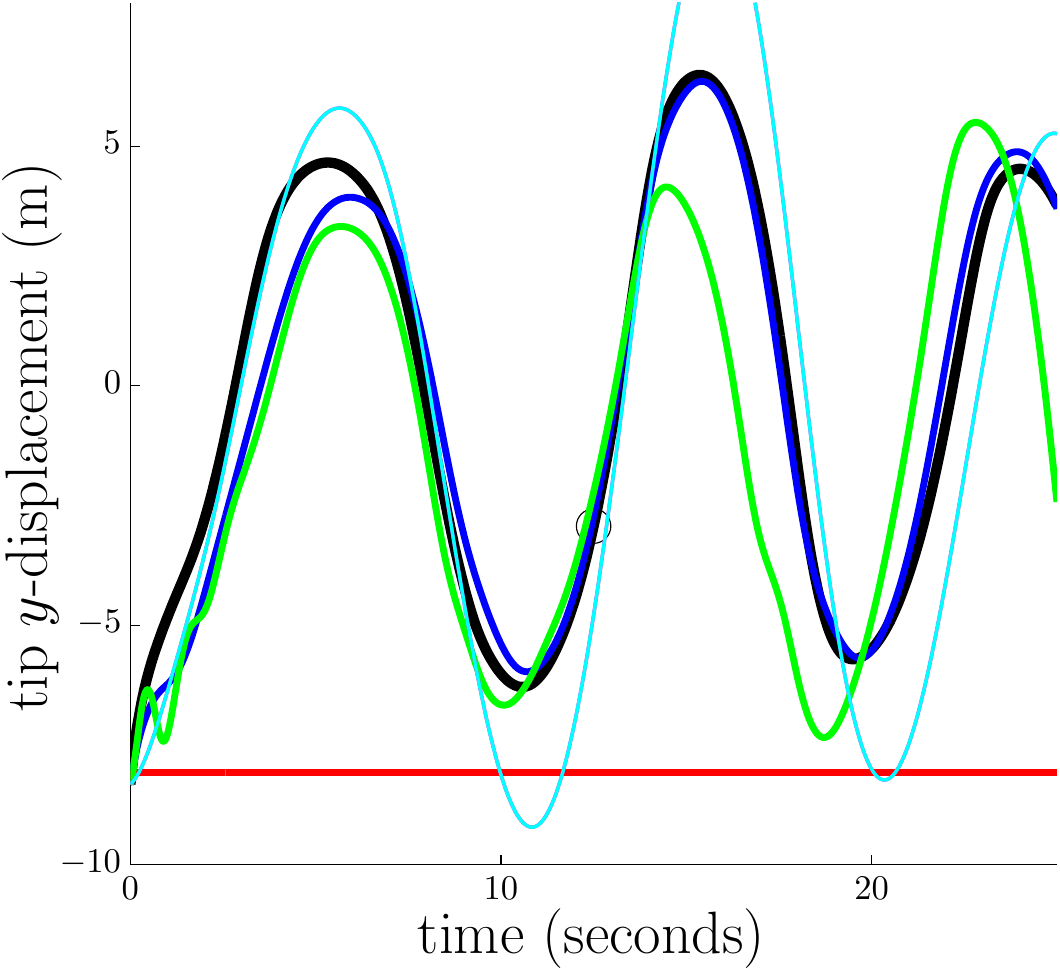}
\includegraphics[width=.3\textwidth]{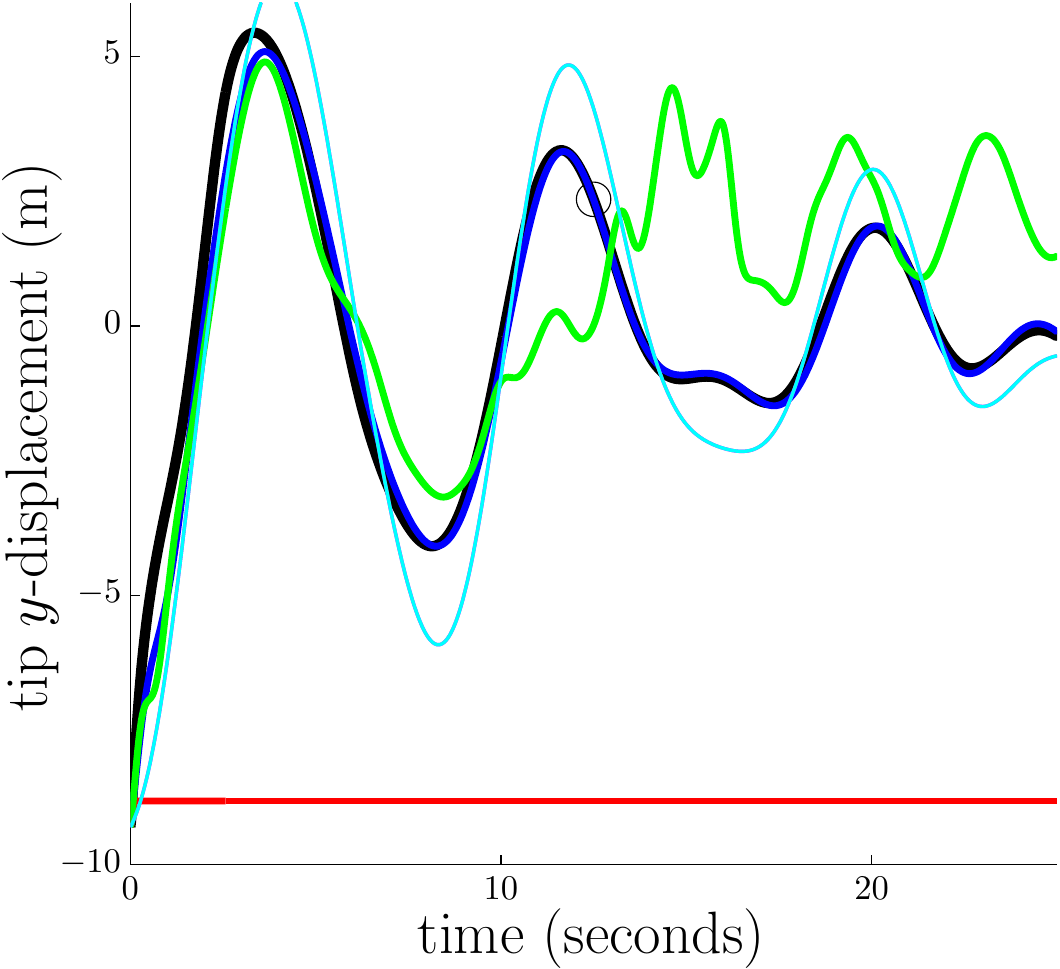}
\includegraphics[width=.3\textwidth]{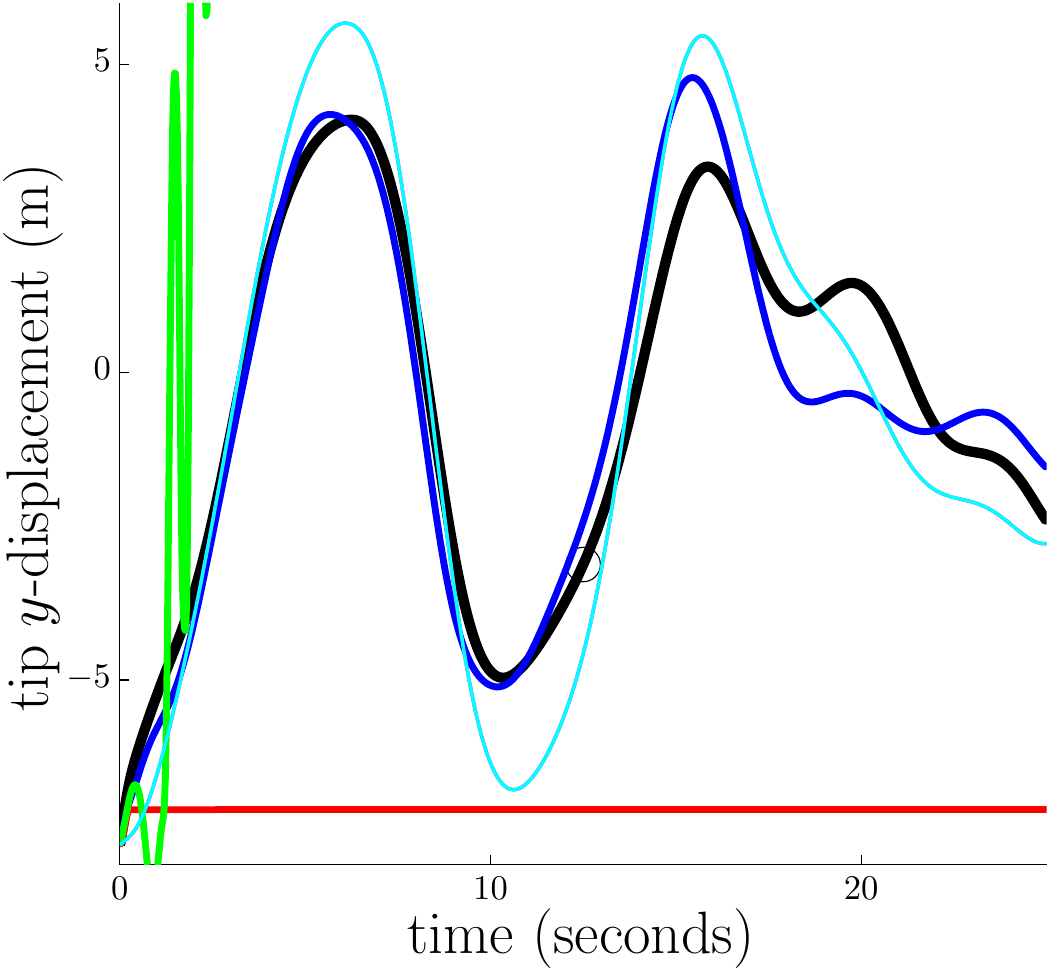}}
\subfigure[predictions at three randomly chosen online points, $80\%$ sampling]
{\includegraphics[width=.3\textwidth]{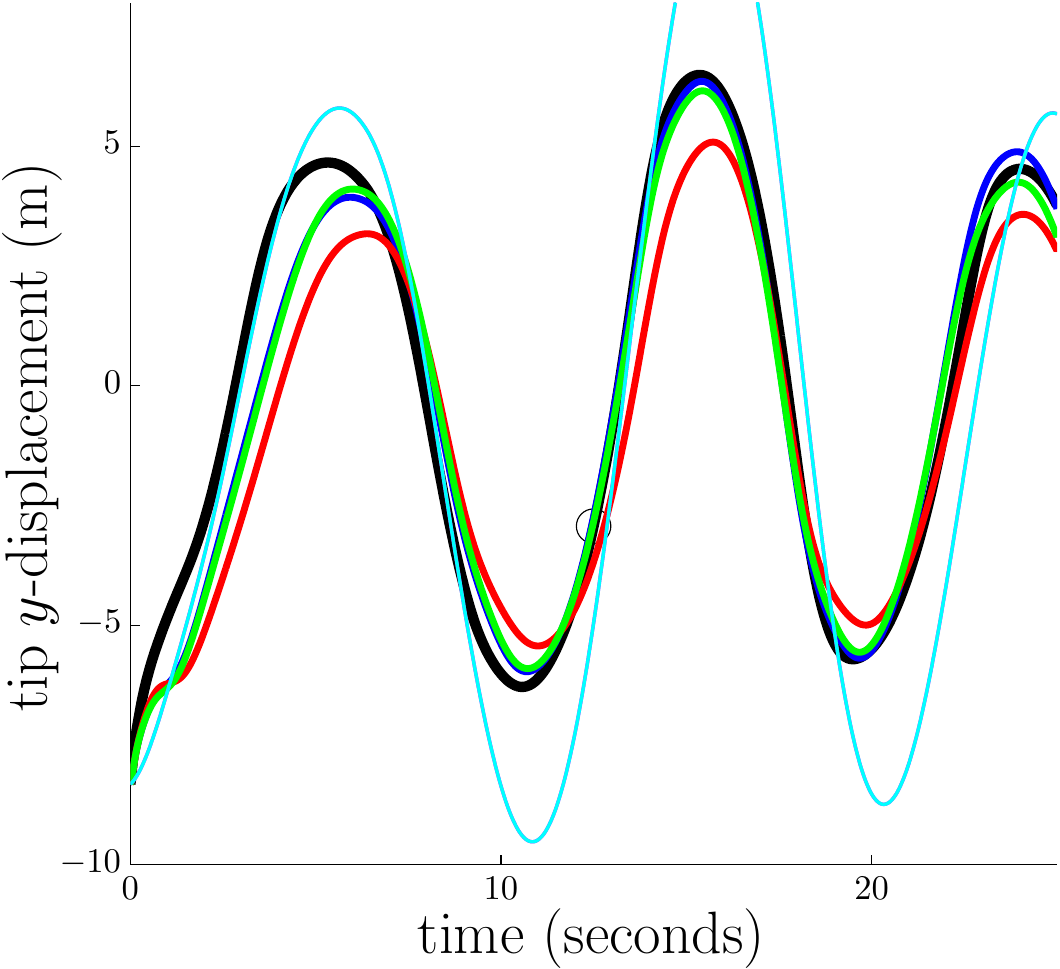}
\includegraphics[width=.3\textwidth]{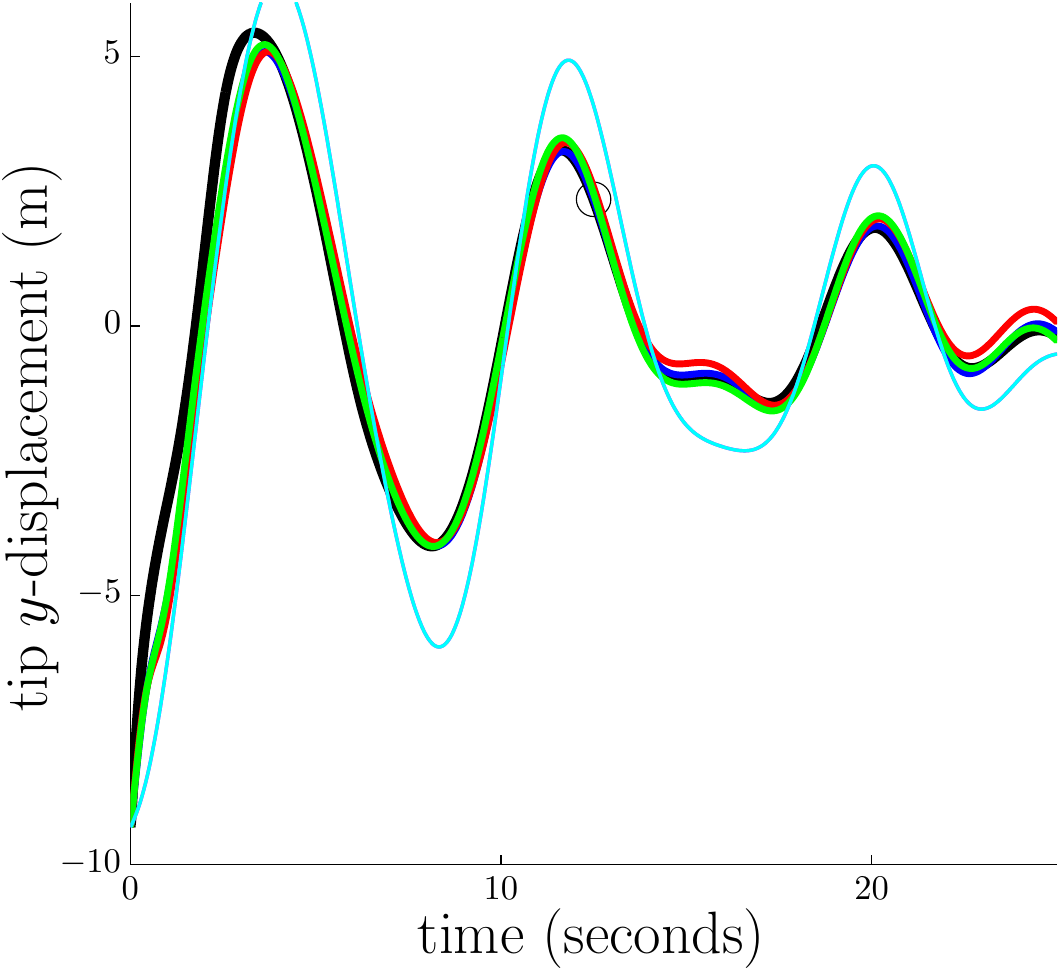}
\includegraphics[width=.3\textwidth]{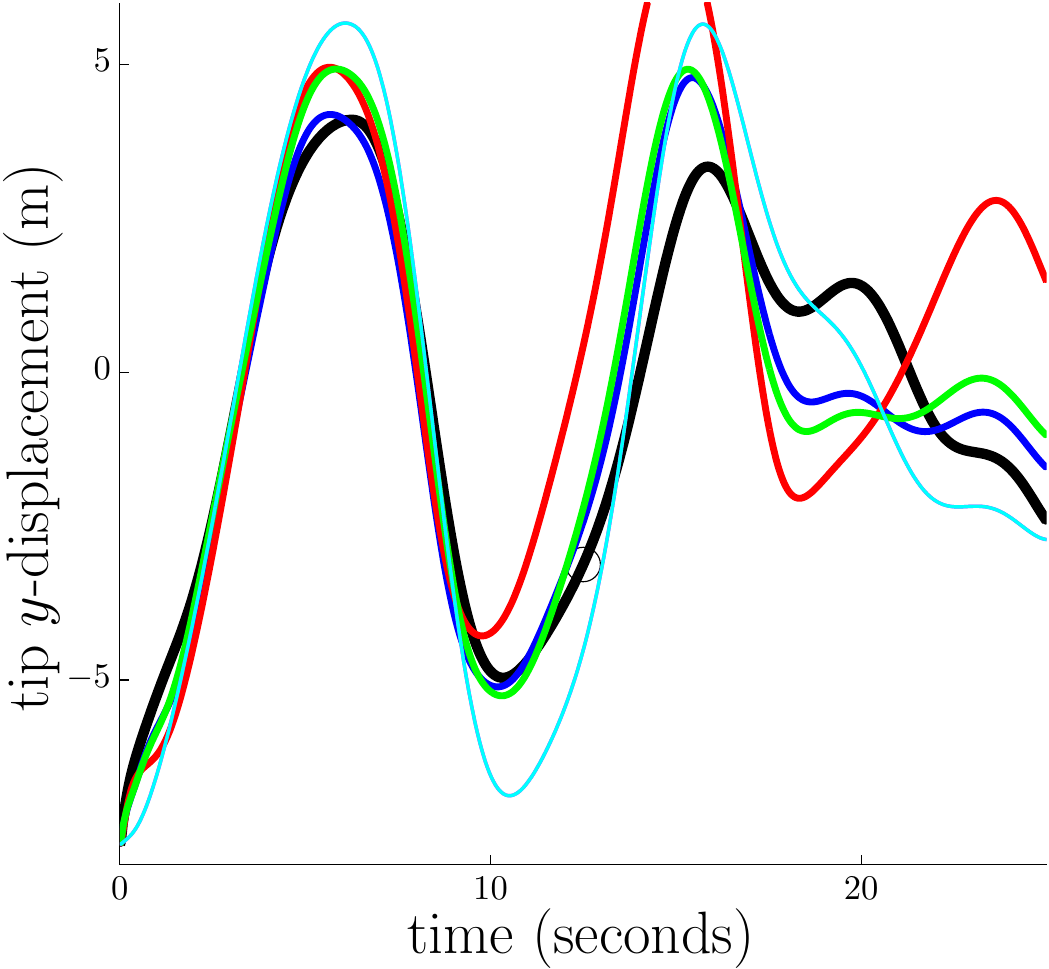}}
\caption{Non-conservative, highly nonlinear parameter-varying case: reduced-order model responses as
a function of sampling percentage $\nsample/N\times 100\%$. Legend: full-order model (black),
Galerkin ROM (dark blue), structure-preserving ROM method 1 (magenta), structure-preserving ROM method 2 (light blue), gappy POD ROM
(red), collocation ROM (green), end of training time interval (black circle).}\label{fig:nonConsNonlinPredictROM}
\end{figure}

\begin{figure}[htbp]
\centering
\subfigure[online point 1]
{\includegraphics[width=.48\textwidth]{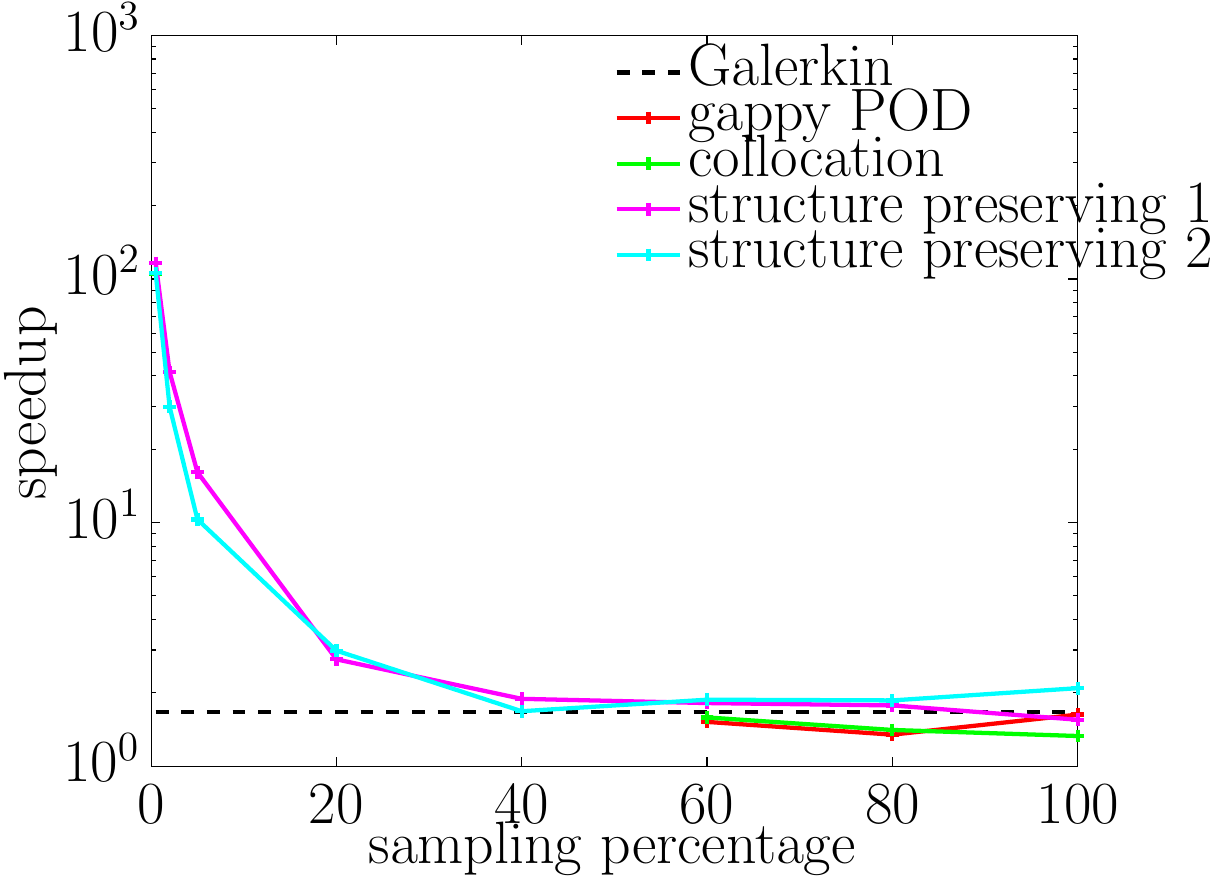}
\includegraphics[width=.48\textwidth]{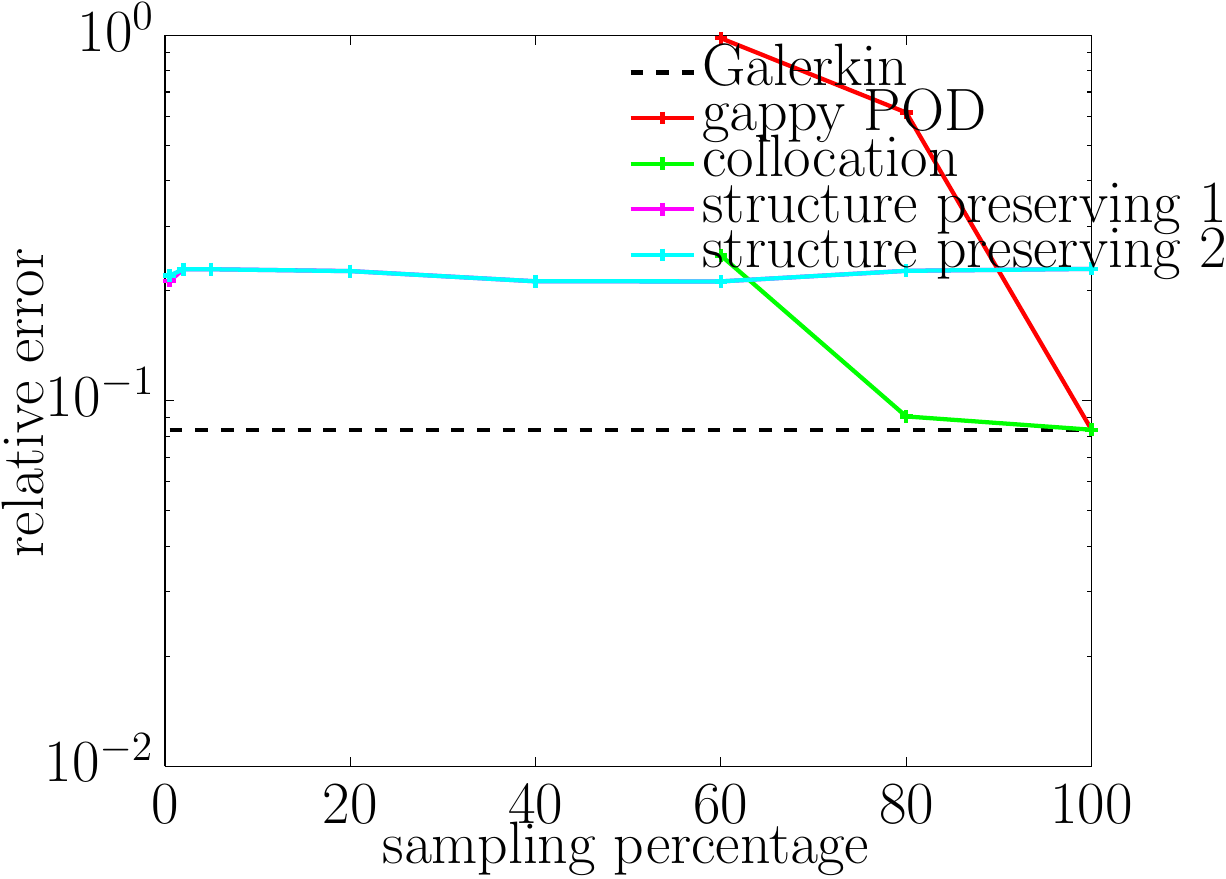}}
\subfigure[online point 2]
{\includegraphics[width=.48\textwidth]{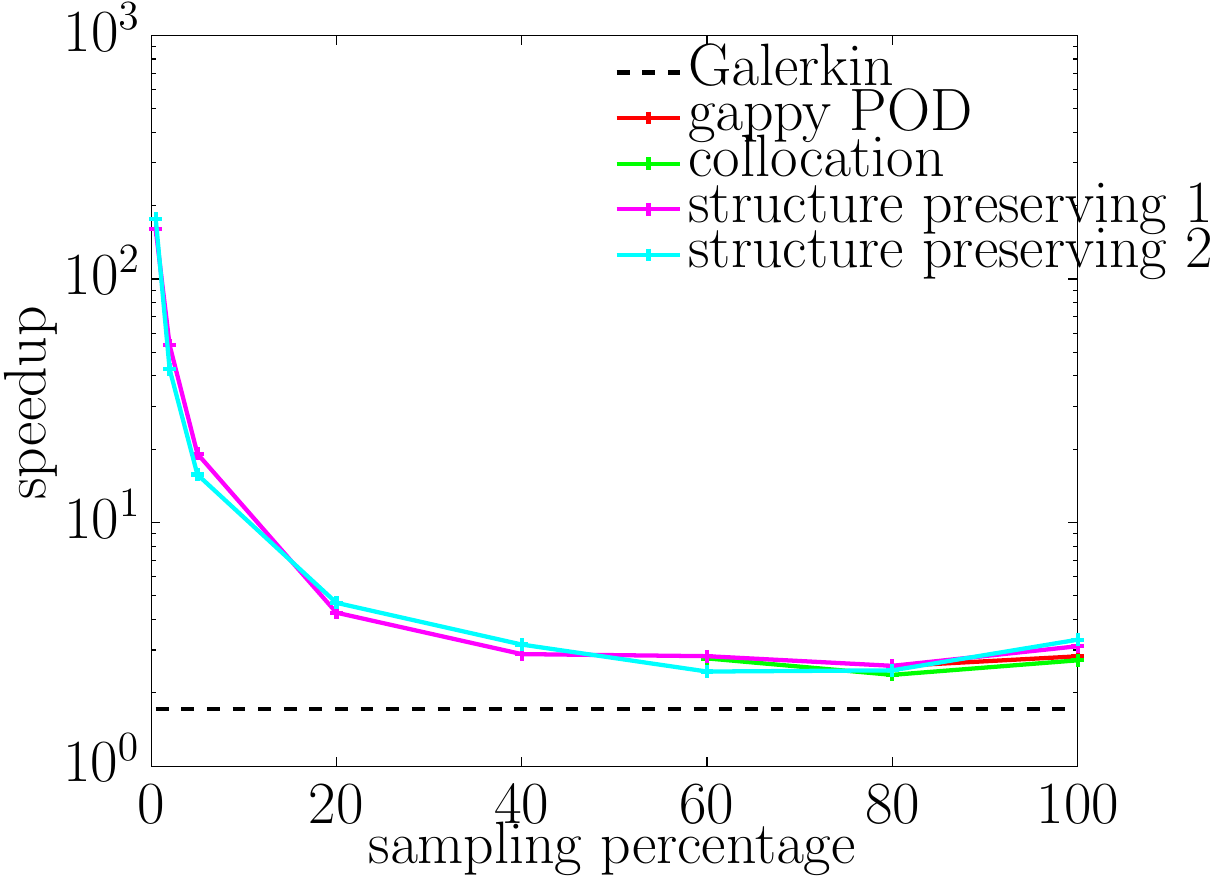}
\includegraphics[width=.48\textwidth]{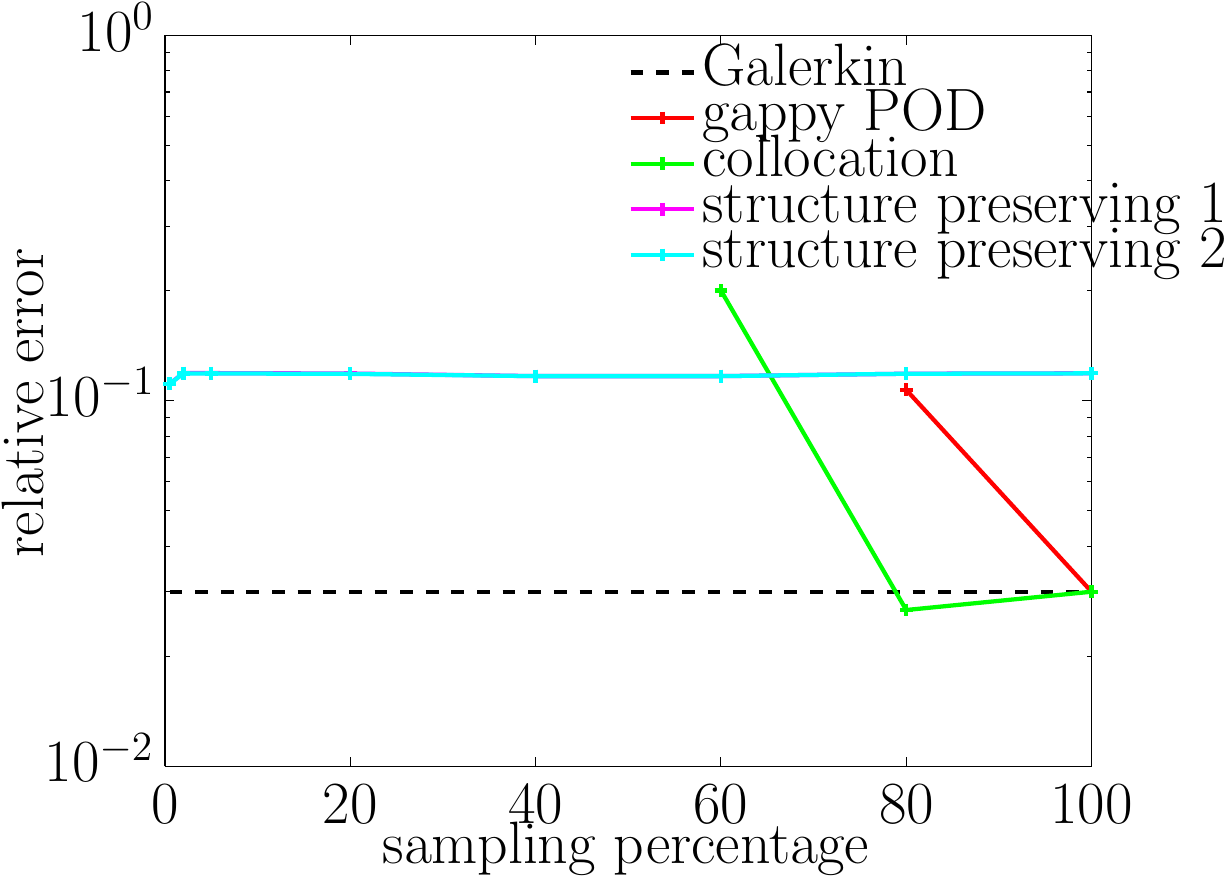}}
\subfigure[online point 3]
{\includegraphics[width=.48\textwidth]{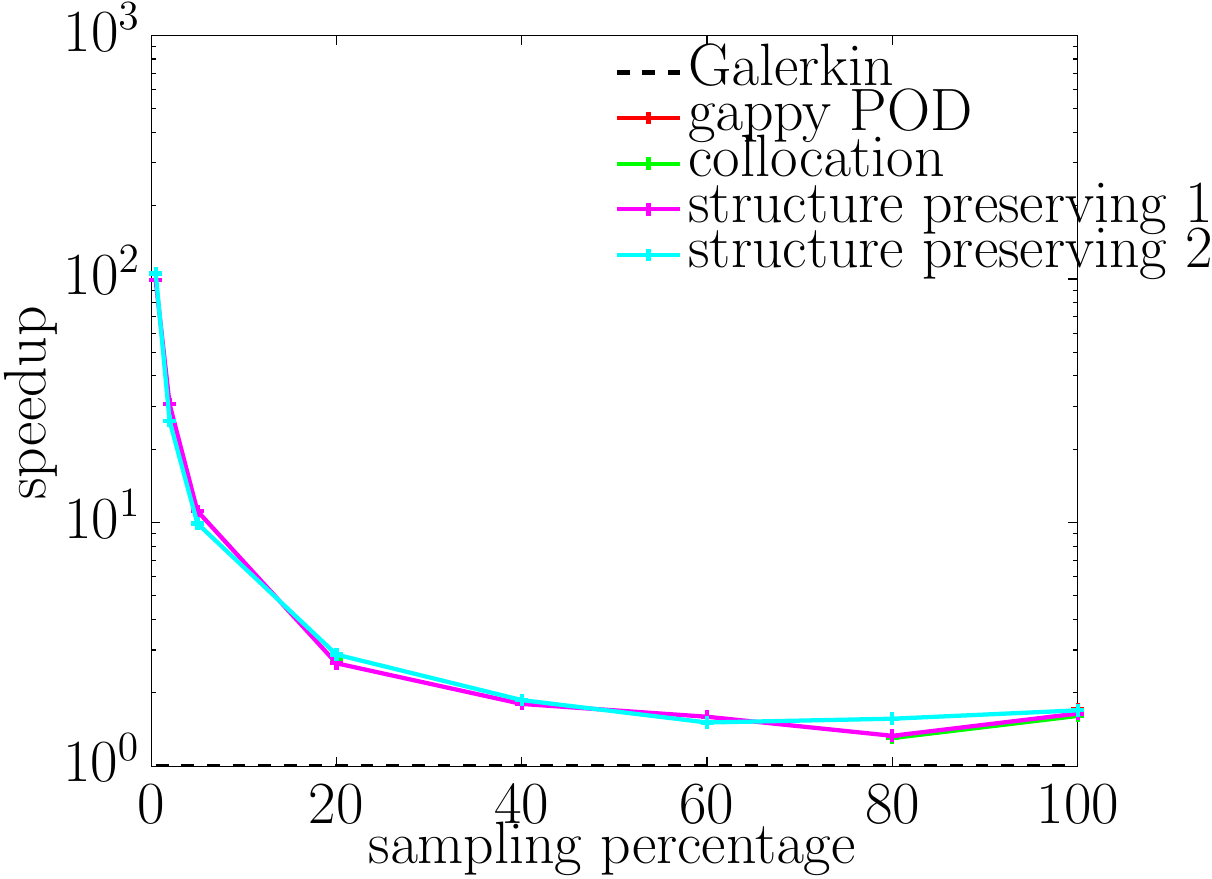}
\includegraphics[width=.48\textwidth]{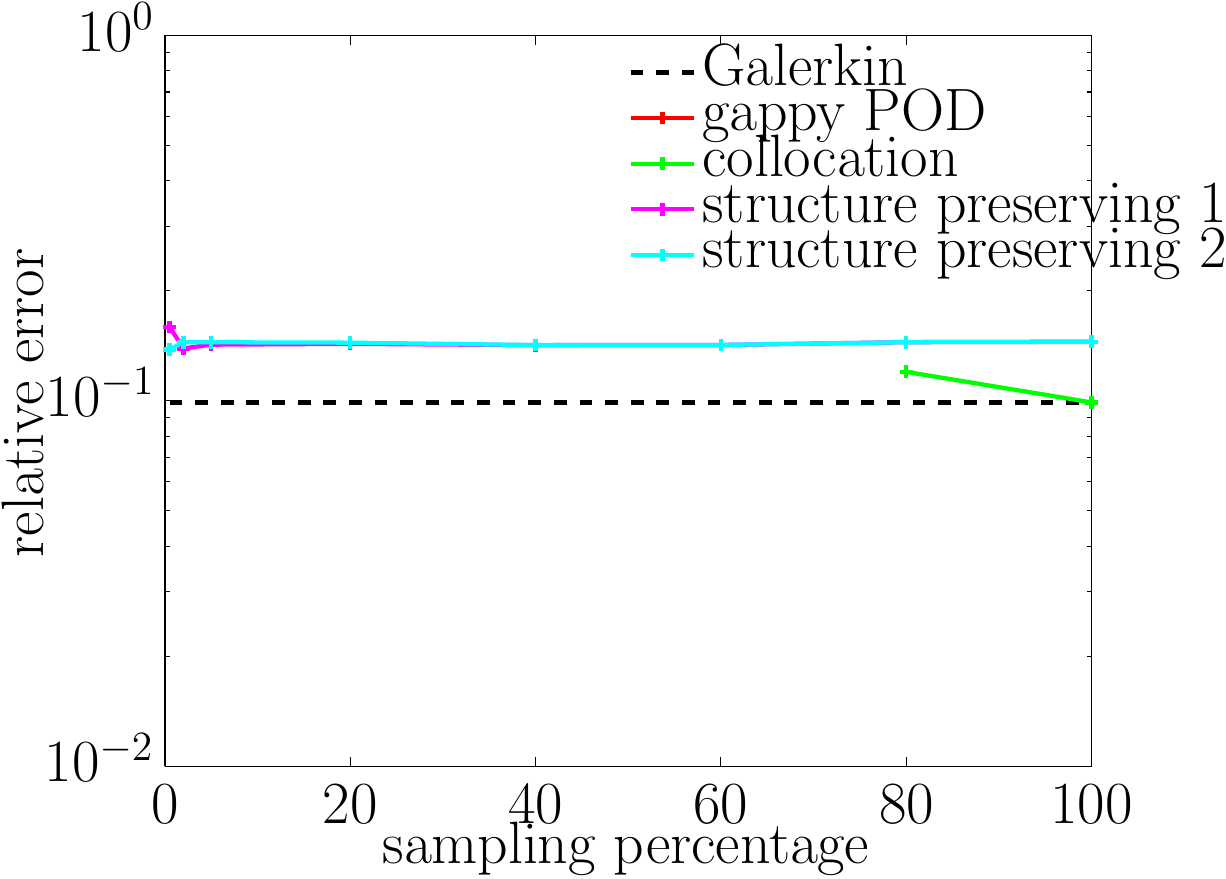}}
\caption{Non-conservative, highly nonlinear parameter-varying case: reduced-order model performance as
a function of sampling percentage $\nsample/N\times 100\%$.}\label{fig:nonconsNonlinPredictPerformError}
\end{figure}

\section{Conclusions}\label{sec:conclusions}
This paper has presented an efficient structure-preserving model-reduction
strategy applicable to simple mechanical systems. The methodology directly approximates
the quantities that define the problem's Lagrangian structure and subsequently
derives the equations of motion, while ensuring
low online computational cost. The method is distinct from typical model-reduction
methods for nonlinear ODEs; these methods are typically based on collocation
and DEIM/gappy POD techniques that approximate the equations of motion and destroy Lagrangian structure.
At the core of the methodology are the reduced-basis sparsification (RBS) and
matrix gappy POD techniques for approximating parameterized reduced matrices
while preserving symmetry and positive definiteness; we also employed the former method to
	preserve potential-energy structure.

Numerical experiments on a geometrically nonlinear parameterized truss
structure highlight the method's benefits: preserving Lagrangian structure
ensured the method always generated stable responses that were often very
accurate. Other model-reduction
techniques were often unstable; achieving stability usually required too many
sample indices to lead to significant performance gains for those methods. The
experiments also showed that both RBS and matrix gappy POD led to nearly the
same performance across a range of experiments. 

Future work includes devising a method to improve the method's robustness in
the presence of strong nonlinearity (e.g., by non-local approximation of the
potential-energy function), applying the method to a
truly large-scale problem, devising a technique-specific method for choosing
the sample indices, and deriving error bounds and error estimates that
rigorously assess the accuracy of the method's predictions.  Finally, the RBS
and matrix gappy POD methods are relevant to a wider class of problems than
model reduction for Lagrangian systems; future work will investigate to their
applicability to other scenarios.

\appendix

\section{Hamiltonian dynamics}\label{app:Hamiltonian}
When the Hamiltonian formulation of classical mechanics is taken, the proposed
structure-preserving reduced-order models also preserve problem structure. For
simplicity, we consider conservative systems with no dissipation or applied
external forces. The Hamiltonian ingredients for conservative simple
mechanical systems are then the same as those for Lagrangian dynamical
systems as described in Section \ref{sec:intro}:
 \begin{itemize} 
  \item 
A differentiable configuration manifold $\Q$, which we set to $\Q =
\RR{N}$.
\item A parameterized Riemannian metric $g(\vVec , \w ;\param)$, which we set
to 
$g(\vVec , \w ; \param) = \vVec ^T\Mparam \w $, 
where $\Mparam$ denotes the $N\times N$ parameterized symmetric positive-definite mass matrix.
 \item A parameterized potential-energy function 
 $\potEn:\Q\times \paramDomain\rightarrow \RR{}$.
	 \end{itemize}
Again, the kinetic energy can be expressed as $T(\dot\q;\param) = \half g(\dot \q ,\dot \q
;\param) = \half\dot \q ^T\Mparam\dot \q $, and the Lagrangian becomes
	$L(\q ,\dot \q ;\param) = \half\dot \q ^T\Mparam\dot \q  - \potEn(\q
	;\param)$.

The conjugate momenta $\p:\left[0,\lastT\right]\times
\paramDomain\rightarrow\RR{N}$ can then be derived as
\begin{align}
\p(t;\param)&=\frac{\partial L}{\partial \dot \q}\left(t;\param\right)\\
&=\Mparam\dot\q(t;\param).
\end{align}
In terms of the conjugate momenta, the kinetic energy then becomes
$T_\q(\p;\param) = \p^T\Mparam^{-1}\p$. By definition,
the Hamiltonian $H:\RR{N}\times\RR{N}\times\paramDomain\rightarrow\RR{}$
is the Legendre transformation of the Lagrangian function:
\begin{align}
H(\p,\q;\param) &= \dot\q^T\p - L(\q,\dot\q,t)\\
&= \frac{1}{2}\p^T\Mparam^{-1}\p + V(\q;\param).
\end{align}
Equivalently, it is the \emph{total energy}, or sum of the kinetic and
potential energies for classical mechanical systems.
The equations of motion can then be obtained by applying Hamilton's equations
of motion
 \begin{align} 
\dot\q &= \frac{\partial H}{\partial \p}\\
\dot\p &= -\frac{\partial H}{\partial \q},
  \end{align} 
which are equivalent to 
 \begin{gather} 
\dot\q = \Mparam\q\\
\Mparam\ddot\q  + \nabla_\q V(\q;\param)=0
  \end{gather} 
	from the definition of the Hamiltonian and conjugate momenta. Note that
	these are equivalent to the equations of motion for a simple mechanical system 
	\eqref{eq:sdmEom} in the conservative case derived using the Lagrangian
	formalism.

\subsection{Galerkin reduced-order model}
To construct the structure-preserving Galerkin reduced-order model for the
Hamiltonian formalism, we follow this same recipe, but with the reduced
ingredients defined previously. First, we define the
  reduced configuration space $\Qred\in\RR{\nstate}$ with 
$\Qredfull\equiv \{\qRef + \podstate \qred\ |\ \qred \in \Qred\}$, and
subsequently the reduced Lagrangian from
Eqs.~\eqref{eq:redLagrangian1}--\eqref{eq:redLagrangian}
 \begin{align} 
 \Lred(\qred,\dot\qred;\param)&\equiv L(\qRef + \podstate
 \qred,\podstate\dot\qred;\param)\\
 \label{eq:approxRedLagrangian}&=\frac{1}{2}\dot\qred^T\podstate^T\Mparam\podstate\dot\qred -
 \potEn(\qRef+\podstate\qred;\param).
 \end{align} 
 Then, the reduced conjugate momenta $\pRed:\left[0,\lastT\right]\times
\paramDomain\rightarrow\RR{\nstate}$ can be computed as
\begin{align}
\pRed(t;\param)&=\frac{\partial L_r}{\partial \dot \qr}\left(t;\param\right)\\
&=\podstate^T\Mparam\podstate\dot\qr(t;\param),
\end{align}
and the reduced Hamiltonian is, by definition,
\begin{align}
H_r(\pRed,\qr;\param) &= \dot\qr^T\pRed - L(\qr,\dot\qr,t)\\
&= \frac{1}{2}\pRed^T\left(\podstate^T\Mparam\podstate\right)^{-1}\pRed + V(\qRef+\podstate\qred;\param).
\end{align}
Applying Hamilton's equations of motion then yields
 \begin{align} 
\dot\qr &= \frac{\partial H_r}{\partial \pRed}\\
\dot\pRed &= -\frac{\partial H_r}{\partial \qr},
  \end{align} 
or equivalently
 \begin{gather} 
\dot\qr = \podstate^T\Mparam\podstate\qr\\
\podstate^T\Mparam\podstate\ddot\qr  + \podstate^T\nabla_\q
V(\qRef+\podstate\qred;\param)=0
  \end{gather} 
Again, these are equivalent to the Galerkin reduced-order equations of motion
\eqref{eq:lagrangeGal} obtained from the Lagrangian perspective in the absence of non-conservative
forces. Clearly, this formulation for a reduced-order model preserves
Hamiltonian structure, as the equations of motion are consistent with the
Hamiltonian formalism.

\subsection{Proposed structure-preserving methods}
The proposed methods preserve Hamiltonian structure. The derivation follows
that of the Galerkin reduced-order model in the previous section, but with the
reduced Lagrangian replaced by the \emph{approximated} reduced Lagrangian from
Eq.~\eqref{eq:approxRedLagrangian}:
 \begin{align} 
\Lredapprox(\qred,\dot \qred;\param) &= \half\metricRedApprox(\dot\qred,\dot\qred;\param)
 - \potEnRedApprox(\qred;\param)  \\
&= \half\dot\qred^T\matApproxM\dot\qred
 - \potEnRedApprox(\qred;\param), 
 \end{align} 
 where $\matApproxM$ is an $\nstate\times\nstate$ symmetric positive
 definite matrix generated by one of the methods presented in Section \ref{sec:matrixApprox},
 and $\potEnRedApprox(\qred;\param)$ is an approximated potential approximated
 according to the method outlined in Section \ref{sec:potEnGen}.

The approximated reduced conjugate momenta are then
\begin{align}
\tilde\pRed(t;\param)&=\frac{\partial \Lredapprox}{\partial \dot \qr}\left(t;\param\right)\\
&=\matApproxM\dot\qr(t;\param),
\end{align}
and the approximated reduced Hamilton is
\begin{align}
\tilde H_r(\tilde\pRed,\qr;\param) &= \dot\qr^T\tilde\pRed - L(\qr,\dot\qr,t)\\
&= \frac{1}{2}\tilde\pRed^T\matApproxM^{-1}\tilde\pRed +
\potEnRedApprox(\qred;\param).
\end{align}
Applying Hamilton's equations yields the following equations of motion
 \begin{gather} 
\dot\qr = \matApproxM\qr\\
\matApproxM\ddot\qr  + \nabla_\qr
\potEnRedApprox(\qred;\param)=0,
  \end{gather} 
	which are computationally inexpensive to solve, and also retain
	Hamiltonian structure, as they are consistent with the Hamiltonian
	formalism. Note that again these are equivalent to the equations of motion
	obtained with the proposed strategy in the Lagrangian case with
	non-conservative forces set to zero.

\section{Proper orthogonal decomposition} \label{app:POD}
Algorithm \ref{PODSVD} describes the method for computing a
proper orthogonal decomposition (POD) basis given a set of snapshots. The method
amounts to computing the singular value decomposition of the
snapshot matrix; the left singular vectors define the POD basis.

\begin{algorithm}[htbp]
\caption{Proper-orthogonal-decomposition basis computation (normalized
snapshots)}
\begin{algorithmic}[1]\label{PODSVD}
\REQUIRE Set of snapshots $\snapsNo\equiv\{\snapvec _i\}_{i=1}^\nsnap\subset\RR{N}$,
energy criterion $\energyCrit\in[0,1]$
\ENSURE  $\podArgs{\snapsNo}{\energyCrit}$
\STATE\label{step:SVD} Compute the thin singular value decomposition $
\snapmat= \U \Sigma \V^T $, where $\snapmat\equiv\left[\snapvec _1/\|\snapvec _1\|\ \cdots\
\snapvec _{n_\snapvec }/\|\snapvec _{n_\snapvec }\|\right]$.
\STATE Choose dimension of truncated basis 
$\nstate = \nenergy(\energyCrit)$, where 
 \begin{align} 
 \nenergy(\energyCrit) &\equiv \min_{i\in \mathcal V(\energyCrit)}i\\
 \mathcal V(\energyCrit)&\equiv \{n\in\{1,\ldots,\nsnap\}\ | \
 \sum_{i=1}^n\sigma_i^2/\sum_{j=1}^{\nsnap}\sigma_j^2\geq\energyCrit\},
  \end{align} 
	and $\Sigma \equiv \mathrm{diag}\left(\sigma_i\right)$ with
	$\sigma_1\geq\cdots\geq \sigma_\nsnap\geq 0$.
\STATE $\podArgs{\snapsNo}{\energyCrit}=\vecmat{\leftSing}{\nstate}$, where $\U
\equiv\vecmat{\leftSing}{\nsnap}$.
\end{algorithmic}
\end{algorithm}

\section{Solving the matrix gappy POD optimization
problem}\label{app:optimization}
This approach reformulates the constraints of problem \eqref{eq:optProblem} 
in terms of eigenvalues of the reduced matrix.
That is, problem \eqref{eq:optProblem} is reformulated as
\begin{align}\label{eq:optProblemMod}
\begin{split}
\underset{x\equiv\left(x_1,\ldots,x_\nA\right)}{\mathrm{minimize}}\quad&
\|
\sampleMatT\Aparam\sampleMat 
-
\sum\limits_{k=1}^\nA
\sampleMatT\Abasis{k} \sampleMat x_k
\|_F^2\\
\mathrm{subject\ to}\quad &\approxlambda_j(x)\geq \epsilon>0, \quad j =
1,\ldots,\nstate.\\
\end{split}
\end{align}
Here, $\approxlambda_j(x)$, $j=1,\ldots,n$ are the eigenvalues of 
the low-dimensional matrix $\sum\limits_{i=1}^\nA \podstate^T\Abasis{i}\podstate
x_i$ and $\epsilon$ denotes a numerical threshold for defining a full-rank matrix.
This problem can be numerically solved, e.g.,  using a gradient-based algorithm. 

The gradient of the quadratic objective function is obvious. The gradient of the
constraint can be derived by assuming distinct eigenvalues:
 \begin{align} \label{eq:eigDerivDistinct1}
\frac{\partial
\approxlambda_j}{\partial x_i} &= \approxeig_j^T\frac{\partial\left(\sum\limits_{k=1}^\nA \podstate^T\Abasis{k}\podstate
x_k\right)}{\partial x_i}\approxeig_j\\
\label{eq:eigDerivDistinct2}&= \approxeig_j^T\left(\podstate^T\Abasis{i}\podstate
\right)\approxeig_j.
 \end{align} 
Here, $\approxeig_j$ is the eigenvector associated with eigenvalue
$\approxlambda_j$. This indicates that computing the gradient $\frac{\partial
\approxlambda_j}{\partial x_i}$ is inexpensive and requires the
following steps:
 \begin{enumerate} 
 \item Compute the eigenvector $\approxeig_j\in\RR{n}$  of the  matrix
 $\sum\limits_{k=1}^\nA \podstate^T\Abasis{k}\podstate
x_k$.
 \item Compute the low-dimensional matrix--vector product
$\w  = \left(\podstate^T\Abasis{i}\podstate
\right)\approxeig_j$.
 \item Compute the low-dimensional vector--vector product $\approxeig_j^T\w$.
 \end{enumerate}

We propose using the unconstrained solution to problem \eqref{eq:optProblemMod} as the
initial guess. In practice, this solution is often feasible, so it is
typically unnecessary to handle the constraints directly.  In the rare
cases where it is necessary to deal with 
multiple equal eigenvalues---or a number of nearby eigenvalues---the methods
presented by Andrew and Tan \cite{AT98} can be used to produce a numerically stable gradient of the
constraint; this was not required in the numerical experiments reported in
Section \ref{sec:experiments}.


\section{Proofs}

\subsection{Proof of Theorem \ref{thm:quadCase}}\label{app:quadCase}

The proof relies on a generalization of the 
well known Cauchy interlacing thereom. To prove this generalization, we 
use a
theorem (Theorem 4.3.10) from Ref.~\cite{horn1990matrix} that we restate
below.
\begin{theorem}
Let two sequences of interlacing real numbers be given by
$(\lambda^{(r)}_i)_{i=1}^\nstate$
and $(\lambda^{(s)}_i)_{i=1}^\nsample$ as described  by inequality~\eqref{new interlace}
when $ \nsample = \nstate +1$.
Define $\LambdaMat^{(r)} = \diag(\lambda^{(r)}_i)$ and
$\LambdaMat^{(s)} = \diag(\lambda^{(s)}_i)$. 
Then, there exists
a real number $\alpha\in\RR{}$ and a vector $\y \in \RR{\nstate} $ such that 
$\LambdaMat^{(s)}$ 
are the eigenvalues of the real symmetric matrix 
$$
\hat{\B }^{(\bordered)} \equiv \left (
\begin{matrix}
\LambdaMat^{(r)} & \y \cr \y^T & \alpha \end{matrix} \right ) .
$$
\end{theorem}
The following corollary is a direct consequence of the above theorem.
\begin{corollary}
Given $\hat{\B }^{(s)} \in \spd{\nsample}$ 
and   $\hat{\B }^{(r)} \in \spd{\nsample-1}$, where $\spd{k}$ denotes
the set of $k\times k$ symmetric
 positive-definite matrices,
whose eigenvalues interlace, then
\begin{equation} \label{the corollary}
\exists~\U _{\nsample} ~\mbox{such that }
\U _{\nsample}^T \hat{\B }^{(s)} \U _{\nsample} = \hat{\B }^{(r)}
~~~~\mbox{with}~~~~\U_{\nsample}^T\U_{\nsample}=I.
\end{equation}
\end{corollary}
\begin{proof}
Using the above theorem, a matrix $\hat{\B }^{(\bordered)}\in\spd{\nsample}$ exists
that shares the same eigenvalues with $\hat{\B }^{(s)}$.
Let $\Qmat ^{(\bordered)}$, 
$\Qmat ^{(s)}$, and $\Qmat ^{(r)}$ be the (square, orthogonal) matrices of eigenvectors for 
$\hat{\B }^{(\bordered)}$, 
$\hat{\B }^{(s)}$, and $\hat{\B }^{(r)}$, respectively. 
Then, 
$$
     \left(\Qmat ^{(\bordered)}\right)^T \hat{\B }^{(\bordered)} \Qmat
		 ^{(\bordered)}   = \LambdaMat^{(s)}  = \Qmat ^{(s)}\hat{\B }^{(s)}
		 (\Qmat ^{(s)})^T,
$$
which implies that
$$
\hat{\B }^{(\bordered)} = 
  \Qmat ^{(\bordered)}(\Qmat ^{(s)})^T\hat{\B }^{(s)}\Qmat ^{(s)}(\Qmat ^{(\bordered)})^T .
$$
From the definition of $\hat{\B }^{(\bordered)}$ it also follows that
\begin{displaymath}
  \IO \hat{\B }^{(\bordered)}\IOT = \LambdaMat^{(r)},
\end{displaymath}
where $I$ is the $(\nsample-1)\times(\nsample-1)$ identity matrix and
$\mathbf{0}$ is the zero column vector of length $\nsample-1$, and thus 
\begin{equation}\label{eq:br}
\hat{\B }^{(r)} = 
  \Qmat ^{(r)}\LambdaMat^{(r)}(\Qmat ^{(r)})^T = 
  \Qmat ^{(r)}\left[ \IO \hat{\B }^{(\bordered)}\IOT \right] (\Qmat ^{(r)})^T .
\end{equation}
Combining the above, we can write
\begin{displaymath}
  \hat{\B }^{(r)} = \Qmat ^{(r)} \IO
  \Qmat ^{(\bordered)}(\Qmat ^{(s)})^T\hat{\B }^{(s)}\Qmat ^{(s)}(\Qmat ^{(\bordered)})^T \IOT (\Qmat ^{(r)})^T
\end{displaymath}
and so \eqref{the corollary} is satisfied taking
$\U _{\nsample} = \left[\Qmat ^{(r)} \IO
  \Qmat ^{(\bordered)}(\Qmat ^{(s)})^T\right]^T$.
\end{proof}

The generalization of the Cauchy interlacing thereom now follows.
\begin{theorem}
Given two matrices $\hat{\B }^{(s)} \in \spd{\nsample} $ and 
$\hat{\B }^{(r)} \in \spd{\nstate } $ (with $\nsample \ge \nstate $), then
\begin{equation} \label{reformulated exactness}
\exists~\U  ~~\mbox{such that}~~~
\U ^T \hat{\B }^{(s)} \U  = \hat{\B }^{(r)}  ~~~~\mbox{with}~~~~ \U ^T \U  = \I 
\end{equation}
if and only if the eigenvalues ${\lambda}_i^{(r)}$, $i=1,\ldots,\nstate$
	interlace the eigenvalues ${\lambda}_i^{(s)}$, $i=1,\ldots,\nsample$
	defined as
\begin{align} \label{reformulated eigenvalues}
\hat{\B }^{(r)} \hat{\x}_i^{(r)} &= {\lambda}_i^{(r)} \hat{\x}_i^{(r)},\quad
i=1,\ldots,\nstate\\
\hat{\B }^{(s)} \hat{\x}_i^{(s)} &= {\lambda}_i^{(s)} \hat{\x}_i^{(s)} ,\quad
i=1,\ldots,\nsample.
\end{align}
The definition of interlacing is given by inequality~\eqref{new interlace}.
\end{theorem}
\begin{proof}
It is well known that if $\hat{\B }^{(s)} \in \spd{\nsample} $ is given
along with an orthogonal $\nsample\times\nstate$ matrix $\U$ (with $\nsample \ge \nstate $),
then the eigenvalues of $\U ^T \hat{\B }^{(s)} \U  $ must interlace those
of $\hat{\B }^{(s)}$. This is referred to as the 
Cauchy interlacing 
theorem (e.g., see \cite{parlett1980symmetric}).

The converse of the Cauchy interlacing theorem is less widely known.
The case $\nsample = \nstate $ follows trivially using an eigenvalue decomposition.
The case $\nsample = \nstate +1$ corresponds to the above corollary.
The proof is completed by generalizing the corollary to 
the $\nsample > \nstate +1$ case.
This follows from an inductive argument where one 
considers a projection that reduces the matrix dimension of $\hat{\B }^{(s)}$ by one. 
According to the above corollary, we have a great deal of flexibility in choosing this lower dimensional
matrix if its eigenvalues interlace those of the higher dimension matrix. We then choose a lower dimensional
matrix whose eigenvalues not only interlace those of $\hat{\B }^{(s)}$ but whose eigenvalues are also
interlaced by those of $\hat{\B }^{(r)}$. That is,
$$
\lambda_i^{(s)} \le \mu_i \le \lambda_{i+1}^{(s)} 
~~~\mbox{and}~~~
\mu_i \le \lambda_i^{(r)} \le \mu_{i+\nsample-\nstate -1},
$$
where $\mu_i$ denotes the $i$th smallest eigenvalue
of the intermediate matrix.
Rewriting this we obtain the following intervals for the eigenvalues
$\mu_i$:
$$
\mu_i \ge 
\left \{
\begin{array}{cc}
\max(\lambda_i^{(s)},\lambda_{i-\nsample+\nstate +1}^{(r)})  & i \ge \nsample-\nstate  \\
\lambda_i^{(s)}  & i < \nsample-\nstate  
\end{array}
\right .
$$
and
$$
\mu_i \le 
\left \{
\begin{array}{cc}
\min(\lambda_{i+1}^{(s)}, \lambda_i^{(r)}) & i \le \nstate  \\
\lambda_{i+1}^{(s)} & i > \nstate  
\end{array}
\right .  .
$$
Using the interlacing property for 
$\hat{\B }^{(s)}$ and $\hat{\B }^{(r)}$,
one can verify that the intervals for the $\mu_i$ are nonempty.
That is, $\lambda_{i+1}^{(s)}  \ge \lambda_i^{(s)}$
and for those $i$ such that $\lambda_{i-\nsample+\nstate +1}^{(r)}$ is
defined, we have
$
\lambda_{i+1}^{(s)}  \ge \lambda_{i-\nsample+\nstate +1}^{(r)},~
\lambda_i^{(r)}   \ge \lambda_i^{(s)}, ~\mbox{and}~
\lambda_i^{(r)}   \ge \lambda_{i-\nsample+\nstate +1}^{(r)} $
Thus, there exists an orthogonal matrix $\U _{\nsample}$ such that the
$(\nsample \hskip -.03in -\hskip -.03in 1) \times (\nsample\hskip -.03in
-\hskip -.03in 1)$ matrix $ \U _{\nsample}^T \hat{\B }^{(s)} \U _{\nsample} $
has eigenvalues that 
are interlaced by those
of $\hat{\B }^{(r)}$.
We repeat this procedure each time reducing the matrix dimension by one until
the final reduction where we take the 
lower dimension matrix to be equal to $\hat{\B }^{(r)} $.  This implies that there exists a set of projection
matrices such that $\U ^T \hat{\B }^{(s)} \U  $ is equal to  $\hat{\B }^{(r)} $ where
$\U  = \U _{\nsample} \U _{\nsample-1} ... \U _{\nstate +1}$.

\end{proof}

Equipped with the generalized Cauchy interlacing theorem, we now
prove the exactness condition for the $\Aparam$ term which is
restated in 
slightly simplified notation.
\begin{theorem}
Let $\Aparam $ have the form
\begin{equation}
\Aparam = \honeparam~\Aone + \htwoparam~\Atwo
\end{equation}
where $\Aone \in \spd{N} $, $\Atwo \in \spd{N} $,
and $\hone,\htwo: \paramDomain \rightarrow \RR{} $. 
\vskip .1in
\noindent
Then, 
\begin{equation} \label{appendix quadratic exact}
\exists~\sparseASmall ~~\mbox{such that}~~~
\sparseASmallT\myspace\sampleMatT\myspace\Aparam\myspace\sampleMat\myspace\sparseASmall
= \podstate^T \myspace\Aparam\myspace\podstate,
~~~~\forall \param \in \paramDomain
\end{equation}
if and only if the eigenvalues of the general matrix pencil
\begin{equation} \label{Br eigenvalues} 
\B ^{(r)} \x_i^{(r)} = \lambda_i^{(r)} \D^{(r)} \x_i^{(r)},\quad
i=1,\ldots,\nstate
\end{equation}
interlace the eigenvalues of 
\begin{equation} \label{Bs eigenvalues}
\B ^{(s)} \x_i^{(s)} = \lambda_i^{(s)} \D^{(s)} \x_i^{(s)} ,\quad
i=1,\ldots,\nsample
\end{equation}
where
$$
\begin{array}{lllllll}
\D^{(r)}&=&\left [ \podstate^T \Aone \podstate \right ] , &
&\D^{(s)}&=&\left [ \sampleMatT \Aone \sampleMat \right ] , \\[10pt]
\B ^{(r)}&=&\left [ \podstate^T \Atwo \podstate \right ] ,&  ~~\mbox{and}~~
&\B ^{(s)}&=&\left [ \sampleMatT \Atwo \sampleMat \right ] .
\end{array}
$$
The definition of interlacing is given by
\begin{equation} \label{appendix new interlace}
\lambda^{(s)}_i \le \lambda^{(r)}_i \le \lambda^{(s)}_{i+\nsample-\nstate } 
~~~for~i=1,...,\nstate  
\end{equation} 
where
the eigenvalues are indexed in order of
increasing magnitude.
\end{theorem}
\begin{proof}
Clearly \eqref{appendix quadratic exact} can only hold for any $\param \in
\paramDomain$ and any functions $\honeparam$ and $\htwoparam$ if and only if
\begin{equation} \label{two matrix match}
\sparseASmallT \D^{(s)} \sparseASmall = \D^{(r)}
~~\mbox{and}~~
\sparseASmallT \B ^{(s)} \sparseASmall = \B ^{(r)} .
\end{equation}
Using a carefully chosen linear transformation, it follows that
proving the theorem is equivalent to proving the following:
\begin{equation} \label{final reformulated exactness}
\exists~\U  ~~\mbox{such that}~~~
\U ^T \hat{\B }^{(s)} \U  = \hat{\B }^{(r)}  ~~~~\mbox{with}~~~~ \U ^T \U  = \I 
\end{equation}
if and only if the eigenvalues $\lambda_i^{(r)}$ interlace the eigenvalues of
	$\lambda_i^{(s)}$, where the eigenvalues previously defined in
	Eqs.~\eqref{Br eigenvalues}--\eqref{Bs eigenvalues} also satisfy
	\begin{align} \label{final reformulated eigenvalues}
	\begin{split}
\hat{\B }^{(r)} \hat{\x}^{(r)} &= {\lambda}^{(r)} \hat{\x}^{(r)}\\
\hat{\B }^{(s)} \hat{\x}^{(s)} &= {\lambda}^{(s)} \hat{\x}^{(s)}.
	\end{split}
\end{align}
The linear transformation relies on Cholesky factorizations
given by 
$\D^{(s)} = \Lmat ^{(s)} (\Lmat ^{(s)})^T $ and $\D^{(r)} = \Lmat ^{(r)} (\Lmat ^{(r)})^T $.
These factors lead to the following definitions 
$$
\begin{array}{llllll}
\hat{\B }^{(s)} &=& (\Lmat ^{(s)})^{-1} \B ^{(s)} (\Lmat ^{(s)})^{-T}, &
\hat{\x}^{(s)} &=& (\Lmat ^{(s)})^T \x^{(s)} , \\[5pt]
\hat{\B }^{(r)} &=& (\Lmat ^{(r)})^{-1} \B ^{(r)} (\Lmat ^{(r)})^{-T}, & 
\hat{\x}^{(r)} &=& (\Lmat ^{(r)})^T \x^{(r)}, ~~\mbox{and}\\[5pt]
\U  &=& (\Lmat ^{(s)})^T \sparseASmall (\Lmat ^{(r)})^{-T},
\end{array}
$$
which can be used in Eqs.~\eqref{Br eigenvalues},
\eqref{Bs eigenvalues} and \eqref{two matrix match}
to obtain Eqs.~\eqref{final reformulated exactness}
and \eqref{final reformulated eigenvalues}.
The proof is completed by recognizing that
Eqs.~\eqref{final reformulated exactness} and \eqref{final reformulated eigenvalues}
correspond to the generalized Cauchy interlacing thereorm.

\REMOVE {
It is well known that if we are given an orthogonal 
matrix $\U  \in \RR{\nsample \times \nstate }$ (with $\nsample > \nstate $) such 
that $\U ^T \hat{\B }^{(s)} \U  = \hat{\B }^{(r)} $, then the eigenvalues of  $\hat{\B }^{(r)} $ must interlace those
of $\hat{\B }^{(s)}$. This is often referred to as the 
Cauchy interlacing theorem\cite{parlett1980symmetric}.
To complete the proof, one must show that the converse of the Cauchy
interlacing theorem is also true. 

That is, given two symmetric matrices
$\hat{\B }^{(s)} $ and $\hat{\B }^{(r)}$ 
 whose eigenvalues interlace (with $\nsample \ge \nstate $), 
there always exist a matrix $\U $ such that 
$\U ^T \hat{\B }^{(s)} \U  = \hat{\B }^{(r)} $.
While the above Cauchy Interlacing Theorem is quite established, its converse
is less well known and can easily be proved 
for the case of 
$\nsample \hskip -.04in = \hskip -.04in \nstate  \hskip -.04in + \hskip -.04in 1$
using Theorem 4.3.10 from \cite{horn1990matrix} which
is stated now with our notation.
}

\end{proof}

\subsection{Proof of Theorem
\ref{matrixGapExact}}\label{sec:matrixGapExactProof}
\begin{proof}
If condition 2 holds, then the unconstrained solution to problem \eqref{eq:optProblem} is
 \begin{equation} \label{eq:gappyUnconSol}
\AcoeffParamOn =
\left(\sampleMatVectorizedT\podVectorizedA\right)^+\sampleMatVectorizedT\vectorize{\AparamOn}.
 \end{equation}

If condition 1 holds, then the vectorized matrix can be expressed as 
 \begin{equation}\label{eq:cond1result} 
\vectorize{\AparamOn} = \podVectorizedA\vectorAparamCoeff,
 \end{equation}
 or equivalently
 \begin{equation}\label{eq:cond1result2} 
\AparamOn = \sum_{i=1}^\nA\vectorAparamCoeffi\Abasisi,
 \end{equation}
 where $\vectorAparamCoeffNo\equiv\entrytuple{\vectorAparamCoeffiNo}{\nA}$.
 Substituting Eq.~\eqref{eq:cond1result} into Eq.~\eqref{eq:gappyUnconSol}
 gives $\AcoeffParamOn = \vectorAparamCoeff$ and so Eq.~\eqref{eq:spApprox}
 yields
 \begin{align} \label{eq:exactApproxMatGap}
\matApproxAOn&=\sum\limits_{i=1}^\nA
\vectorAparamCoeffi\podstate^T\Abasis{i}\podstate.
  \end{align} 
Comparing Eqs.~\eqref{eq:exactApproxMatGap} and \eqref{eq:cond1result2} gives
the exactness result: $\matApproxAOn=\podstate^T\AparamOn\podstate^T$.
\end{proof}

\subsection{Proof of Lemma \ref{quadratic solvability}}\label{app:ray}
\begin{proof}
The equation  $\widetilde{\X }^T \widetilde{\X } = I  $
simply states that the $\nsample $ columns
of $\widetilde{\X }$ are orthogonal and so any orthogonal
matrix $\widetilde{\X }$ satisfies the first part of (\ref{lemma version}). 
Thus, solvability amounts to finding an orthogonal matrix
$\widetilde{\X }$ such that 
$\tilde{\podstate}^T \tilde{\cvector } = \widetilde{\X }^T \sampleMatT   \tilde{\cvector }$.
For a solution to exist, however, it is obviously necessary that
$|| \tilde{\podstate}^T \tilde{\cvector } ||_2 = || \widetilde{\X }^T \sampleMatT   \tilde{\cvector } ||_2$. 
If the vector $\sampleMatT   \tilde{\cvector }$ lies within the span of the basis
defined by the columns of $\widetilde{\X }^T$,
then  $\widetilde{\X }^T \sampleMatT   \tilde{\cvector }$ preserves its 2-norm 
and so $ || \widetilde{\X }^T \sampleMatT   \tilde{\cvector } ||_2 = || \sampleMatT   \tilde{\cvector } ||_2$.
That is, application of $\widetilde{\X }^T $ corresponds to a rotation of
$\sampleMatT   \tilde{\cvector }$ about the origin and so length is preserved. If instead 
the vector $\sampleMatT   \tilde{\cvector }$ lies only partially within the span of the 
orthogonal basis, then 
$ || \widetilde{\X }^T \sampleMatT   \tilde{\cvector } ||_2 < || \sampleMatT   \tilde{\cvector } ||_2$.
That is,  application of $\widetilde{\X }^T $ corresponds to a rotation of
the component of $\sampleMatT   \tilde{\cvector }$ lying within the span of the
orthogonal
basis. This implies that a necessary condition for a solution
to (\ref{lemma version}) is that 
\begin{equation} \label{necessary}
|| \tilde{\podstate}^T \tilde{\cvector } ||_2 \le || \sampleMatT   \tilde{\cvector } ||_2 .
\end{equation}
\vskip .1in
\noindent
\underline{Case 1: $\nsample  = \nstate $}
\newline 
\noindent
$\sampleMatT   \tilde{\cvector }$ must lie within the range of $\widetilde{\X } $ (as it is a full
rank square matrix) and so it is necessary to have equality in 
(\ref{necessary}) when $\nsample  = \nstate $.
One possible $\widetilde{\X }$ in this case is obtained by first
defining a $\Qmat _1 \in \RR{\nstate \times \nstate }$ and a 
$\Qmat _2 \in \RR{\nstate  \times \nstate }$ such that the first row of 
$\Qmat _1$ is $\alpha_1 \tilde{\cvector }^T \tilde{\podstate} $ with $\alpha_1 = 1/\|\tilde{\cvector }^T \tilde{\podstate}\|_2$.
Likewise, the first row of $\Qmat _2$ 
is taken as $\alpha_1\tilde{\cvector }^T \sampleMat $; note that $\alpha_1$ also normalizes this row
because we assume $\|\tilde{\podstate}^T\tilde{\cvector }\|_2 = \|\sampleMatT  \tilde{\cvector }\|_2$.
All remaining rows are 
chosen so that both $\Qmat _1$ and $\Qmat _2$ are orthogonal matrices.
This gives 
$$
\Qmat _1 \tilde\podstate^T \tilde{\cvector } = || \tilde\podstate^T \tilde{\cvector }  ||_2 \evector_1  = || \sampleMatT   \tilde{\cvector }  ||_2 \evector_1
= \Qmat _2 \sampleMatT   \tilde{\cvector },
$$
where $\evector_1$ is the first canonical unit vector of length $\nstate $ (first element is one and all other
$\nstate -1$ components are zero). A suitable $\widetilde{\X }$ that solves
(\ref{lemma version}) is then given by $\widetilde{\X }^T = \Qmat _1^T \Qmat _2$.
Thus, equality in (\ref{necessary}) is also sufficient when $\nsample  = \nstate $.
\vskip .1in
\noindent
\underline{Case 2: $\nsample  > \nstate $}
\newline 
\noindent
The matrix $\widetilde{\X }$ is now rectangular.  One possible $\widetilde{\X }$ is 
obtained by defining $\Qmat _1$ as before while instead defining an
$\nsample  \times \nsample $ orthogonal matrix $\Qfull$ with the first row
again set to $\alpha_2 \tilde{\cvector }^T \sampleMat  $ with $\alpha_2 = 1/\|\sampleMatT  \tilde \cvector \|_2$.
This gives
$$
\Qmat _1 \tilde\podstate^T \tilde{\cvector } = || \tilde\podstate^T \tilde{\cvector }  ||_2 \evector_1 
~~\mbox{and}~~
\Qfull \sampleMatT   \tilde{\cvector } = || \sampleMatT   \tilde{\cvector }  ||_2 \tilde{e}_1
$$
where $\tilde{e}_1 \in \RR{\nsample \times 1}$ is the first canonical 
unit vector of length $\nsample $.
If $|| \sampleMatT   \tilde{\cvector }  ||_2 = || \tilde\podstate^T \tilde{\cvector }  ||_2 $, then a suitable
$\tilde{\podstate}$ solving (\ref{lemma version}) is given by 
taking $\Qmat _2$ to be the first $\nstate $ rows of $\Qfull$
(as $\Qmat _1 \tilde\podstate^T \tilde{\cvector } = \Qmat _{2} \sampleMatT   \tilde{\cvector }$) and taking 
$\widetilde{\X }^T = \Qmat _1^T \Qmat _2$.
If $|| \sampleMatT   \tilde{\cvector }  ||_2 > || \tilde\podstate^T \tilde{\cvector }  ||_2 $, then 
we define a vector $y$ as an arbitrary linear combination of
the last $\nsample  - \nstate $ rows of $\Qfull$ such that $y$ has unit
norm. The first row of $\Qmat _2$ is then 
taken as 
$$
(\Qmat _2)_1 = \alpha_3 \frac{\sampleMatT   \tilde{\cvector }}{|| \sampleMatT   \tilde{\cvector } ||_2 } +
\sqrt{1 - \alpha_3^2}~y,
$$
where 
$(\Qmat )_k$ denotes the $k$th row of a matrix $\Qmat $
and $\alpha_3 = || \tilde{\podstate}^T \tilde{\cvector } ||_2 / 
             || \sampleMatT    \tilde{\cvector } ||_2 $.
The remaining rows of $\Qmat _2$ are simply $(\Qfull)_k$ for 
$k = 2 , \ldots, \nstate $.
It is easy to verify that 
$\Qmat _2$ is again orthogonal and that 
$\Qmat _{2} \sampleMatT   \tilde{\cvector } = || \tilde\podstate^T \tilde{\cvector }  ||_2 \evector_1 $. Thus,
$\Qmat _{2} \sampleMatT   \tilde{\cvector } = || \tilde\podstate^T \tilde{\cvector }  ||_2 \evector_1 $ and
$\widetilde{\X }^T = \Qmat _1^T \Qmat _2$ is a possible solution  implying
that $|| \sampleMatT   \tilde{\cvector }  ||_2 \ge || \tilde\podstate^T \tilde{\cvector }  ||_2 $
is a necessary and sufficient condition when $\nsample  > \nstate $.
\end{proof}

\section*{Acknowledgments}
The authors acknowledge Julien Cortial for both insightful discussions and for
providing the original nonlinear-truss code that was modified to generate the
numerical results. The authors also acknowledge Clancey Rowley for useful
comments received at the 2013 SIAM Conference on Computational Science and
Engineering.

This research was supported in part by an appointment to the Sandia National
Laboratories Truman Fellowship in National Security Science and Engineering,
sponsored by Sandia Corporation (a wholly owned subsidiary of Lockheed Martin
Corporation) as Operator of Sandia National Laboratories under its U.S.
Department of Energy Contract No. DE-AC04-94AL85000.  The authors also
acknowledge support by the Department of Energy Office of Advanced
Scientific Computing Research under contract 10-014804. 

\bibliography{references}
\bibliographystyle{siam}
\end{document}